%% file: sec-arxiv.tex
\documentclass[letterpaper,11pt]{article}
\usepackage[utf8]{inputenc}
\usepackage{graphicx,fullpage,paralist}
\usepackage{amsmath, amssymb, amsthm, mathtools}
\usepackage{comment,hyperref}
\usepackage{thmtools}
\usepackage{tikz}
\usepackage{pgfplots}
\usepackage{varwidth}
\pgfplotsset{compat=1.10}
\usepgfplotslibrary{fillbetween}
\usetikzlibrary{backgrounds}
\usetikzlibrary{patterns}
\newcommand{\PP}{\mathbb{P}}
\newcommand{\calR}{\mathcal{R}}

\newcommand{\calH}{\mathcal{H}}
\newcommand{\calM}{\mathcal{M}}
\newcommand{\calL}{\mathcal{L}}
\newcommand{\calP}{\mathcal{P}}
\newcommand{\calQ}{\mathcal{Q}}
\newcommand{\calF}{\mathcal{F}}

\newcommand{\calT}{\mathcal{T}}
\newcommand{\calE}{\mathcal{E}}
\newcommand{\RR}{\mathbb{R}}
\newcommand{\NN}{\mathbb{N}}
\newcommand{\ZZ}{\mathbb{Z}}
\newcommand{\FF}{\mathbb{F}}
\newcommand{\E}{\mathbb{E}}

\usepackage{amsbsy}
\DeclareMathOperator{\Bin}{Bin}
\DeclareMathOperator{\pre}{Left}
\DeclareMathOperator{\nex}{Right}

\newcommand{\ALG}{\mathrm{ALG}}

\newcommand{\OPT}{\mathrm{OPT}}
\newcommand{\calI}{\mathcal{I}}

\newtheorem{thm}{Theorem}
\newtheorem{lem}[thm]{Lemma}

\newtheorem{rmk}[thm]{Remark}

\theoremstyle{definition}
\newtheorem{defi}[thm]{Definition}

\usepackage{multirow}
\usepackage[noend]{algpseudocode}
\usepackage{algorithm,algorithmicx}

\newlength{\algofontsize}
\usepackage[textsize=tiny,textwidth=2cm]{todonotes}
\newcommand{\jscom}[1]{\todo[color=orange!25!white]{JS: #1}}

\usepackage{xcolor}
\hypersetup{
	colorlinks,
	linkcolor={red!50!black},
	citecolor={blue!50!black},
	urlcolor={blue!80!black}
}

\begin{document}
	
	\algrenewcommand\algorithmicrequire{\textbf{Input:}}
	\algrenewcommand\algorithmicensure{\textbf{Output:}}
	
	\title{\vspace{-2.5em}Strong Algorithms for the Ordinal Matroid Secretary Problem\thanks{This work was partially supported by FONDECYT project 11130266, Conicyt PCI PII 20150140 and N\'ucleo Milenio Informaci\'on y Coordinaci\'on en Redes ICM/FIC RC130003}}

	\author{
	Jos\'e A.~Soto\thanks{Departamento de Ingenier\'ia Matem\'atica \& CMM, Universidad de Chile. Email: {\tt jsoto@dim.uchile.cl}}
	\and
	Abner Turkieltaub\thanks{Departamento de Ingenier\'ia Matem\'atica, Universidad de Chile. Email: {\tt abnerturkieltaub@gmail.com}}
	\and 
	Victor Verdugo\thanks{Departamento de Ingenier\'ia Industrial, Universidad de Chile \& D\'epartment d'Informatique CNRS UMR 8548, \'Ecole normale sup\'erieure, PSL Research University. Email: {\tt vverdugo@dii.uchile.cl}}}
\date{\vspace{-3em}}
\maketitle
 \begin{abstract}
 \noindent In contrast with the standard and widely studied utility variant, in the ordinal Matroid Secretary Problem (MSP) candidates do not reveal numerical weights but the decision maker can still discern if a candidate is better than another. We consider three competitiveness measures for the ordinal MSP. An algorithm is $\alpha$ ordinal-competitive if for every weight function compatible with the ordinal information, the expected output weight  is at least $1/\alpha$ times that of the optimum; it is $\alpha$ intersection-competitive if its expected output includes at least $1/\alpha$ fraction of the elements of the optimum, and it is $\alpha$ probability-competitive if every element from the optimum appears with probability $1/\alpha$ in the output. This is the strongest notion as any $\alpha$ probability-competitive algorithm is also $\alpha$ intersection, ordinal and utility (standard)~competitive.
 
 Our main result is the introduction of a technique based on \emph{forbidden sets} to design algorithms with strong probability-competitive ratios on many matroid classes. In fact, we improve upon the guarantees for almost every matroid class considered in the MSP literature: we achieve probability-competitive ratios of $e$ for transversal matroids (matching Kesselheim et al.~\cite{KRTV2013}, but under a stronger notion); of $4$ for graphic matroids (improving on $2e$ by Korula and Pál~\cite{KP2009}); of $3\sqrt{3}\approx 5.19$ for laminar matroids (improving on 9.6 by Ma et al.~\cite{MTW2016}); and of $k^{k/(k-1)}$ for a superclass of $k$ column sparse matroids, improving on the $ke$ result by Soto \cite{Soto2013}. We also get constant ratios for hypergraphic matroids, for certain gammoids and for graph packing matroids that generalize matching matroids. The forbidden sets technique is inspired by the backward analysis of the classical secretary problem algorithm and by the analysis of the $e$-competitive algorithm for online weighted bipartite matching by Kesselheim et al.~\cite{KRTV2013}. Additionally, we modify Kleinberg's $1+O(\sqrt{1/\rho})$ utility-competitive algorithm for uniform matroids of rank $\rho$ in order to obtain a  $1+O(\sqrt{\log \rho/\rho})$ probability-competitive algorithm.
 Our second contribution are algorithms for the ordinal MSP on arbitrary matroids. We devise an $O(1)$ intersection-competitive algorithm, an $O(\log \rho)$ probability-competitive algorithm and an $O(\log\log \rho)$ ordinal-competitive algorithm for matroids of rank $\rho$. The last two results are based on the $O(\log\log \rho)$ utility-competitive algorithm by Feldman et al. \cite{FSZ2015}.
\end{abstract}

\section{Introduction}

In the classical secretary problem (see, e.g., \cite{Ferguson2010} for a survey) an employer wants to select exactly one out of $n$ secretaries arriving in random order. After each arrival, the employer learns the relative merits of the new candidate (i.e., he can compare the candidate with previous ones but no numerical quantity is revealed), and must reject or accept immediately. Lindley \cite{Lindley1961} and Dynkin \cite{Dynkin1963} show that the strategy of sampling $1/e$ fraction of the candidates and then selecting the first record has a probability of at least $1/e$ of selecting the best secretary and that no algorithm can beat this constant. During the last decade, generalizations of this problem have attracted the attention of researchers, specially due to applications in online auctions and online mechanism design. 

Arguably the most natural extension is the \emph{generalized secretary problem}  by Babaioff et al.~\cite{BIK2007}. In their setting, a set $R$ of candidates of known size is presented to an algorithm on uniform random order. On arrival, each element $r$ reveals its hidden weight $w(r)$ and the algorithm must irrevocably decide whether to select it or not while preserving the independence of the set of selected elements on a given independence system\footnote{An independence system is a pair $(R,\calI)$ where $R$ is finite, and $\calI$ is a nonempty family of subsets of $R$ that is closed under inclusion. The sets in $\calI$ are called independent sets.} $(R,\calI)$. 
%Depending on the assumptions, the algorithm can either check independence of subsets of elements it has already seen, via an independence oracle, or it may be able to check independence of arbitrary subsets whose elements may not have arrived yet (it has full access to $\calI$ upfront). 
The objective is to maximize the total weight of the selected independent set. Mainly motivated by applications and by the richness of their structural properties, Babaioff et al.~focuses on the case in which $(R, \calI)$ is a \emph{matroid}\footnote{A matroid is an independence system $(R, \calI)$ satisfying the next augmentation property: whenever $X, Y \in \calI$, and $|X|<|Y|$, there must be an element $r\in Y\setminus X$ such that $X+r\in \calI$.}, where this problem	takes the name of \emph{Matroid Secretary Problem} (MSP). An $\alpha$ competitive algorithm (which we call $\alpha$ utility-competitive algorithm to differentiate it from other measures) is one that returns an independent set 
%$\ALG$ 
whose expected weight 
%$\E[w(\ALG)]$ 
is at least $1/\alpha$ times the weight of an optimum independent set. Babaioff et al.~posts the well-known, and still open, \emph{Matroid Secretary Conjecture} which in its weak form states that there must be a constant competitive algorithm for the MSP on any matroid, and on its strong form, claims that the constant is $e$. 
%They also show that the {\it Threshold Price  Algorithm} (TPA) which use a constant fraction sample of the elements to carefully choose a random threshold $T$, and then greedily selects any element whose weight is at least $T$, as long as it preserves independence, is $O(\log \rho)$ utility-competitive for the MSP on matroids of rank $\rho$. 
They also provide a {\it Threshold Price  Algorithm} (TPA) achieving an $O(\log \rho)$ utility-competitive ratio on matroids of rank $\rho$. 	Better algorithms have improved this ratio to $O(\sqrt{\log \rho})$ by Chakraborty and Lachish \cite{CL2012} and later to the current best ratio of $O(\log\log \rho)$ by Lachish \cite{Lachish2014} and independently by Feldman et al.~\cite{FSZ2015}.  

The generalized secretary problem has also been studied on non-matroidal systems such as knapsack \cite{BIKK2007}, online matchings on graphs and hypergraphs \cite{Dimitrov2012, KP2009, KRTV2013} and LP packings \cite{KRTV2014}. Babaioff et al.~\cite{BIK2007} show that for general independence systems every algorithm on $n$ elements must be $\Omega(\log n /\log \log n)$ utility-competitive and Rubinstein \cite{Rubinstein2016} provides an $O(\log n \log\rho)$ competitive algorithm. Many works consider alternative objective functions, such as minimum sum of ranks \cite{AMW2001}, time discounted variants \cite{BDGIT2009}, convex costs under knapsack type restrictions \cite{BUCM12} and submodular objectives \cite{BHZ2013, FNS11, FZ15}. Another rich line of research focuses on relaxing the random order condition and the power of the adversary setting the weights. There are $O(1)$ competitive algorithms for the MSP on the random order and assignment model \cite{Soto2013}, on the adversarial order random assignment model \cite{OV2013, Soto2013} and on the free order model \cite{JSZ2013}. Non-uniform arrival orders have also been considered \cite{KKN2015}. Related to this line of research is the design of order-oblivious algorithms, which can sample a constant fraction of the elements but afterwards, the arrival order may be arbitrary; and results related to prophet inequality settings on which the weight of each element is a random variable with known distribution but the arrival order is adversarial, see \cite{KW12, AKW14, DK15, RubinsteinS17}.

\subsection{Ordinal MSP versus Utility MSP}
A common requirement for the mentioned algorithms for the MSP on general matroids is that they require to use numerical weights: they all treat elements with similar weights more or less equivalently. In contrast, in the definition of the classical secretary problem, the decisions made by an algorithm only rely on the ordinal preferences between candidates. This is a very natural condition as for many applications it is difficult to determine (for instance, via an interview) a numerical weight representing each candidate, but it is often simple to compare any two of them. For this reason we restrict our study to the \emph{ordinal MSP} in which the decision maker can only compare seen elements. This is, it can check which element is higher in an underlying hidden total order $\succ$. In particular, the algorithm cannot really compute the \emph{weight} of a set. However, matroids have the following extremely nice feature: it is possible to find a maximum weight independent set using only ordinal information by greedily constructing an independent set using the order $\succ$. The set $\OPT$ found this way is the unique lexicographically\footnote{A set $A=\{a_1,\dots, a_k\}\subseteq R$ is lexicographically larger than a set $B=\{b_1,\dots, b_k\}\subseteq R$ of the same size, with respect to the order $\succ$ on $A\cup B$ if $a_i\succ b_i$ for the first $i$ on which $a_i$ and $b_i$ differ.} maximum independent of the matroid.
 So it is quite natural to ask what can be done without numerical weights in the secretary setting. In order to measure the performance of an algorithm we introduce three notions of competitiveness for the ordinal MSP. We say that an algorithm is $\alpha$ \emph{ordinal-competitive} if for every nonnegative weight function compatible with the total order, the expected weight of the output independent set $\ALG$ is at least $1/\alpha$ the weight of $\OPT$. It is $\alpha$ \emph{intersection-competitive} if the expected fraction of elements in $\OPT$ that is included into $\ALG$ is at least $1/\alpha$, and it is $\alpha$ \emph{probability-competitive} if for every element $r\in \OPT$, the probability that $r$ is included into $\ALG$ is at least $1/\alpha$. It is not hard to see that any $\alpha$ ordinal-competitive algorithm is also $\alpha$ competitive in the classical (utility) sense and also that any $\alpha$ probability-competitive algorithm is also $\alpha$ competitive in every other setting. %We explore the relations among these performance metrics, and provide algorithms achieving good guarantees.
 
All the mentioned algorithms  \cite{BIK2007, CL2012, Lachish2014, FSZ2015} for the standard or \emph{utility MSP} do not work in the ordinal model, and so the existence of an $\alpha$-utility competitive algorithm does not imply the existence of an $\alpha$ probability/intersection/ordinal competitive algorithm. 
%As an important example, we later show that the $O(\log \rho)$-competitive TPA by Babaioff et al.~may have an arbitrarily bad competitive ratio for the other three measures (see Section \ref{sec:incomparable}). 
To the authors' knowledge the only known algorithms for ordinal MSP on general matroids (apart from the trivial $O(\rho)$ probability-competitive that selects the top independent singleton using a classical secretary algorithm) are Bateni et al.'s $O(\log^2 \rho)$ ordinal-competitive algorithm \cite{BHZ2013} (which also works for submodular objective functions) and Soto's variant of TPA \cite{Soto2013} which is $O(\log \rho)$ ordinal-competitive.

Independently and simultaneously from this work, Hoefer and Kodric \cite{HoeferK2017} study the ordinal secretary problem obtaining constant ordinal-competitive algorithms for bipartite matching, general packing LP and independent sets with bounded local independence number. They also show that Feldman and Zenklusen's reduction \cite{FZ15} from submodular to linear MSP works in the ordinal model.

The strong matroid secretary conjecture may hold even for the probability-competitive  notion. We note, however, that for some matroid classes the notion of competitiveness matters. In the $J$-choice $K$-best secretary problem by Buchbinder et al.~\cite{Buchbinder2014}, one has to select $J$ elements from a randomly ordered stream getting profit only from the top $K$ elements. In particular, every $\alpha$ competitive algorithm for the $\rho$-choice $\rho$-best secretary problem that only uses ordinal information, is an $\alpha$ intersection-competitive algorithm for uniform matroids of rank $\rho$. Chan et al.~\cite{Chan2015} show that the best competitive ratio achievable by an algorithm that can only use ordinal information on the $2$-choice $2$-best secretary problem is approximately $0.488628$. In contrast, by using numerical information one can achieve a higher competitive ratio of $0.492006$ that also works for the sum-of-weights objective function (i.e., for utility-competitiveness). 
%Since any $\alpha$ probability competitive algorithm is also $\alpha$ intersection-competitive, we conclude that there is a gap between the optimal probability-competitive algorithm and the optimal utility-competitive algorithm for uniform matroids of rank~2. 
Thus showing a gap between the optimal probability ratio and the optimal utility ratio achievable on uniform matroids of rank 2.

\subsection{MSP on specific matroids}

An extensive amount of work has been done on the last decade on the MSP on specific matroid classes including unitary \cite{Lindley1961,Dynkin1963,GM1966,Beckmann1990,Bruss2000}, uniform \cite{BIKK2007,Kleinberg2005}, transversal \cite{BIK2007,Dimitrov2012,KP2009}, graphic \cite{BIK2007, BDGIT2009,KP2009}, cographic \cite{Soto2013}, regular and MFMC \cite{DK2014}, $k$-column sparse \cite{Soto2013} and laminar matroids \cite{IW2011,JSZ2013,MTW2016}. Even though the algorithms are stated for the utility MSP, many of the proofs work directly on either the ordinal MSP under ordinal or even probability competitiveness. We include in Table \ref{tab:resumen} a summary of all known results together with the improved bounds obtained in this paper.% We delay the definition of each particular matroid class to the main body of the paper.

\subsection{Our results and techniques}
We first formalize our new competitiveness notions for ordinal MSP and study their interrelations. We say that a performance notion is stronger than another if any algorithm that is $\alpha$ competitive for the former is also $\alpha$ competitive for the latter. Under this definition, we show that probability is stronger than ordinal and intersection, and that ordinal is stronger than utility. 

On the first part of our article we focus on a powerful technique to define strong algorithms for the MSP on many classes of matroids. Informally (see the exact definition in Section \ref{sec:specific}), we say that an algorithm has forbidden set of size $k$ if it samples without selecting $s\sim \Bin(n,p)$ elements and the following condition holds. Suppose that an element $r^*$ of $\OPT$ arrives in a position $t>s$ and let $R_t$ be the set of elements that arrived on or before time $t$. For each time step $i$ between $s+1$ and $t-1$ there is a  random set $\calF_i$ of at most $k$ forbidden elements such that if for every $i$, the element arriving at time $i$ is not forbidden ($i\not\in \calF_i$) then $r^*$ is sure to be selected. The following is the key lemma that shows why algorithms with small forbidden sets are useful for the MSP.

\begin{restatable}[Key Lemma]{lem}{key}
	\label{lem:keylemma}
	By setting the right sampling probability $p=p(k)$, every algorithm with forbidden sets of size $k$ is $\alpha(k)$ probability competitive, where
	$$(p(k),\; \alpha(k)) = \begin{dcases*} (1/e, \; e)& if $k=1$,\\
	(k^{-\frac{1}{k-1}},\;  k^{\frac{k}{k-1}})& if $k\geq 2$.
	\end{dcases*}$$
\end{restatable}

\begin{table}[t!] \small
	\begin{tabular}{|l|p{0.8em}l|c|c|c|}
		\hline
		\textbf{
			Matroid Class          } &   \multicolumn{2}{c|}{\textbf{Previous guarantees}                 }    & \multicolumn{3}{c|}{\textbf{   New algorithms}             } \\ \cline{4-6}
		&               \multicolumn{2}{c|}{\textbf{(u,o,p competitive)}  }                &        \textbf{ p-guarantee}         & \textbf{ref.} &  \textbf{forb.~size}  \\ \noalign{\hrule height 1.5pt}
		Transversal                                & o:   & $16$ \cite{Dimitrov2012}, $8$ \cite{KP2009}, $e$ \cite{KRTV2013} &              $e$               & Alg.~\ref{alg:transversal} &        1        \\ \hline 
		$\mu$ exch.~gammoids                       & u:   & $O(\mu^2)$ \cite{KP2009}, $e\mu$ \cite{KRTV2013}                 &      $\mu^{\mu/(\mu-1)}$       & Alg.~\ref{alg:gammoid} &      $\mu$      \\ 
		(type of hypermatching)                    &      &                                                                  &                                &  &\\ \hline
		Matching matroids                          &      & -                                                                &              $4$               & Alg.~\ref{alg:packing}  &        2        \\ \hline
		$\mu$ exch.~matroidal packings             &      & -                                                                &      $\mu^{\mu/(\mu-1)}$       &           Alg.~\ref{alg:packing} &      $\mu$      \\ \hline
		Graphic                                    & o:   & 16 \cite{BIK2007}, $3e$ \cite{BDGIT2009}, $2e$ \cite{KP2009}     &               4                & Alg,~\ref{alg:graphic}  &        2        \\ \hline
		Hypergraphic                                    &    & -     &               4                & Alg,~\ref{alg:graphic}  &        2        \\ \hline
		$k$-sparse matroids                        & o:   & $ke$ \cite{Soto2013}                                          &         $k^{k/(k-1)}$          & Alg.~\ref{alg:framed}  &       $k$       \\ \hline
		$k$-framed matroids                        &      & -                                                                &         $k^{k/(k-1)}$          & Alg.~\ref{alg:framed}  &       $k$       \\ \hline
		Semiplanar gammoids                        &      & -                                                                &           $4^{4/3}$            & Alg.~\ref{alg:semiplanar} &        4        \\ \hline
		Laminar                                    & o:   & $177.77$ \cite{IW2011}, $3\sqrt{3}e\approx 14.12$ \cite{JSZ2013}  &  $3\sqrt{3} \approx 5.19615$   & Alg.~\ref{alg:laminar}  &        3        \\ 
		& p:   & 9.6 \cite{MTW2016}                                               &                                &           &  \\ \hline
		Uniform $U(n,\rho)$                        & p:   & $e$ \cite{BIKK2007}                                           & $1+O(\sqrt{\log \rho / \rho})$ &           Alg.~\ref{alg:uniform}&    -   \\ 
		& o: & $1+O(\sqrt{1 / \rho})$                                           &                                &           &  \\ \hline
		Cographic                                  & p:   & $3e$ \cite{Soto2013}                                          &               -                &     -     &        -        \\ \hline
		Regular, MFMC                              & o:   & $9e$ \cite{DK2014}                                               &               -                &     -     &        -        \\ \hline
	\end{tabular}
	
	\caption{\label{tab:resumen}\small State of the art competitive ratios for all known matroid classes, including our results.}
	
\end{table}

Obtaining algorithms with small forbidden sets is simple for many matroid classes. In fact, it is easy to see that the variant of the standard classical secretary algorithm which samples $s\sim \Bin(n,p)$ elements and then selects the first element better than the best sampled element, has forbidden sets of size 1. Suppose that the maximum element $r^*$ arrives at time $t>s$ and denote by $x$ the second best element among those arrived up to time $t$. If $x$ does not arrive at any time between $s+1$ and $t-1$ then for sure, $x$ will be used as threshold and thus $r^*$ will be selected. In the above notation, all forbidden sets $\calF_i$ are equal to the singleton $\{x\}$. Using the key lemma, by setting $p=1/e$, this algorithm is $e$ competitive.  We provide new algorithms that beat the state of the art guarantees for transversal, graphic, $k$-sparse and laminar matroids. We also provide new algorithms for other classes of matroids such as matching matroids, certain matroidal graph packings (which generalize matching matroids), hypergraphic matroids, $k$-framed matroids (which generalize $k$-sparse matroids), semiplanar gammoids and low exchangeability gammoids. As an interesting side result, we  revisit Kleinberg's $1+O(\sqrt{1/\rho})$ ordinal-competitive algorithm. We show that its probability-competitiveness is bounded away from 1 and propose an algorithm that achieves a probability-competitive ratio of $1+O(\sqrt{\log \rho / \rho})$. Our new results for specific classes of matroids are summarized on Table \ref{tab:resumen}, which, for completeness includes all matroid classes ever studied on the MSP, even those for which we couldn't improve the state of the art. In the references within the table u, o and p stand for utility, ordinal and probability competitiveness respectively.

On the second part of this paper we obtain results for the ordinal MSP on general matroids.

\begin{thm}\label{thm:intersection}
	There exists a $\ln(2/e)$ intersection-competitive algorithm for the MSP.
\end{thm}

\begin{thm}\label{thm:ord-prob}
	There exist an $O(\log \log \rho)$ ordinal-competitive algorithm  and an $O(\log \rho)$ probability competitive algorithm for the MSP.
\end{thm}

Even though Theorem \ref{thm:intersection} is attained by a very simple algorithm, we note that standard ideas such as thresholding do not work for the intersection notion since elements outside $\OPT$ that are higher than actual elements from $\OPT$ do not contribute to the objective. 
The algorithms mentioned in Theorem \ref{thm:ord-prob} are based on the recent $O(\log\log \rho)$ utility-competitive algorithm by Feldman et al.~\cite{FSZ2015}. 
Their algorithm samples a fraction of the elements so to classify most of the non-sampled ones into $h=O(\log \rho)$ weight classes consisting on elements whose weights are off by at most a factor of 2. 
They implement a clever random strategy to group consecutive weight classes together into buckets, each one containing roughly the same random number of weight classes. 
On each bucket they define a single random matroid with the property that if one picks an independent set from each one of these matroids, their union is independent in the original one. 
Their algorithm then selects a greedy independent set on each bucket matroid and output their union. 
Due to the random bucketing, the expected number of elements selected on each weight class is at least $\Omega(1/\log h)$ times the number of elements that $\OPT$ selects from the same class.	 
This means that the global algorithm is actually $O(\log h)=O(\log\log\rho)$ utility-competitive. 

In the ordinal setting we cannot implement this idea in the same way, since basically we do not have any weights. 
However, we can still partition the elements of the matroid into \emph{ordered layers}.
The idea is to select a collection of thresholds obtained from the optimum of a sample, and use them as separators to induce the layers. 
This motivates the definition of the Layered-MSP (see Section \ref{subsec:reduction-layered-ordinal}).
For both ordinal and probability competitiveness we provide a reduction from the Layered-MSP to the ordinal MSP.
For the ordinal notion, we select as thresholds a geometrically decreasing subset of the sample optimum according to the value order, so we can partition the matroid into $h=O(\log \rho)$ layers. For the probability notion we use all the elements of the sample optimum as thresholds.
The crucial result in this part is that our reduction allows to go from any $g(h)$ competitive algorithm for the Layered-MSP to a  $g(O(1+\log h))$ ordinal-competitive algorithm, and to a $g(O(h))$ probability-competitive algorithm. 
In particular, by applying Feldman et al.'s algorithm, interpreting these layers as weight classes, we get an $O(\log\log \rho)$ ordinal-competitive algorithm and an $O(\log \rho)$ probability-competitive algorithm for the original matroid.

\subsection{Organization}
In Section \ref{sec:preliminaries} we fix some notation and formally describe the performance guarantees for the ordinal MSP, studying their relations. In Section \ref{sec:specific} we prove our key lemma for algorithms with small forbidden sets. We then devise simple algorithms for all the matroid classes mentioned in Table \ref{tab:resumen}. In Section \ref{sec:general} we describe our new algorithms for general matroids, and prove Theorems \ref{thm:intersection} and \ref{thm:ord-prob}. To keep the discussion clear and simple, we defer some of the proofs to the Appendix. 

\section{Preliminaries} \label{sec:preliminaries}

Let $\calM=(R,\calI,\succ)$ be a matroid with ground set $R=\{r^1,r^2,\ldots,r^n\}$, and $\succ$ a total order. We call $\succ$ the \emph{value order} and say that $r^1$ is the highest valued element, $r^2$ is the second one, and so on, then $r^1\succ r^2 \succ \dots \succ r^n$. By matroid properties, for every subset $Q\subseteq R$, there is a \emph{unique lexicographically optimum base\footnote{A base of a set $Q$ is a maximal independent subset $X\subseteq Q$.}} $\OPT(Q)$ obtained by applying the greedy algorithm in the order~$\succ$, over the set $Q$. We say that a nonnegative weight function $w\colon E\to \RR_+$ is compatible with the value order if $r^i \succ r^j \implies w(r^i)\ge w(r^j)$. Note that for every compatible weight function the set $\OPT=\OPT(E)$ is a maximum weight independent set. 
We  reserve the use of superscripts $k\in \NN$ on a set $Q$ to denote the subset of the highest $\min\{k,|Q|\}$ valued elements of $Q$. In particular $R^k$ and $\OPT^k$ denote the set of the top $k$ elements of the matroid and of $\OPT$ respectively. We also reserve $n$ and $\rho$ to denote the number of elements of $R$ and its rank respectively.

% Cambio... a ver si queda mejor: In the ordinal MSP, the elements of a totally ordered matroid are presented in \emph{uniform random order} to an online algorithm that does not know the value order of unrevealed elements. At any moment, the algorithm can compare (according to $\succ$) any pair of revealed elements. When a new element $r$ is presented, the algorithm must decide whether to add $r$ to the solution and this decision is irrevocable. The algorithm must guarantee that the set of selected elements is at all times independent\footnote{For particular classes of matroids, we may assume that $\calM$ is known beforehand by the algorithm, or alternatively that it is discovered by the algorithm via an independence oracle that allows it to test any subset of revealed elements. In any case, it is a standard assumption that the algorithm at least know the number of elements $n$ in the matroid.} in the matroid.  The objective of the algorithm is to return a set $\ALG$ as close as $\OPT$ as possible according to a given measure.

In the (utility/ordinal) MSP, the elements of a (nonnegatively weighted/totally ordered matroid) are presented in \emph{uniform random order} to an online algorithm that does not know a priori the (weights/value order) of unrevealed elements. At any moment, the algorithm can (view the weight of/compare in the total order) any pair of revealed element. When a new element $r$ is presented, the algorithm must decide whether to add $r$ to the solution and this decision is permanent. The algorithm must guarantee that the set of selected elements is at all times independent\footnote{For particular classes of matroids, we may assume that $\calM$ is known beforehand by the algorithm, or alternatively that it is discovered by the algorithm via an independence oracle that allows it to test any subset of revealed elements. In any case, it is a standard assumption that the algorithm at least know the number of elements $n$ in the matroid.} in the matroid.  The objective of the algorithm is to return a set $\ALG$ as close as $\OPT$ as possible according to certain competitiveness metric.

To make notation lighter, we use $+$ and $-$ for the union and difference of a set with a single element respectively. That is, $Q+r-e=(Q\cup \{r\})\setminus \{e\}$.
The {\it rank} of a set $S$ is the cardinality of its bases, $\rho(Q)=\max\{|I|:I\in \calI,\; I\subseteq Q\}$, and the {\it span} of $Q$ is $\text{span}(Q)=\{r\in R:\rho(Q+r)=\rho(Q)\}$. 
In matroids, $\OPT$ has the property of {\it improving} any subset in the following sense.

\begin{restatable}{lem}{lemoptrestricted}
\label{lem:opt-restricted}
	Let $Q\subseteq R$. Then, $\OPT\cap Q\subseteq \OPT(Q)$.
\end{restatable}

%\lemoptrestricted*
\begin{proof}[Proof of Lemma \ref{lem:opt-restricted}]
Let $r=r^k$ be an element of $\OPT\cap Q$. Since $r^k$ is selected by the Greedy algorithm we have that $r\not\in \text{span}(R^{k-1})$. But then $r\not\in \text{span}(R^{k-1}\cap Q)$ and so it is also selected by the Greedy algorithm applied only on the set $Q$. Therefore, $r\in \OPT(Q)$.
\end{proof}

Recall that the algorithm considered for the MSP have access to the elements in an online and uniformly at random fashion. We denote by $r_1$ the first element arriving, $r_2$ the second, and so on. In general, $R_t=\{r_1,r_2,\ldots,r_t\}$ is the set of elements arriving up to time $t$. \\

\noindent {\it Utility competitiveness. }
%The {\it utility variant} of the MSP is the standard and most studied version of the problem.  It was introduced by Babaioff et al. \cite{BIK2007} in 2007. In this variant we assume the algorithm actually learns the hidden weight $w(r)\in \RR_+$ of  element $r$ when it arrives, and we want to maximize  the total weight of the selected elements. For $\alpha\ge 1$, we say an algorithm returning a set $\ALG$ is $\alpha$ {\it utility-competitive} if
An algorithm for the utility MSP returning a set $\ALG$ is $\alpha\geq 1$ utility-competitive if
\begin{equation}
\E[w(\ALG)] \geq w(\OPT) / \alpha.\vphantom{\Bigg|} \label{eq1}
\end{equation}

%\end{defi}

\subsection{Measures of competitiveness for the ordinal MSP}\label{sec:variants}

We introduce three measures of competitiveness or the ordinal MSP. Recall that the algorithm only learns ordinal information about the elements but it cannot access numerical weights. In this sense, they are closer to the classical secretary problem than the utility variant (for a discussion about this aspect in the original secretary problem, see \cite{Ferguson2010}). 
%In the {\it ordinal measure}, the objective function is the same as in the utility variant. 

Since the weight function remains completely hidden for the algorithm, the first measure of competitiveness we consider is the following. An algorithm is $\alpha$ {\it ordinal-competitive} if \emph{for every} weight function $w$ compatible with the value order, condition \eqref{eq1} holds.
%\begin{defi}[Ordinal Utility Variant (or simply, Ordinal Variant)]
%	 This variant has the same objective function as the utility variant. Since the weight function remains completely hidden for the algorithm, we have to change the definition of competitiveness. We say that an algorithm is $\alpha$ ordinal-competitive if \emph{for every} weight function $w$ compatible with the ranking, condition \eqref{eq1} holds. \end{defi} 
An equivalent characterization of competitiveness is obtained by the following lemma.

\begin{lem}\label{lem:ordinal-equivalence} 
An algorithm is $\alpha\geq 1$ ordinal-competitive if and only if for every $k\in [n]$,
	\begin{equation}
	\quad \E|\ALG\cap R^k| \geq \E|\OPT\cap R^k| / \alpha. \label{ordinal2}
	\end{equation}
%where $R^k$ is the set of top $k$ elements of $R$.
\end{lem}

\begin{proof}
	Consider an $\alpha$ ordinal-competitive algorithm returning a set $\ALG$ and let $k\in [n]$.  Define
	the weight function $w(r)=1$ if $r\in R^k$ and zero otherwise, which is compatible with the value order. Then, $\E[|\ALG\cap R^k|]=\E[w(\ALG)] \geq \frac{1}{\alpha}w(\OPT) = \frac{1}{\alpha} |\OPT\cap R^k|.$
	Now suppose that \eqref{ordinal2} holds. Then, for any compatible function $w$, and defining $w(r^{n+1}):=0$, 
	\begin{align*}
		\E[w(\ALG)] &=\sum_{k=1}^n (w(r^k)-w(r^{k+1}))\cdot \E[|\ALG \cap R^k|]\\
				&\geq   \sum_{k=1}^n (w(r^k)-w(r^{k+1}))\cdot \frac{1}{\alpha}\E[|\OPT \cap R^k|] = \frac{1}\alpha \E[w(\OPT)]. \qedhere
	\end{align*}
\end{proof}

%The previous characterization will be very helpful to prove that one of our algorithms is ordinal competitive.
In the second measure we consider, we want to make sure that every element of the optimum is part of the output with large probability. We say that an algorithm is $\alpha\geq 1$ {\it probability-competitive} if for every $e\in \OPT$,
\begin{equation}
\Pr(e\in \ALG) \geq 1/\alpha.
\end{equation}
Finally, in the third measure we want to maximize the number of elements in the optimum that the algorithm outputs. We say an algorithm is $\alpha\geq 1$ {\it intersection-competitive} if
\begin{equation}\E[|\OPT \cap \ALG|]\geq |\OPT| / \alpha.\end{equation}

%The two variants considered so far, ordinal and probability, are stronger than the classic utility variant. 	
%\begin{defi}[Intersection Variant] In this variant, we want to maximize the number of elements of $\OPT$ that the algorithm selects. We say an algorithm is $\alpha\geq 1$ intersection-competitive if
%\begin{equation}\E[|\OPT \cap \ALG|]\geq \frac{|\OPT|}{\alpha}.\end{equation}
%\vvcom{mencionar que es incomparable al resto de las variantes.}
%\end{defi}

\noindent \emph{Relation between variants.} The ordinal and probability measures are in fact stronger than the standard utility notion. That is, any $\alpha$ ordinal/probability-competitive algorithm yields to an $\alpha$ utility-competitive algorithm. Furthermore, the probability is the stronger of them all.
%We start observing that the ordinal utility variant is \emph{harder} than the utility variant and that the probability-competitive variant is the hardest of them all.\\
\begin{restatable}{lem}{variantrelations}
\label{lem:variant-relations}
	If an algorithm is $\alpha$ ordinal-competitive then it is $\alpha$ utility-competitive. If an algorithm is $\alpha$ probability-competitive then it is also $\alpha$ ordinal, utility and intersection-competitive.
\end{restatable}

%\variantrelations*
\begin{proof}[Proof of Lemma \ref{lem:variant-relations}]
	If an algorithm is $\alpha$ ordinal-competitive then by definition it is $\alpha$ utility-competitive. Now consider an $\alpha$ probability-competitive algorithm returning a set $\ALG$. For any $k\in [n]$, we have
	$\E[|\ALG\cap R^k|] \geq \E[|\ALG\cap \OPT\cap R^k|] = \sum_{r\in \OPT\cap R^k}\Pr(r\in \ALG)\geq \frac{|\OPT\cap R^k|}{\alpha}.$
	which means that the algorithm is $\alpha$ ordinal-competitive, and therefore, it is also $\alpha$ utility-competitive.
	%	Thanks to Lemma \ref{ordinal-utility}, the algorithm is also $\alpha$ utility-competitive. 
	To see that the algorithm is also $\alpha$ intersection-competitive we note that
	
	\[\E[|\ALG\cap \OPT|] =\sum_{e \in \OPT} \Pr(e\in \ALG) \geq \frac{|\OPT|}{\alpha}.\quad \qedhere\]
\end{proof}

%\begin{lem}
%	If an algorithm is $\alpha$ ordinal-competitive then it is $\alpha$ utility-competitive. \label{ordinal-utility}
%\end{lem}
%\begin{proof}
%	Direct from the definition.
%\end{proof}
%
%\begin{lem}\label{lem:probhardest}
%	If an algorithm is $\alpha$ probability-competitive then it is also $\alpha$ ordinal-competitive, $\alpha$ utility-competitive and $\alpha$ intersection-competitive and $\alpha$ .
%\end{lem}
%\begin{proof}
%	Indeed, for any $k$, we have
%	$$\E[|\ALG\cap E^k|] \geq \E[|\ALG\cap \OPT\cap E^k|] = \sum_{e\in \OPT\cap E^k}\Pr(e\in \ALG)\geq \frac{|\OPT\cap E^k|}{\alpha}.$$
%	which means that the algorithm is $\alpha$ ordinal-competitive. Thanks to Lemma \ref{ordinal-utility}, the algorithm is also $\alpha$ utility-competitive. To see that the algorithm is also $\alpha$ intersection-competitive we note that
%	
%	\[\E[|\ALG\cap \OPT|] =\sum_{e \in \OPT} \Pr(e\in \ALG) \geq \frac{|\OPT|}{\alpha}.\quad \qedhere\]
%\end{proof}

An algorithm for the utility variant, which recall is able to use the elements' weights, may not be adapted to the ordinal MSP, since in this model it can not use any weight information. In other words, the existence of an $\alpha$ utility-competitive for a matroid (or matroid class) does not imply the existence of an $\alpha$-competitive algorithm for any of the other three measures.
 It is  worth noting that the intersection-competitive measure is incomparable with the other measures. It is not hard to find fixed families of instances for which a given algorithm is almost 1 intersection-competitive but has unbounded utility/ordinal/probability-competitiveness. There are also examples achieving almost 1 ordinal-competitiveness but unbounded intersection-competitive ratio.

\section{Improved algorithms for specific matroids}\label{sec:specific}

In this section we describe a powerful \emph{forbidden sets} technique to analyze algorithms for the ordinal MSP. Thanks to this technique we devise algorithms for many matroid classes previously studied in the context of MSP and to other matroids that have not been studied in this context. Our results improve upon the best known competitive ratios for almost all studied matroid classes. The only classes in which we do not find an improvement are cographic matroids for which there is already a $3e$ probability-competitive algorithm \cite{Soto2013}, and regular-and-MFMC matroids for which there is a $9e$ ordinal-competitive algorithm \cite{DK2014}. At the end of this section we briefly study uniform matroids, for which Kleinberg's algorithm \cite{Kleinberg2005} yields a $1+O(\sqrt{1/\rho})$ ordinal-competitive guarantee. We show that the probability-competitiveness of Kleinberg's algorithm is at least 4/3 and we propose a variant that achieves a slightly weaker ratio of $1+O(\sqrt{\log \rho / \rho})$ in the probability notion.

\subsection{Forbidden sets technique}

The following is the key definition that allows us to devise constant probability-competitive algorithms for specific classes of matroids. Recall that $|R|=n$, and for $t\in [n]$, $r_t$ is the random element arriving at time $t$, and $R_t$ is the random set of elements arriving at or before time $t$.

\begin{defi} An algorithm has forbidden sets of size $k$ if it has the following properties.
	\begin{compactenum}
		\item \textbf{(Correctness)} The algorithm returns an independent set $\ALG$.
		\item \textbf{(Sampling property)} It chooses a sample size $s$ at random from $\Bin(n,p)$ for some fixed \emph{sampling probability} $p$, and it does not accept any element from the first $s$ arriving ones.
		\item \textbf{($k$-forbidden property)}\label{property:k-forbidden}  For every triple $(X,Y,r^*)$ with $Y \subseteq R$, $r^* \in \OPT(Y)$ and $X\subseteq Y-r^*$, one can define a set $\calF(X,Y,r^*)\subseteq X$ of at most $k$ \emph{forbidden elements} of $X$ such that the following condition holds. Let $t\geq s+1$ be a fixed time. If $r_t\in \OPT(R_t)$ and for every $j\in \{s+1,\dots, t-1\}$,  $r_j \not\in \calF(R_j,R_t,r_t)$ then $r_t$ is selected by the algorithm.
		\end{compactenum}
\end{defi}

To better understand the $k$-forbidden property suppose that a fixed element $r^*\in \OPT$ arrives at step $t\geq s+1$, that is $r_t=r^*$. Note that the set $R_{t-1}$ of elements arriving before $r_t$ is a random subset of size $t-1$ of $R-r_t$, and no matter the choice of $R_{t-1}$, $r^*$ is always part  of $\OPT(R_t)=\OPT(R_{t-1}+r_t)$.
%Conditional on the set $R_{t-1}$, $r_{t-1}$ is a uniform random element of $R_{t-1}$. 
Inductively, for $j=t-1$ down to $j=1$, once $R_{j}$ is specified, $r_{j}$ is a uniform random element of $R_j$, and $R_{j-1}$ is defined as $R_j-r_j$. Moreover, this choice is \emph{independent} of the previous random experiments (i.e., the choices of $\{r_{j+1},\dots, r_{t-1}, R_{t-1}\}$). The $k$-forbidden property (\ref{property:k-forbidden}.) says that if for every $j\in \{s+1, \dots, t-1\}$ element $r_j$ is not a forbidden element in $\calF(R_j, R_t, r_t)$ then $r^*$ is guaranteed to be selected by the algorithm.  Designing algorithms with small forbidden sets is the key to achieve constant competitiveness as our key Lemma (that we restate below) shows.

%To understand a little bit better the $k$-forbidden property fix an element $r^*\in \OPT$ and suppose that $r^*$ arrives at time $t\geq s+1$. Note that the set $R_{t-1}$ of elements arriving before $r_t$ is a random subset of size $t-1$ of $R\setminus r_t$. And no matter the choice of $R_{t-1}$, $r^*$ is always part  of $\OPT(R_t)=\OPT(R_{t-1}+r_t)$.
%Conditional on the set $R_{t-1}$, $r_{t-1}$ is a uniform random element of $R_{t-1}$. Inductively, for $j=t-1$ down to $j=1$, once $R_{j}$ is specified, $r_{j}$ is a uniform random element of $R_j$, and $R_{j-1}$ is defined as $R_j-r_j$. Moreover, this choice is \emph{independent} of the previous random experiments (i.e., the choices of $\{r_{j+1},\dots, r_{t-1}, R_{t-1}\}$). The $k$-forbidden property says that if for every $j\in \{s+1, \dots, t-1\}$ element $r_j$ is not inside the set of forbidden elements $\calF(R_j, R_t, r_t)$ (which happens with probability at least $1-|\calF(R_j,R_t,r_t)|/j \geq 1-k/j$) then $r_t$ is guaranteed to be selected by the algorithm.  Designing algorithms with small forbidden sets is key to achieve constant competitive ratios as the following lemma shows.
%
%Now we recall our Key Lemma \ref{lem:keylemma} and prove it.

\key*

\begin{proof}
	Fix an element $r^*$ from $\OPT$ and condition on the realization of $s\sim \Bin(n,p)$, on the time $t$ on which $r^*=r_t$ arrives and on the set $Y=R_t\ni r_t$ of the first $t$ elements arriving. Abbreviate $\PP_t(\cdot )=\Pr(\cdot | r_t=r^*, R_t=Y, s).$ By Lemma \ref{lem:opt-restricted},  $r^*\in \OPT(R_t)$ and  by the $k$-forbidden property, 
\begin{align*}
\PP_t(r^*\in \ALG) &\geq \PP_t\left(\text{For all } j \in \{s+1,\dots, t-1\},\, r_j \in R_j \setminus \calF(R_j,Y,r^*)\right)\\
&= \prod_{j=s+1}^{t-1} \Pr(r_j \in R_j\setminus \calF(R_j,Y,r^*)) \geq \prod_{j=s+1}^{t-1}\left(\frac{j-k}{j}\right)_+,
\end{align*}
where $x_+=\max\{0,x\}$. The equality above holds because of the independence of the random experiments defining iteratively $r_{t-1}$, $r_{t-2}$, down to $r_{s+1}$ as mentioned before the statement of this lemma.  
By removing the initial conditioning we get
\begin{align}
\Pr(r^*\in \ALG) &\geq \E_{s\sim \Bin(n,p)}\frac{1}{n}\sum_{t=s+1}^n \prod_{j=s+1}^{t-1}\left(1-\frac{k}{j}\right)_+. \label{rhs}
\end{align}

To compute the right hand side we use the following auxiliary process. Suppose that $n$  people participate in a game. Each player $x$ arrives at a time $\tau(x)$ chosen uniformly at random from the interval $[0,1]$. Each person arriving after time $p$ selects a subset of $k$ \emph{partners} from the set of people arriving before them, without knowing their actual arrival times (if less than $k$ people have arrived before her, then all of them are chosen as partners).  A player wins if she arrives after time $p$ and \text{every} one of her partners arrived before time $p$. Since the arrival times are equally distributed and the event that two people arrive at the same time has zero probability the arrival order is uniform among all possible permutations. Furthermore, the number of people arriving before time $p$ distributes as $\Bin(n,p)$. Using these facts, the probability that a given person $x$ wins is exactly the right hand side of \eqref{rhs}.
But we can also compute this probability using its arrival time $\tau(x)$ as
\begin{align*}
\int_{p}^1 \Pr(\text{ all partners of $x$ arrived before time $p$ } |\ \tau(x)=\tau)\ d\tau &\geq 
\int_{p}^1 (p/\tau)^k d\tau,
\end{align*}
which holds since the arrival time of each partner of $x$ is a uniform random variable in $[0,\tau]$, and conditioned on $\tau$, each partner arrives before time $p$ with probability $p/\tau$. Since $x$ may have less than $k$ partners, we don't necessarily have equality. We conclude that for every $r^*\in \OPT$,
\begin{align*}
\Pr(r^*\in \ALG) \geq \int_{p}^1(p/\tau)^k d\tau =\begin{dcases*}
-p\ln(p),& if $k=1$,\\
\frac{p-p^k}{k-1}, &if $k\geq 2$.
\end{dcases*}
\end{align*}

By optimizing the value of $p$ as a function of $k$, we obtain that the probability of $r^*\in \ALG$ is $1/\alpha(k)$, with $\alpha(k)$ as in the statement of the lemma and the probability achieving it is $p=p(k)$.\end{proof}

The idea of analyzing an algorithm as a series of stochastically independent experiments which defines the reverse arrival sequence appears very early in the history of the secretary problem. As mentioned in the introduction, one can prove that the algorithm for the classical secretary algorithm that samples $s\sim\Bin(n,p(1))$ elements and then selects the first element better than all the sampled ones is $e$ probability-competitive showing that it has forbidden sets of size 1. In our notation, for each $(X,Y,r^*)$ with $r^*$ the maximum element of $Y$ and $X\subseteq Y-r^*$, define the forbidden set $\calF(X,Y,r^*)$ as the singleton $\OPT(X)$. The $1$-forbidden condition states that if the element $r_t$ arriving at time $t$ is a record (it is in $\OPT(R_t)$), and if for every time $j\in \{s+1,\dots, t-1\}$, $r_j \not\in \OPT(R_j)$ (i.e., the $j$-th arriving element is not a record), then $r_t$ will be chosen by the algorithm. Since all forbidden sets have size at most 1, setting the sampling probability to be $p(1)=1/e$ guarantees the probability-competitive ratio of $\alpha(1)=e$.

Kesselheim et al.'s algorithm for online bipartite matching \cite{KRTV2013} uses a technique very similar to that in the previous lemma to compute the expected weight contribution of the $k$-th arriving vertex, which is enough to prove asymptotically $e+O(1/n)$ utility-competitiveness for online bipartite matchings. In fact, for transversal matroids, this yields to an $e+O(1/n)$ ordinal-competitive algorithm. But their analysis does not imply any guarantee on the probability notion. In the next section we show by using a different analysis, that their algorithm has forbidden sets of size 1 and so it is actually $e$ probability-competitive (always, not just asymptotically). 

It turns out that many classes of matroids behave very similarly to transversal matroids in the following sense. A set $X$ is independent on a given matroid if and only if each element $r\in X$ can be mapped to an object (e.g., an edge covering $r$, a path ending in $r$, a subgraph covering $r$, etc.) $\tilde{r}$ such that the set $\tilde{X}=\{\tilde{r}\colon r\in X\}$  satisfies a combinatorial property (e.g., a matching covering $X$, a collection of edge/node disjoint paths connecting $X$ with some source, a collection of disjoint subgraphs covering $X$). We call $\tilde{X}$ a \emph{witness} for the independence of $X$. A set $X$ may have multiple witnesses, but we will always assume that there is a \emph{canonical witness}, $\text{witness}(X)$, that can be computed by the algorithm and furthermore, the choice of the witness cannot depend on the set of elements seen so far nor on its arrival order. 
The fact that the witnesses do not depend on the arrival order makes them amenable to the analysis by the reverse arrival sequence analysis above, which in turn, will help us to devise algorithms with constant-size forbidden sets. 

\subsection{Transversal matroids and Gammoids}

\noindent{\it Transversal matroids.} Let $G=(L \cup R, F)$ be a bipartite graph with color classes $L$ and $R$, where the elements of $R$ are called the {\it terminals} of $G$. 
The transversal matroid $\calT[G]$ associated to $G$ is the matroid with ground set $R$ whose independent sets are those $X\subseteq R$ that can be covered by a matching in $G$.  
We call $G$ the {\it transversal presentation} of $\calT[G]$. \\

%which does not depend on the order in which we saw the elements of $X$.
%A matroid is transversal if it admits such presentation.  
\noindent{\it Gammoids.} Let $G=(V,E)$ be a digraph and two subsets $S,R\subseteq V$ called {\it sources} and {\it terminals} respectively, which do not need to be disjoint. 
The {\it gammoid} $\Gamma(G,S,R)$ is the matroid over 
the terminals where $X\subseteq R$ is independent if $X$ is linked to $S$, that is, if there are node-disjoint directed paths starting from $S$ and ending on each element of $X$. 
We say that $(G,S,R)$ is the gammoid presentation of the matroid. 
%A given matroid is a gammoid if it admits a gammoid presentation. 
Note that transversal matroids are gammoids by declaring all the non-terminals as sources and directing the arcs in the transversal presentation from sources to terminals. \\

\noindent{\it Transversal and Gammoid MSP.} In the \emph{Transversal MSP}, a transversal presentation $G$ for an unknown ordered matroid $\calM=\calT[G]$ is either revealed at the beginning of the process, or it is revealed online in the following way. Initially, the algorithm only knows the number of terminals. Terminals arrive in random order and whenever $r\in R$ arrives, it reveals its ordinal value information and the set of its neighbors in $L$. 
In the \emph{gammoid MSP}, a gammoid presentation $(G,S,R)$ for an unknown gammoid is either revealed at the beginning or it is revealed online as elements from $R$ arrive: when a terminal $r\in R$ arrives all possible $S$-$r$ paths are revealed. 
%More precisely, let $R_t$ denote the set of terminals revealed up to time $t$. 
At that time the algorithm has access to the subgraph  $G_t\subseteq G$ only containing the arcs belonging to every possible $S$-$R_t$ path and can test whether a vertex is in $S$ or not.

Most of the known algorithms for transversal MSP  \cite{BIK2007, Dimitrov2012, KP2009} work with ordinal information. %(i.e.\ without using numerical weights) 
The best algorithm so far, by Kesselheim et al.~\cite{KRTV2013} achieves an asymptotically optimal utility-competitive ratio of $e+O(1/n)$ for the more general (non-matroidal) \emph{vertex-at-a-time bipartite online matching problem}, in which edges incident to the same arriving vertex may have different weights and the objective is to select a matching of maximum total weight. For the specific case of transversal matroid this algorithm can be implemented in the ordinal model, meaning that is $e+O(1/n)$ ordinal-competitive. 
Interestingly, for the broader bipartite matchings case, Kesselheim's algorithm does not work in the ordinal model, but the previous algorithm by Korula and P\'al \cite{KP2009} does, achieving $8$ ordinal-competitiveness. A recent result by Hoefer and Kodric \cite{HoeferK2017} improves this factor to $2e$ ordinal-competiive for bipartite matchings.\\

\noindent{\it Exchangeability parameter for gammoids.} Let us define a parameter to control the competitiveness of our algorithm for the gammoid MSP. 
Let $X$ be an independent set and let $Q$ be a path linking a terminal $r\in R\setminus X$ outside $X$ to $S$. 
The \emph{exchangeability} $\mu$ of the presentation $(G,S,R)$ is the maximum number of paths in $\calP_X$ that $Q$ intersects. 
The intuition behind is the following: in order to include $Q$ into $\calP_X$ while keeping disjointness we have to remove or exchange at least $\mu$ paths from $\calP_X$. 
For instance, if we define the \emph{diameter} $d$ of the gammoid presentation as the maximum number of nodes in any source-terminal path, then $\mu$ is at most $d$. If furthermore the terminals are sinks, that is out-degree 0, then $\mu$ is at most $d-1$, since paths ending at different terminals cannot intersect on a terminal. 

\begin{rmk}\label{remark:trans-are-gam}
This is the case for transversal matroids: their diameter in the gammoid presentation is $2$ and their exchangeability is 1. For our results, we assume the algorithm also knows an upper bound $\mu$ for the exchangeability parameter, and in this case we call the problem {\it $\mu$-gammoid MSP} or bounded exchangeability gammoid MSP.
\end{rmk}
The $\mu$-gammoid MSP problem is a special case of the (non-matroidal) {\it hypergraph vertex-at-a-time matching} (HVM) with edges of size at most $\mu+1$ studied by Korula and P\'al \cite{KP2009} and later by Kesselheim et al.'s \cite{KRTV2013} online hypermatching problem (see the discussion in those papers for precise definitions). They achieve $O(\mu^2)$ and $e\mu$ utility-competitiveness respectively.  Below we propose an algorithm for the $\mu$-gammoid MSP that has forbidden sets of size $\mu$. Provided we know $\mu$ upfront we get an $\alpha(\mu)$ probability-competitive algorithm. Note that for $\mu\geq 2$, $\alpha(\mu)=\mu^{1+1/(\mu-1)} < e\mu$, so our guarantee strictly improves on that of previous algorithms for HVM and hypermatching, on the special case of $\mu$-gammoids.\\

\noindent{\it Our algorithms.} We use the convention that for every vertex $v$ covered by some matching $M$, $M(v)$  denotes the vertex matched with $v$ in $M$. Furthermore, for every independent set $X\subseteq R$, we select canonically a witness matching $M_X:=\text{witness}(X)$ that covers $X$. In the case of gammoids, for any set $\calP$ of node-disjoint paths linking some set $X$ to $S$, and for every $v\in X$, $\calP(v)$  denotes the unique path in $\calP$ linking 
%agregue node-disjoint arriba
$v$ to $S$. We also say that $\calP$ covers a vertex $u$ if $u$ is in the union of the vertices of all paths in $\calP$.  
Furthermore, for every independent set $X\subseteq R$, we canonically select a fixed collection of node-disjoint $S$-$X$ directed paths $\calP_X:=\text{witness}(X)$ linking $X$ to $S$, and we assume that this choice does not depend on the entire graph but only on the minimum subgraph containing all arcs in every $S$-$X$ path. We also recall that on step $i$, $R_i=\{r_1,\dots, r_i\}$ denotes the set of revealed terminals and $G_i$ denotes the subgraph of the presentation currently revealed.
\noindent \begin{minipage}[t]{8cm}
  \vspace{0pt}  
\begin{algorithm}[H]
	\small
	\begin{algorithmic}[1]
		\Require{Presentation of a transversal matroid $\calT[G]$ whose terminals arrive in random order.}
		\Statex $\triangleright${$M$ and $\ALG$ are the currently chosen matching and right vertices respectively.}
		\Statex 
		\State $\ALG\gets \emptyset$, $s\gets \Bin(n,p)$, $M \gets \emptyset$ 
		\For{$i=s+1$ to $n$}
%		\State Compute the matching $M_{\OPT(R_i)}$ covering $\OPT(R_i)$ in $G_i$.
		\If{$\vphantom{\big|}r_i \in \OPT(R_i)$ and $\ell_i:=M_{\OPT(R_i)}(r_i)$ is not covered by $M$ $\vphantom{\big|}$} 
		\State $\vphantom{\big|}\ALG\gets \ALG+r_i$, $M \gets M\cup \{\ell_ir_i\}$	
		\EndIf 
		\EndFor
		\State Return $\ALG$.
	\end{algorithmic}
	\caption{\small for transversal matroids. \label{alg:transversal}}
\end{algorithm}

\end{minipage}%
\hfill
\begin{minipage}[t]{8.1cm}
  \vspace{0pt}
\begin{algorithm}[H]
	\small
	\begin{algorithmic}[1]
		\Require{Presentation of a gammoid $\Gamma:=\Gamma(G,S,R)$ whose terminals arrive in random order.}
		\Statex $\triangleright${$\calP$ and $\ALG$ are the currently chosen collection of node-disjoint paths and terminals selected respectively.}
		\State $\ALG\gets \emptyset$, $s\gets \Bin(n,p)$, $\calP \gets \emptyset$. 
		\For{$i=s+1$ to $n$}
%		\State Compute the set $\calP_{\OPT(R_i)}$ of node-disjoint paths linking $\OPT(R_i)$ to $S$ in $G_i$.
		\If{$\vphantom{\big|}r_i \in \OPT(R_i)$ and  no vertex in the path $\calP_{\OPT(R_i)}(r_i)$ is covered by $\calP\vphantom{\big|}$} 
		\State $\vphantom{\big|}\ALG\gets \ALG+r_i$, $\calP \gets \calP \cup \{\calP_{\OPT(R_i)}(r_i)\}$	
		\EndIf 
		\EndFor
		\State Return $\ALG$.
	\end{algorithmic}
	\caption{\small for $\mu$-bounded gammoids. \label{alg:gammoid}}
\end{algorithm}
\end{minipage}
\vspace{.5cm}

The algorithms above can compute $\OPT(R_i)$ without knowing the terminal weights, by just using the greedy algorithm. 
This requires that one is able to check independence algorithmically in each case. Indeed, for the transversal MSP algorithm, a set $X\subseteq R_i$ is independent if and only if the maximum cardinality matching on $G_i[N_{G}(X)\cup X]$ has size $|X|$. In the case of gammoids, one can check if $X\subseteq R_i$ is independent, by a standard reduction to a flow problem on $G_i$.
\begin{thm}
	Algorithm \ref{alg:gammoid} has forbidden sets of size equal to the exchangeability $\mu$ of the gammoid presentation. If $\mu$ is known, we can set $p=p(\mu)$ to get an $\alpha(\mu)$ probability-competitive algorithm.
\end{thm}
\begin{proof}
%	The proof is very similar to that for transversal matroids. 
	By construction, the set $\calP$ contains node-disjoint paths covering $\ALG$ at every time step, hence the algorithm is correct. The sampling condition is also satisfied by design. Let $r^*\in \OPT(Y)$ where $Y$ is a fixed set of terminals of size $t\geq s+1$, and suppose that $R_t=Y$ and $r_t=r^*$. Note that $r_t$ is selected by the algorithm if all vertices in the path $\calP_{\OPT(R_t)}(r_t)$ are not covered by the collection $\calP$ prior to that iteration. 
	In other words, by defining the forbidden sets to be
	$$\calF(X,Y,r^*)=\{v\in \OPT(X)\colon \calP_{\OPT(X)}(v) \text{ intersects } \calP_{\OPT(Y)}(r^*)\},$$
	the element $r_t$ is selected if $r_j\not\in \calF(R_j,R_t,r_t)$ for all $j\in \{s+1,\dots,t-1\}$. By definition, each forbidden set has size at most $\mu$.	 
\end{proof}
Algorithm \ref{alg:transversal} is essentially Algorithm \ref{alg:gammoid} applied over the gammoid presentation of the transversal matroid $\calT[G]$. Together with Remark \ref{remark:trans-are-gam}, it follows the result for the transversal MSP. 
\begin{thm}
	Algorithm \ref{alg:transversal} has forbidden sets of size 1, and therefore, by choosing $p=1/e$, it is an (optimal) $e$ probability-competitive for transversal matroids.
\end{thm}
We remark that every constant utility-competitive algorithm so far for transversal MSP requires to learn a bipartite presentation online. It is an open problem to find constant competitive algorithms for transversal matroids that only access the matroid via an independence oracle.

\subsection{Matching Matroids and Matroidal Graph Packings}

\noindent{\it Matroidal graph packings.} Let $\calH$ be a finite set of graphs. An {\it $\calH$-packing} of a host graph $G=(V,E)$ is a collection $\calQ=\{H_i\}_{i=1\dots k}$ of node-disjoint subgraphs of $G$  such that each $H\in \calQ$ is isomorphic to some graph in $\calH$. 
A vertex in $G$ is said to be covered by $\calQ$ if it belongs to some graph of $\calQ$.  
Let $R\subseteq V$ be a set of vertices called terminals and consider the independence system $\calM(R,G,\calH)$ over $R$ whose independent sets are all $X\subseteq R$ for which there is an $\calH$-packing covering $X$ in $G$. 
%Analogously to previous sections, for every independent $X$ we select a canonical packing $\calQ_X:=\text{witness}(X)$ that does not depend on the arrival order.
We say that $\calH$ is \emph{matroidal} if $\calM(V(G),G,\calH)$ defines a matroid for every graph $G$.  
Note that in this case, if $\calM(R,G,\calH)$ is  the restriction of $\calM(V(G),G,\calH)$ to a subset $R$, then it is also a matroid. 
We call every such $\calM(R,G,\calH)$ a graph packing matroid. \\

\noindent{\it Matroidal families.} 
%There are many matroidal families known. 
For an extensive treatment see Loebl and Poljak \cite{Loebl1988}, Janata \cite{Janata2005} and the references therein. 
In the following examples, $K_n$ denotes the complete graph on $n$ vertices and $S_n$ denotes the star with $n$ legs. 
If $\calH=\{K_1\}$ then $\calM(V, G,\calH)$ is the free matroid over $V$ where all sets are independent.
If $\calH=\{K_2\}$ then $\calM(V,G,\calH)$ is the \emph{matching matroid} of $G$, whose independent sets are all vertex sets that can be covered by a matching.
If for some $k$ the family $\calH=\{S_1,S_2,\dots, S_k\}$ is a sequential sets of stars, then $\calM(V,G,\calH)$ is matroidal. It is, in fact, the matroid union of many matching matroids.
The family $\calM(V,G,\calH)$ is also matroidal if $\calH=\{K_2, H\}$ where $H$ is either  a factor-critical graph or a 1-propeller.\footnote{A factor-critical graph is one such that $H-x$ admits a perfect matching for all $x\in V(H)$. A 1-propeller is a graph having a leaf $r$, and a vertex $c$ such that for every $x\in V(H)-c$, $H-x$ admits a perfect matching.}
%\begin{compactenum}
%\item If $\calH=\{K_1\}$ then $\calM(V, G,\calH)$ is the free matroid over $V$ where all sets are independent.
%\item If $\calH=\{K_2\}$ then $\calM(V,G,\calH)$ is the \emph{matching matroid} of $G$, whose independent sets are all vertex sets that can be covered by a matching.
%\item If for some $k$ the family $\calH=\{S_1,S_2,\dots, S_k\}$ is a sequential sets of stars, then $\calM(V,G,\calH)$ is matroidal. It is, in fact, the matroid union of many matching matroids.
%\item $\calM(V,G,\calH)$ is also matroidal if $\calH=\{K_2, H\}$ where $H$ is either  a factor-critical graph or a 1-propeller.\footnote{ Here, a factor-critical graph is one such that $H\setminus \{x\}$ admits a perfect matching for all $x\in V(H)$. A 1-propeller is a graph with two special vertices: a leaf $r$, and a vertex $c$ such that for every $x\in V(H)\setminus\{c\}$, $H\setminus\{x\}$ admits a perfect matching.}
%\end{compactenum}
For most known matroidal classes there are polynomial time algorithms available to check independence. This is the case for all classes above. Observe that transversal matroids are instances of matching matroids restricted to one side of the bipartition of the  host graph.

As we did for gammoids, we also define an exchangeability parameter $\mu$ to control the competitiveness of our algorithm. Consider an $\calH$-packing $\calQ$ of $G$, and a subgraph $H\subseteq G$ from the class $\calH$ covering a terminal $r\in R$ that is not covered by $\calQ$. The exchangeability $\mu$ of $\calM(R,G,\calH)$ is the maximum number of terminals over all such $\calQ$ and $H$ that would become uncovered if we removed all graphs from $\calQ$ that intersect $H$, namely,
\[\mu:=\max\left\{\sum_{H'\in \calQ\colon V(H)\cap V(H')\neq \emptyset} |V(H')\cap R| \,\colon H \text{ covers a terminal not covered by $\calQ$}\right\}.\] 
This parameter may be  complicated to compute but there is a simple upper bound: let $h$ be the maximum number of vertices of a graph from $\calH$. Then the worst possible situation occurs when $\calQ$ contains only graphs of size $h$, and $H$ is also a graph of size $h$ intersecting every graph in $\calQ$ on exactly one vertex (different from $r$). In this case, all graphs from $\calQ$ must be removed, so the number of newly uncovered vertices becomes $h\cdot(h-1)$.  This means that $\mu\leq h(h-1)$. 

In the \emph{$\calH$-Packing MSP}, the algorithm receives a collection $\calH$ of graphs. A host graph $G=(V,E)$ with terminals $R\subseteq V$ is either revealed at the beginning or it is revealed online as elements from $R$ arrive: when a terminal $r\in R$ arrives, all possible edges that belong to a graph $H\subseteq G$ with $r\in V(H)$ with $H\in \calH$ (via isomorphism) are revealed. More precisely, let $R_t$ denote the set of terminals revealed up to time $t$. At that time the algorithm has access to the subgraph  $G_t\subseteq G$ induced by all vertices belonging to every possible subgraph $H\subseteq G$, with $H\in \calH$ that intersects $R_t$. The algorithm can also test whether a vertex is a terminal or not. We also assume that an upper bound $\mu$ for the exchangeability parameter is available, and in this case we call the problem bounded $\calH$-packing MSP. Analogously to previous sections, for every independent $X$ we select a canonical packing $\calQ_X:=\text{witness}(X)$ that does not depend on the arrival order.

In the description of the algorithm we use the convention that for every $\calH$-packing $\calQ$ of a set $X$, and for every $v\in V$ covered by $\calQ$, $\calQ(v)$ denotes the unique graph in $\calQ$ covering $v$.  We also recall that on step $i$, $R_i=\{r_1,\dots, r_i\}$ denotes the set of revealed terminals and $G_i$ denotes the subgraph of the presentation currently revealed. 

\begin{algorithm}[H]
	\small
	\begin{algorithmic}[1]
		\Require{A matroidal family $\calH$ and a host graph $G=(V,E)$ whose terminals $R\subseteq V$ arrive in random order.}
		\Statex \Comment{$\calQ$ and $\ALG$ are the currently chosen $\calH$-packing and terminals selected respectively.}
		\State $\ALG\gets \emptyset$, $s\gets \Bin(n,p)$, $\calQ \gets \emptyset$. 
		\For{$i=s+1$ to $n$}
		%		\State Compute the $\calH$-packing $\calQ_{\OPT(R_i)}$ covering $\OPT(R_i)$. 
		\If{$r_i \in \OPT(R_i)$ and $r_i$ is already covered by $\calQ$.}
		\State $\ALG\gets \ALG+r_i$.
		\ElsIf{$r_i \in \OPT(R_i)$ and $\calQ\cup \{\calQ_{\OPT(R_i)}(r_i)\}$ is an $\calH$-packing.} 
		%		\State $\ALG\gets \ALG+r_i$, $\calQ \gets \calQ \cup \{H_i(r_i)\}$	
		\State $\ALG\gets \ALG+r_i$, $\calQ\gets \calQ \cup \{\calQ_{\OPT(R_i)}(r_i)\}$.			
		\EndIf 
		\EndFor
		\State Return $\ALG$.
	\end{algorithmic}
	\caption{\small for $\mu$ bounded $\calH$-packing matroids. \label{alg:packing}}
\end{algorithm}

The algorithm can compute $\OPT(R_i)$ without knowing the terminal weights, by  applying the greedy algorithm for $\calM(R,G,\calH)$.

\begin{thm}
	Algorithm \ref{alg:packing} has forbidden sets of size equal to the exchangeability $\mu$ of the $\calH$-presentation. If $\mu$ is known beforehand, we can set $p=p(\mu)$ to obtain an $\alpha(\mu)$ probability-competitive algorithm for $\mu$-bounded graph packing matroids.
\end{thm}
\begin{proof}
	Correctness and the sampling condition for forbidden sets are satisfied by design. Now let $r^*\in \OPT(Y)$ where $Y$ is a fixed set of terminals of size $t\geq s+1$, and suppose that $R_t=Y$ and $r_t=r^*$. Terminal $r_t$ is selected by the algorithm if either $\calQ$ already covers it on arrival, or if all vertices in the graph $\calQ_{\OPT(R_t)}(r_t)$ were not covered by graphs in $\calQ$ prior to that iteration. 
	
	In any case by defining as forbidden sets
	$$\calF(X,Y,r^*)=\{v\in \OPT(X)\colon V(\calQ_{\OPT(X)}(v)) \text{ intersects } V(\calQ_{\OPT(Y)}(r^*))\setminus \{r^*\}\},$$
	we have that $r_t$ is selected if $r_j\not\in \calF(R_j,R_t,r_t)\,$ for all $j\in \{s+1,\dots,t-1\}$. By definition, each forbidden set  has size at most $\mu$.	 
\end{proof}

We remark that the competitiveness achievable by an algorithm heavily depends on the way the matroid is presented. For instance, consider a matching matroid $\calM$ with host graph $G=(V,E)$. Note that the exchangeability of this matroid is $\mu\leq 2(2-1)=2$.
If the graph is presented online then we can use the previous algorithm to obtain an $\alpha(2)=4$ probability-competitive algorithm. If on the other hand the graph is presented upfront, then we can use the fact that every matching matroid is transversal \cite{Edmonds1965} to construct a transversal presentation of $\calM$ and then apply our $e$ probability-competitive algorithm using that presentation.

\subsection{Graphic and Hypergraphic Matroids}

\noindent {\it Graphic and hypergraphic matroids.} The {\it graphic} matroid $\calM[G]=(R,\calI)$ associated to a graph $G=(V,R)$ is the one whose independent sets are all the subsets of edges $X\subseteq R$ such that $(V,X)$ is a forest. The {\it hypergraphic} matroid $\calM[G]$ associated to a hypergraph $G=(V,R)$, whose edges may be incident to any number of vertices (if all the edges have size 1 or 2 we are back in the graphic case), is the matroid over $R$ whose independent sets $X\subseteq R$ are those, for which one can canonically choose for every $r\in X$ an edge denoted by $\text{edge}(r,X)=u(r)v(r)$ in $K_V=(V,\binom{V}{2})$ with both endpoints in $r$ in such a way that all
$\text{edge}(r,X)$, for $r\in X$ are different, and the collection $\text{edge}(X)=\{\text{edge}(r,X)\colon r\in X\}$ is a forest \cite{Lorea1975}.

One can check that the hypergraphic matroid $\calM[G]=(R,\calI)$ is the matroid induced from the graphic matroid $\calM[K_V]$ via the bipartite graph $(R\cup \binom{V}{2}, \tilde{E})$ with $ef \in \tilde{E}$ if $f\subseteq e$. In other words, $X$ is independent in the hypergraphic matroid $\calM[G]$ if $\text{edge}(X)$ is independent in the graphic matroid $\calM[K_V]$. Moreover, if $G$ is already a graph, $\text{edge}(X)=X$ and $\text{edge}(r,X)=r$ for all $r\in X$.

In the \emph{graphic MSP}/\emph{hypergraphic MSP} we assume that the underlying graph/hypergraph $G$ of a matroid $\calM[G]$ is either revealed at the beginning or revealed online in the natural way: we learn edges as they arrive. Let $X\subseteq R$ be an independent set. By orienting each connected component of the forest $\text{edge}(X)$ from an arbitrary root, we obtain a \emph{canonical orientation} $\text{arc}(X)$ of $\text{edge}(X)$ (for convenience, we denote by $\text{arc}(e,X)$ the oriented version of $\text{edge}(e,X)$) with indegree $\deg^-_{\text{arc}(X)}(v)\leq 1$ for every vertex $v$. The converse is \emph{almost} true in the following sense. If $A$ is a set of arcs (maybe including loops) such that  $\deg^-_A(v)\leq 1$ for every vertex then the underlying graph is not necessarily a forest in $K_V$, but a pseudoforest:
every connected component contains at most 1 cycle, which is directed. In fact, the edge sets of pseudoforest of a given graph $J$ are exactly the independent set of the so called bicircular matroid of $J$. This matroid is transversal with presentation $H$, where $V(H)= V(J) \cup E(J)$ and $ve\in E(H)$ if and only if $e$ is incident to $v$. This is the starting point for our algorithm for graphic matroids. \\

\noindent {\it The algorithm.} The plan is to only consider edges that belong to the current optimum. Furthermore, if we select an edge, then we orient it and include it into an arc set $A$ with the property that each vertex has maximum in-degree 1. Instead of using a random orientation (as in the algorithms by Korula and P\'al \cite{KP2009} or Soto \cite{Soto2013}), at every step we use the canonical orientation of the current optimum forest. In order to avoid closing a cycle we also impose that an arc $(u,v)$ can not be added to $A$ if $\deg^-_A(u)=1$ or $\deg^-_A(v)=1$.  The same algorithm works on hypergraphic matroids if we replace each independent set $X$ on the hypergraphic matroid by its associated forest $\text{edge}(X)$.  We also recall that on step $i$, $R_i=\{r_1,\dots, r_i\}$ denotes the set of revealed edges. The algorithm is fully described below.

\begin{algorithm}[H]
	\small
	\begin{algorithmic}[1]
		\Require{A hypergraphic matroid $\calM[G]$ with underlying hypergraph $G=(V,R)$, whose edges arrive in random order.}
		
		\Statex \Comment{$\ALG$ and $A$ are the currently selected independent set and the orientation of its associated forest.}
		
		\State $\ALG\gets \emptyset$, $s\gets \Bin(n,p)$, $A \gets \emptyset$. 
		\For{$i=s+1$ to $n$}
		\If{$r_i\in \OPT(R_i)$}
		\State Let $a_i=(u_i,v_i)=\text{arc}(r_i,\OPT(R_i))$ be the canonical orientation of $\text{edge}(r_i,\OPT(R_i))$.
			\If{$\deg^-_A(u_i)=0=\deg^-_A(v_i)$}
		\State $\ALG\gets \ALG + r_i$, $A \gets A + a_i$
		\EndIf
		\EndIf
		\EndFor
		\State Return $\ALG$.
	\end{algorithmic}
	\caption{\small for graphic or hypergraphic matroids. \label{alg:graphic}}
\end{algorithm}

\begin{thm}
	Algorithm \ref{alg:graphic} has forbidden sets of size 2. By setting $p=p(2)=1/2$, we get an $\alpha(2)=4$ probability-competitive algorithm for both graphic and hypergraphic matroids.
\end{thm}

\begin{proof}
	We first prove that the edge set $\hat{A}$ obtained from $A$ by removing its orientation is acyclic. 	Suppose by contradiction that at the end of some step $i$, $A$ contains for the first time a set $C$ such that its unoriented version  $\hat{C}$ is an undirected cycle.  
	Since this is the first time a cycle appears, $r_i$ must be selected and $a_i=(u_i,v_i)$ must be contained in $C$. 
	Since $a_i$ is included in $A$, we know that after its inclusion, $\deg^-_{A}(v_i)=\deg^-_{C}(v_i)=1$ (before its inclusion, the indegree of $v_i$ was 0) and that $\deg^-_{A}(u_i)=\deg^-_{C}(u_i)=0$. But since $\hat{C}$ is a cycle the outdegree of $u_i$ is $\deg^+_{C}(u_i)=2-\deg^-_{C}(u_i)=2$. But then, there must be another vertex $x$ in $C$ with indegree 2. This cannot happen because at every moment the indegree of each vertex is at most 1.
	
	The proof above guarantees correctness of the algorithm: for the graphic case $\hat{A}=\ALG$ and for the hypergraphic case, each edge $r_i$ of $\ALG$ is mapped to $\text{edge}(r_i,\OPT(R_i))\in \hat{A}$ which form a forest. In both cases we conclude $\ALG$ is independent.
	
	Since the sampling condition is satisfied by design, we only need to prove the 2-forbidden condition. Let $r^*\in \OPT(Y)$ where $Y$ is an arbitrary set of $t\geq s+1$ edges, and suppose that $R_t=Y$ and $r_t=r^*$. The algorithm would then define an arc $a_t=(u_t,v_t)$ and will proceed to add $r_t$ to $\ALG$ provided that no arc $a_j=(u_j,v_j)$ considered before has head $v_j$ equal to $u_t$ or $v_t$. In other words, by defining
	\begin{align*}\calF(X,Y, r^*)&=\left\{f \in \OPT(X)\colon  \substack{\text{arc}(f,\OPT(X))\text{ is not oriented}\\\text{ towards any endpoint of } \text{edge}(r^*,\OPT(Y))}\right\}
	\end{align*}
	then $r_t$ is selected if $r_j\not\in \calF(R_j,R_t,r_t)$ for all $j\in \{s+1,\dots, t-1\}$. Moreover, since each arc set $\text{arc}(\OPT(X))=\{\text{arc}(f,\OPT(X))\colon f\in \OPT(X)\}$ has maximum indegree 1, there are at most 2 arcs in $\text{arc}(\OPT(X))$ oriented towards an endpoint of $\text{edge}(r^*,\OPT(Y))$. So each forbidden set has size at most~2.
\end{proof}

\subsection{Column Sparse Representable Matroids and Multi-Framed matroids}

\noindent {\it Column sparse matroids.} An interesting way to generalize graphic matroids is via their matrix representation. We say that a matroid $\calM$ is represented by a $m\times n$ matrix $M$ with coefficients in a field $\FF$ if we can bijectively map the elements of its ground set to the columns of $M$ in such a way that the independent sets of $\calM$ are in correspondence with the sets of columns that are linearly independent in $\FF^{ m}$. Graphic matroids are representable by their adjacency matrix interpreted in $GF(2)$. In fact they can be represented in any field. Matroids that have this property are called regular. Note that each graphic matroid is representable by a very sparse matroid: each column has only 2 non-zero elements. Following \cite{Soto2013} we say that a matroid is $k$ column sparse representable if it admits a representation whose columns have at most $k$ nonzero entries each. These matroids include many known classes such as graphic matroids ($k=2$), rigidity matroids \cite{Whiteley1996} on dimension $d$ ($k=2d$) and more generally matroids arising from rigidity theory from $d$-uniform hypergraphs. These matroids are called $(k,\ell)$-sparse matroids (where $0\leq \ell \leq kd-1$), and they are, using our notation, $kd$ column sparse \cite{Streinu2011}. Interesting cases include generic bar-joint framework rigidity in the plane \cite{Laman1970}, which are characterized by $(2,3)$-sparse matroids in dimension $d=2$ (they are $4$ column sparse) and generic body-bar framework rigidity in $\RR^d$ which are characterized by $(\binom{d+1}{2},\binom{d+1}{2})$-sparse matroids \cite{Tay1984} (they are $d\binom{d+1}{2}$ column sparse).

Let $\calM$ be a matroid with $k$ sparse representation $M$ and let $X$ be an independent set of columns. It is easy to show (see e.g., \cite{Soto2013}) that we can select one nonzero coordinate from each column in $X$ such that no two selected entries lie on the same row. In other words, each independent set in $\calM$ is also independent in the transversal matroid whose bipartite representations has color classes the rows and columns of $M$ and where a column $i$ is connected to a row $j$ if entry $M_{ij}$ is nonzero. Even though the converse is not true, we can use the intuition obtained from graphic matroids to extend the algorithm to $k$ column sparse representable matroids with only few changes. Instead of doing that, we are going to further generalize this class of matroids in a different direction.\\

\noindent {\it Multiframed matroids.} A matroid $\calM$ is called a frame matroid \cite{Zaslavsky1994} if it can be extended to a second matroid $\calM'$ (i.e. $\calM$ is a restriction of $\calM'$) which possesses a \emph{frame} $B$, that is, a base  such that every element of $\calM$ is spanned by at most 2 elements of $B$. For instance, take a graphic matroid $\calM[G]$ and consider the graph $H=(V(G)+v_0, E(G) \cup \{v_0v\colon v\in V(G)\}$ where $v_0$ is a new vertex. Then, $\calM[H]$ is an extension of $\calM[G]$ containing the star centered at $v_0$ as basis $B:=\delta_H(v_0)=\{v_0v\colon v\in V(G)\}$. Since every edge $uv$ in $G$ is spanned by the set $\{v_0u, v_0v\}$ of (at most) 2 elements, we conclude that $B$ is a frame for $\calM[G]$.  We define a new class of matroids, called multiframed of $k$-framed matroids, in a natural way. A matroid $\calM=(R,\calI)$ is a $k$-framed matroid if it admits an extension $\calM^B=(R',\calI)$ having a $k$-frame $B$, i.e., a base such that each element of $\calM$ is spanned by at most $k$ elements of $B$. Without loss of generality we assume that $B\cap R=\emptyset$ (by adding parallel elements) and $R'=B\cup R$. In other words, $\calM^B$ is obtained by adjoining the $k$-frame $B$ to the original matroid $\calM$. Observe that if $\calM$ is represented by a $k$ column sparse matrix $M$, then the matrix $[I|M]$ obtained by adjoining an identity (in the field $\FF$) represents an extension $\calM^B$ of $\calM$ where the columns of $I$ form a base $B$ such that each column in $M$ is spanned by at most $k$ elements from $B$ (exactly those elements associated to the $k$ nonzero rows of the column). This means that $k$-framed matroids generalizes $k$ column sparse matroids. This generalization is strict since there are frame matroids that are nonrepresentable.

We define the \emph{$k$-framed MSP} as the variant of the MSP in which the $k$-framed matroid $\calM$, and its extension $\calM^B$ is either fully known beforehand or we simply have access to $B$ and an independence oracle for $\calM^B$ (in the case of $k$ column sparse matroid it is enough to receive the columns of the representation in an online fashion).

We need some notation. For every $r$ in the ground set $R$, we define the set $C(B,r)=\{y\in B\colon B+r-y \text{ is independent}\}$ which is also the minimal subset of $B$ spanning $r$.
It is easy to see that $C(B,r)+r$ is the unique circuit in $B+r$, often called the \emph{fundamental circuit} of $r$ with respect to the base $B$ in $\calM^B$.
Define also for each $y\in B$, the set
$K(B,y)=\{r\in R\colon B+r-y \text{ is independent}\}$. It is easy to see that $K(B,y)+y$ is the unique cocircuit inside $R + y$ in the matroid $\calM^B$, often called the \emph{fundamental cocircuit} of $y$ with respect to the base $B$. Observe that by definition, $r\in K(B,y) \iff y\in C(B,r)$. Furthermore, by definition of $k$-framed matroids, $C(B,r)$ has at most $k$ elements. Before presenting the algorithm we need the following result.

\begin{lem} Let $X$ be an independent set of a $k$-framed matroid $\calM$ with $k$-frame $B$. There is a (canonical) injection $\pi_X\colon X\to B$ such that $B+x-\pi_X(x)$ is independent for all $x\in X$.
\end{lem}
\begin{proof}
	Extend $X$ to a base $X'$ of $\calM'$. By the strong basis exchange axiom there is a bijection $\pi\colon X'\to B$ such that $B+x-\pi(x)$ is a base for all $x\in X'$. The restriction of $\pi$ to $X$ yields the desired injection.
\end{proof}

\begin{algorithm}[H]
	\small
	\begin{algorithmic}[1]
		\Require{A $k$-frame matroid $\calM$,  with independence oracle access to $\calM^B$ and to $B$. The elements of $\calM$ arrive in random order.}
		
		\Statex \Comment{$\ALG$ is the set currently selected and $B'$ is the set of elements of the frame $B$ that have been marked.}
		
		\State $\ALG\gets \emptyset$, $s\gets \Bin(n,p)$, $B' \gets \emptyset$. 
		\For{$i=s+1$ to $n$}
		\If{$r_i\in \OPT(R_i)$ and $C(B,r_i)\cap B'=\emptyset$}
		\State $\ALG\gets \ALG + r_i$, $B' \gets B' + \pi_{OPT(R_i)}(r_i)$
		\EndIf
		\EndFor
		\State Return $\ALG$.
	\end{algorithmic}
	\caption{\small for $k$-framed matroids. \label{alg:framed}}
\end{algorithm}

\begin{thm}
	Algorithm \ref{alg:framed} has forbidden sets of size $k$. By setting $p=p(k)$, we get an $\alpha(k)$ probability-competitive algorithm for $k$-framed matroids.
\end{thm}

\begin{proof}
	Suppose that the algorithm is not correct, and let $Z=\{r_{i(1)},r_{i(2)},\dots, r_{i(\ell)}\}$ be a circuit in $\ALG$ with $s+1\leq i(1)\leq i(2) \leq \dots \leq i(\ell)$. When $r_{i(1)}$ arrived, $y:=\pi_{\OPT(R_{i(1)})}(r_{i(1)})\in C(B,r_{i(1)})$ was marked.	Furthermore, our algorithm guarantees that for every element $r_{i}\in\ALG$ with $i> i(1)$, $y\not\in C(B,r_i)$. In particular, $y \in C(B,r_{i(j)}) \iff j=1$, or equivalently,
	$K(B,y)\cap Z = \{r_{i(1)}\}$. But this implies that the circuit $Z$ and the cocircuit $K:=K(B,y)+y$ intersect only in one element, which cannot happen in a matroid. Therefore, the algorithm is correct.
	
	Since the sampling condition is satisfied by design, we only need to prove the $k$-forbidden condition. Let $r^*\in \OPT(Y)$ where $Y$ is an arbitrary set of $t\geq s+1$ elements, and suppose that $R_t=Y$ and $r_t=r^*$. The algorithm accepts $r^*$ if and only if no element of $C(B,r_t)$ was marked  before. Let $y\in C(B,r_t)$ be an arbitrary element. A sufficient condition for $y$ not being marked at step $j$ is that $\pi_{\OPT(R_j)}(r_j)\neq y$. Therefore, if we define the forbidden sets
		\begin{align*}\calF(X,Y,r^*)&=\{f \in \OPT(X)\colon \pi_{\OPT(X)}(f)\in C(B,r^*)\}
		\end{align*}
	it is clear that $r_t$ is selected if $r_j\not\in \calF(R_j,R_t,r^*)$ for all $j\in \{s+1,\dots,t-1\}$.  Moreover, since $\pi_{\OPT(X)}$ is injective we conclude that $|\calF(X,Y,r^*)|\leq |C(B,r^*)|\leq k$.
\end{proof}
	
\subsection{Laminar Matroids and Semiplanar Gammoids}

\noindent{\it Arc-capacitated Gammoids.} In this section we define special classes of gammoids amenable for our techniques. An {\it arc-capacitated gammoid} (ACG) $\calM[N]$ is defined by a directed network $N=(G,s,R,c)$ where $G=(V,E)$ is a digraph, $s\in V$ is a single source, $R\subseteq V\setminus \{s\}$ is a collection of terminals and $c\colon E\to \ZZ^+$ is a strictly positive integer capacity function on the arcs of $G$. The ground set of  $\calM[N]$ is $R$ and every set $J\subseteq R$ is independent if and only if there exist an $s$-$R$ flow on $N$ satisfying the capacity constraints on the arcs and where each $j\in J$ receives one unit of flow. Without loss of generality we assume that every terminal is reachable from $s$ via a directed path. It is easy to see that ACGs are equivalent to gammoids without loops, but they are more closely related to transportation or distribution applications.  

A {\it semiplanar gammoid} is an ACG whose digraph $G$ admits a \emph{semiplanar drawing}, which is a planar drawing where all terminals are on the $x$-axis, the source node is on the positive $y$-axis, and the rest of the graph (arcs and nodes) are strictly above the $x$-axis, in a way that they do not touch the infinite ray starting upwards from the source node.
It is easy to see that $G$ admits a semiplanar drawing if and only if it is planar and all the terminals and the source are contained in the same face (which could be a cycle, a tree or a pseudoforest). 

\noindent{\it Laminar Matroids.} An important example of semiplanar gammoids are \emph{laminar matroids}. A collection $\calL$ of non-empty subsets of a finite set of terminals $R$ is called laminar if\,  for every $ L, H\in \calL$, $L\cap H \in \{L,H,\emptyset\}$.
The laminar matroid $\calM[R, \calL, c]$, where $\calL$ is a laminar family over $R$ and $c\colon \calL \to \ZZ^+$ is a positive integer capacity function on the sets of $\calL$, is the matroid with ground set $R$ and whose independent sets are those $X\subseteq R$ for which the number of elements that $X$ contains in any set of the laminar family does not exceed its capacity, i.e. $|X\cap L|\leq c(L)$. Observe that, since $c$ is positive, every singleton is independent. Furthermore, without loss of generality we assume that $R\in \calL$ (with $c(R)$ equal to the rank of the matroid) and that all singletons are in $\calL$ (with $c(\{r\})=1$ for each $r\in R$). Adding those sets does not destroy laminarity. 

The Hasse diagram of the containment order of $\calL$ forms a tree $T'$ where every set in $\calL$ is a child of the smallest set in $\calL$ that strictly contains it. Note that the leaves of $T'$ are in bijection with the terminals and the root $v_R$ that represents $R$. Consider the directed graph $G$ obtained by adding an auxiliary vertex $s$ to $T'$ connected to $v_R$ and orienting the tree away from $s$. If we put $s$ on the upper semiaxis, we draw $G$ with all arcs pointing downwards, and we assign capacities to each arc equal to the capacity of the laminar set associated to the arc's head we obtain a semiplanar representation of $\calM[R,\calL,c]$. Furthermore, the terminals in $R$ appear in the $x$-axis in the left-to-right order induced by the drawing of the tree.\\ 

\noindent {\it Semiplanar drawings and neighbours.} In what follows we fix a semiplanar drawing $G$ of a semiplanar gammoid $\calM[N]$. We identify $R$ with the set $[n]$ by labeling the terminals from left to right as they appear in the $x$-axis. For convenience in the rest of the presentation, we add to $G$ two auxiliary nodes on the $x$-axis, $0$ and $n+1$, where 0 is located to the left of 1, and $n+1$ is located to the right of $n$, together with the  arcs $s0$ and $s(n+1)$. Observe that by our assumptions on the drawing, it is possible to add those arcs without destroying semiplanarity. 

For any set $J\subseteq [n]$ and any $y\in J\cup \{0,n+1\}$, we denote by $\pre_J(y)$ the closest element to $y$ in $J\cup \{0\}$ that is located strictly to its left (understanding $\pre_J(0)=0$). Similarly, we denote by $\nex_J(y)$ the first element in $J\cup \{n+1\}$ located to the right of $y$ (understanding $\nex_J(n+1)=n+1$). For $y\in [n]\setminus J$, we call $\pre_{J+y}(y)$ and $\nex_{J+y}(y)$ its left and right neighbors in $J$ (note that they maybe equal to 0 or $n+1$ and thus, they are not necessarily elements of $J$). \\
%%%INNECESARIO
%\begin{align*}
%\pre_J(y)&=\begin{cases}
%\max\{j\in J\colon j < y\} &\text{ if } \{j\in J\colon j < y\}\neq \emptyset\\
%\max J &\text{ otherwise,}
%\end{cases}\\
%\nex_J(y)&=\begin{cases}
%\min\{j\in J\colon y < j\} &\text{ if } \{j\in J\colon y < j\}\neq \emptyset\\
%\min J &\text{ otherwise}.
%\end{cases}
%\end{align*}

\noindent{\it Ancestors, tree-order and representatives.} In the future we will map each terminal outside an independent set $J$ to its neighbors in $J$. For the case in which $G$ is a tree (i.e., for laminar matroids) we want to consistently assign each element to just one of them. We do this as follows. For every pair of nodes $x, y$ in $G$, let $xGy$ be the unique undirected $x$-$y$ path in $G$. We say that $x$ is an {\it ancestor} of $y$ (and $y$ is a {\it descendant} of $x$) if $sGy$ contains $sGx$; in that case we denote $x\sqsupseteq y$ (and $y\sqsubseteq x$).
Note that $(V,\sqsubseteq)$ is a partial order, and in fact, it is a join-semilattice where $x\vee y$ is the lowest common ancestor of $x$ and $y$. Note that by our drawing choice, for any $i \in [j,j']\subseteq [0,n+1]$ we have $i\sqsubseteq j\vee j'$. In particular, if $[j,j']\subseteq [k,k']\subseteq [0,n+1]$ then $j\vee j' \sqsubseteq k\vee k'$.

%We write $x\sqsupseteq y$ (or equivalently $y\sqsubseteq x$) if $x$ is an ancestor of $y$ (or $y$ is a descendant of $x$), i.e. if $sGy$ contains $sGx$. We also write $x\sqsupsety$ (or equivalently, $y\sqsubset x$) if $x\sqsupseteq y$ and $x\neq y$. Note that $(V,\sqsubseteq)$ is a partial order, in fact it is a join-semilattice where $x\vee y$ is the lowest common ancestor of $x$ and $y$. Note that by our drawing choice, for any $i \in [j,j']\subseteq [0,n+1]$ we have $i\sqsubseteq j\vee j'$. In particular, if $[j,j']\subseteq [k,k']\subseteq [0,n+1]$ then $j\vee j' \leq k\vee k'$.
For every nonempty set $J\subseteq [n]$, and for every terminal $y\in [n]$, we define its \emph{representantive} $\pi_J(y)$ in $J$ such that
%If $y\in J$ then $\pi_J(y)=y$. Otherwise,  other case, we define:
$$\pi_J(y)=\begin{cases}
y &\text{if $y\in J$,}\\
\pre_{J+y}(y)&\text{if $y \in [n]\setminus J$ and } y\vee \pre_{J+y}(y) \sqsubset y\vee \nex_{J+y}(y),\\
\nex_{J+y}(y)&\text{if $y \in [n]\setminus J$ and } y\vee \pre_{J+y}(y) \sqsupseteq y\vee \nex_{J+y}(y).\end{cases}$$

This element is well defined since for all $y\in [n]\setminus J$, both $y\vee \pre_{J+y}(y)$ and $y\vee \nex_{J+y}(y)$ belong to $sGy$ and so one is an ancestor of the other. Observe that $\pi_J(y)$ is never equal to 0 (because then $s=0\vee y \sqsubset \nex_{J+y}(y)\vee y$ which is a contradiction as $s$ is the root), nor $n+1$ (because then $s=(n+1)\vee y \sqsubseteq \pre_{J+y}(y) \vee y \sqsubseteq v_R \sqsubset s$). In particular, $\pi_J(y)\in \{\pre_{J+y}(y),y,\nex_{J+y}(y)\}\cap J$. A graphical way to understand the definition of the representative of an element $y$ outside $J$ is the following: let  $j$ and $j'$ be the left and right neighbors of $y$ in $J$ respectively and call $sGj$ and $sGj'$ the left and right paths respectively. The representative of $y$ is its left neighbor (respectively, its right neighbor) if and only if by starting from $y$ and walking up on $G$ against its orientation, the first path hit is the left path (respectively the right path). In case of a tie, the representative is the right neighbor (see Figure \ref{fig:laminar}).

We claim also that for every $j\in J$, the set of elements \emph{to its right} having $j$ as representative is an interval of the form $[j,k]$ with $k<\nex_J(j)$. Indeed, this is true if $\nex_{J}(j)=n+1$. Suppose now that $j':=\nex_{J}(j)\leq n$, and that the claim does not hold. Then, there must be two consecutive terminals $i, i+1\in (j,j')$ with $\pi_J(i+1)=j$ and $\pi_J(i)=j'$. But then, we get the following contradiction:
\begin{align*}
	(i+1) \vee j' \sqsubseteq i \vee j' \sqsubseteq i \vee j \sqsubseteq (i+1) \vee j \sqsubset (i+1) \vee j'
\end{align*}
where the first inequality holds since $[i+1,j']\subseteq [i,j']$, the second, by the definition of $\pi_J(i)$, the third since $[j,i]\subseteq [j,i+1]$ and the fourth by definition of $\pi_J(i+1)$. Since every element in $(j,j')$ either has $j$ or $j'$ as representative, we conclude, in fact, that the entire set $\pi^{-1}_J(j)$ of elements with $j$ as representative is an interval of terminals enclosing $j$ but strictly contained in $[\pre_J(j),\nex_J(j)]$. In other words $(\pi_J(j))_{j\in J}$ is a partition of $[n]$ into $|J|$ intervals.\\

\begin{figure}
	\centering
\begin{tikzpicture}
\tikzstyle{myedgestyle} = [draw=gray!25, line width=1]
\tikzstyle{fat} = [draw=black!90, line width=2,-latex]
\tikzstyle{rojo} = [draw=red!90, line width=2,-latex]
\tikzstyle{thin} = [draw=black!50, line width=1,-latex]
\tikzset{every node/.style={draw,circle, minimum size=0.5cm}}

\node  (v5) at (4.4,1.6) {};
\node (v3) at (1.2,3.0) {};
\node (v1) at (-5.5,0) {1};
\node  [fill=red!50] (v2) at (5,0) {15};
\node  [fill=red!50]  (v17) at (3.4,0) {13};
\node (v31) at (4.2,0) {14};
\node  (v15) at (-0.5,0) {8};
\node (v30) at (-1.2,0) {7};
\node (v29) at (-2,0) {6};
\node [fill=red!50] (v14) at (-2.6,0) {5};
\node (v13) at (-3.3,0) {4};
\node (v18) at (-6.5,0) {0};
\node (v19) at (7.7,0) {18};
\node (v23) at (1.8,0) {11};
\node (v27) at (-4.8,0) {2};
\node [fill=red!50] (v28) at (-4.1,0) {3};
\node (v20) at (5.8,0) {16};
\node (v26) at (0.2,0) {9};
\node [fill=red!50] (v25) at (1,0) {10};
\node (v22) at (2.6,0) {12};
\node (v11) at (-3.7,1.6) {};
\node (v7) at (5.6,1.2) {};
\node (v6) at (-4.8,1.2) {};

\node (v9) at (-1,2.6) {};

\node (v12) at (-2.6,1.2) {};
\node (v16) at (1,1.2) {};
\node (v24) at (3.3,1.2) {};
\node (v21) at (6.6,0) {17};

\draw [thin] (v7) edge (v2);
\draw [thin] (v7) edge (v20);
\draw [thin] (v7) edge  (v21);
\draw [thin] (v5) edge node[above right, draw=none] {2}  (v7);
\draw [thin] (v5) edge node[above left, draw=none] {2} (v24);
\draw [thin] (v24) edge (v22);
\draw [thin] (v24) edge (v17);
\draw [thin] (v3) edge node[above right, draw=none] {3}  (v5);
\draw [thin] (v9) edge node[right, draw=none] {2}  (v16);
\draw [thin] (v16) edge (v23);
\draw [thin] (v16) edge (v25);
\draw [thin] (v16) edge (v26);
\draw [thin] (v3) edge node[above left, draw=none] {4} (v9);
\draw [thin] (v9) edge (v15);
\draw [thin] (v6) edge (v1);
\draw [thin] (v6) edge (v27);
\draw [thin] (v6) edge (v28);
\draw [thin] (v12) edge (v13);
\draw [thin] (v12) edge (v14);
\draw [thin] (v11) edge node[above left, draw=none] {2}   (v6);
\draw [thin] (v11) edge node[above right, draw=none] {2}  (v12);
\draw [thin] (v9) edge node[above left, draw=none] {3}  (v11);
\node [fill=blue] (v4) at (1.2,4.0) {};
\draw [thin] (v4) edge node[left, draw=none] {5}  (v3);
\draw [thin] plot[smooth, tension=.7] coordinates {(v4) (-5,2.6) (v18)};
\draw [thin] plot[smooth, tension=.7] coordinates {(v4) (6.1,2.6) (v19)};

\draw [thin] (v12) edge (v29);
\draw [thin] (v9) edge (v30);
\draw [thin] (v24) edge (v31);

\draw [line width=1.5] (-6,0.3) -- (-6,-0.5) -- (7.1,-0.5) -- (7.1,0.4);

\draw [line width=1.5] (-1.6,0.3) -- (-1.6,-0.5);
\draw [line width=1.5] (2.2,0.4) -- (2.2,-0.5);
\draw [line width =1.5] (-3.7,0.3) -- (-3.7,-0.5);

\draw [line width =1.5] (4.6,0.4) -- (4.6,-0.5);
\end{tikzpicture}
	\caption{\label{fig:laminar} \small Semiplanar drawing of a laminar matroid with ground set $[17]$. Nonunit arc capacities are shown. Below we show the partition $(\pi_{J}(j))_{j\in J}$ induced by the independent set $J=\{3,5,10,13,15\}$. }
\end{figure}
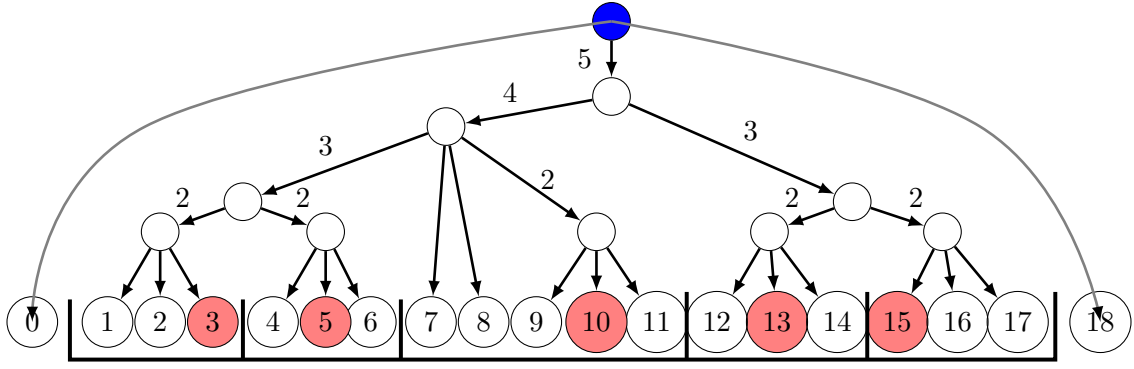

\noindent{\it The algorithms.} In what follows we describe our algorithms for semiplanar gammoids and  laminar matroids. Apart from some special case occurring when the sample is empty, our algorithms are extremely simple. Each one computes the optimal solution $\OPT(R_s)$ of the sample and leave their elements \emph{unmarked}. When a terminal $r_i$ that is part of the current optimum arrives, it checks if the (closest neighbors / representative) of $r_i$ in $\OPT(R_s)$ (are / is) unmarked. If so, it marks (them / it) and selects $r_i$ as part of the solution.

\noindent \begin{minipage}[t]{.56 \textwidth}
 \vspace{0pt}  
\begin{algorithm}[H]
	\algrenewcommand\algorithmicindent{0.8em}
	\small
	\begin{algorithmic}[1]
		\Require{An semiplanar gammoid with a fixed semiplanar drawing}	
		\State $\ALG\gets \emptyset$, $s\gets \Bin(n,p)$, 
		\If{$s=0$}
		\State{$\ALG\gets \{r_{1}\}$}
		\Else\ $B\gets \emptyset$.
		\For{$i=s+1$ to $n$}
		\If{$r_i\in \OPT(R_i)$, $\pre_{\OPT(R_s)+r_i}(r_i) \not \in B$ and \phantom{.......} \phantom{.....}$\nex_{\OPT(R_s)+r_i}(r_i) \not \in B$}\label{alg:line-6}
		\State \begin{varwidth}[t]{\linewidth}
		$\ALG\gets \ALG + r_i$, and \\
		$B\gets B\cup\{ \pre_{\OPT(R_s)+r_i}(r_i), \nex_{\OPT(R_s)+r_i}(r_i)\}$
	\end{varwidth}	\EndIf
		\EndFor
		\EndIf
		\State Return $\ALG$.
	\end{algorithmic}
	\caption{\small for semiplanar gammoids. \label{alg:semiplanar}}
\end{algorithm}

\end{minipage}%
\hfill
\begin{minipage}[t]{.4 \textwidth}
  \vspace{0pt}
\begin{algorithm}[H] 
	\algrenewcommand\algorithmicindent{0.9em}
	\small
	\begin{algorithmic}[1]
		\Require{A laminar matroid with a fixed semiplanar drawing }		
		\State $\ALG\gets \emptyset$, $s\gets \Bin(n,p)$, 
		\If{$s=0$}
		\State{$\ALG\gets \{r_1\}$}
		\Else\ $B\gets \emptyset$.
		\For{$i=s+1$ to $n$}
		\If{ $r_i\in \OPT(R_i)$ and \phantom{\qquad \qquad ......} \phantom{.....} $\pi_{\OPT(R_s)}(r_i)\not\in B$\quad}
		\State \begin{varwidth}[t]{\linewidth}
				$\vphantom{\big|}\ALG\gets \ALG + r_i$, and \\
				$B\gets B\cup \{\pi_{\OPT(R_s)}(r_i)\}$.
		\end{varwidth}
		\EndIf
		\EndFor
		\EndIf
		\State Return $\ALG$.
	\end{algorithmic}
	\caption{\small for laminar matroids. \label{alg:laminar}}
\end{algorithm}
\end{minipage}
\vspace{.5cm}

\begin{thm}\label{thm:semiplanar}
	Algorithm \ref{alg:semiplanar} has forbidden sets of size $4$. By setting $p=p(4)=\sqrt[3]{1/4}$ we get an $\alpha(4)=4^{4/3}\approx 6.3496$ probability-competitive algorithm for semiplanar ACGs.
\end{thm}

The previous result applied to laminar matroids already beat the best guarantees known for that class (9.6 \cite{MTW2016} and $3\sqrt{3}e$ \cite{JSZ2013, JSZ2014}). But we can do better for laminar matroids since for those, there is a unique path from the source to each terminal. 

\begin{thm}\label{thm:laminar}
	Algorithm \ref{alg:laminar} has forbidden sets of size $3$. By setting $p=p(3)=\sqrt{1/3}$ we get an $\alpha(3)=3\sqrt{3}$ probability-competitive algorithm for laminar matroids.
\end{thm}

We observe that this algorithm is very similar to the $3\sqrt{3}e$ competitive algorithm of Jaillet et al.~\cite{JSZ2013, JSZ2014}. Note that the partition of the terminals given by $\pi_{\OPT(R_s)}$ induces a unitary partition matroid $\calP'$. After the sample, Algorithm \ref{alg:laminar} simply selects on each part, the first arriving element that is part of the current optimum.  It can be shown that the partition matroid $\calP'$ we define is the same as the one defined in \cite[Section 3.2]{JSZ2014}. The main algorithmic difference is that in \cite{JSZ2014}, the authors use the algorithm for the classic secretary problem to select one terminal on each part that has constant probability of being the largest element. Instead, we select the first arriving element on each part that is part of the optimum at the time of arrival.
This small change makes the competitive ratio of our algorithm $e$ times smaller than theirs, but the analysis is more involved. To prove Theorems \ref{thm:semiplanar} and \ref{thm:laminar} we need to develop some extra tools.

\paragraph{The unit capacity semiplanar gammoid associated to an independent set and its standard arc-covering.} 
For every nonempty independent set $J\subseteq [n]$ in the semiplanar gammoid $\calM[N]$ we chose an arbitrary but fixed $s$-$J$ flow $f_J$ satisfying the arc capacity constraints. Define the unit capacity network $N^1(J)=(G^1(J),s,R,1)$, where $G^1(J)$ is obtained from $G$ by splitting each arc $uv$ with $f_J(uv)\geq 2$ into $f_J(uv)$ parallel arcs (we keep all arcs $uv$ with $f_J(uv)\in \{0,1\}$, we also keep the arcs $s0$ and $s(n+1)$). We do this in such a way that $G^1(J)$ is still semiplanar. The matroid $\calM[N^1(J)]$ is the unit capacity semiplanar gammoid associated to $J$. 

Observe that every path in $G^1(J)$ corresponds canonically to an unsplitted path in $G$ going through the same nodes. Furthermore, every arc-disjoint collection $\calP$ of $s$-$R$ paths in $G^1(J)$ can be regarded, by unsplitting parallel arcs, as an $s$-$R$ flow on $G$ satisfying the capacity constraints. Therefore, we have the following important lemma.

\begin{lem}\label{lem:subsubmatroid}
	Any independent set in $\calM[N^1(J)]$ is also independent in the original matroid $\calM[N]$.
\end{lem}

We say that two paths $P$ and $Q$ in the drawing of a graph \emph{cross} if $P$ enters $Q$ on the left, shares zero or more arcs with $Q$ and then exits $Q$ on the right, or viceversa. A collection of paths is mutually noncrossing if no pair of them cross.\footnote{An alternative way to understand this definition is to imagine that in the drawing nodes and arcs have positive area (they are circles and thick lines). A drawing of a path is any continuous non-self-intersecting curve drawn \emph{inside} the union of all circles and thick lines associated to the nodes and arcs of $P$ visiting them in the correct order. A collection of paths are mutually noncrossing if we can find drawings such that no pair of them intersects. Note that two noncrossing paths can still share arcs and nodes.} The $s$-$J$ flow $f_J$ in the network $N$ can be mapped to an $s$-$J$ flow $f^1_J$ in $N^1(J)$ in a natural way. By applying flow-decomposition on $f^1_J$ we get a collection $\{P^k\}_{k\in J}$ of arc-disjoint paths. Define also $P^0$ and $P^{n+1}$ as the paths consisting of a single arc $s0$ and $s(n+1)$ respectively. By a standard planar uncrossing argument, we can assume that all paths in $\{P^k\}_{k\in J\cup \{0,n+1\}}$ are mutually noncrossing (but they can still share internal nodes, in fact all paths $P^k$ with $k\in J$ contain $s$ and $v_R$, see Figure \ref{fig:1}). For $j\in J+0$, call $j'=\nex_{J+0}(j)$ and define the collection of arcs $A^j\subseteq E(G^1(J))$ as those that are drawn in the closed planar region $\mathcal{R}_j$ bounded by $P^j$, $P^{j'}$ and the $x$-axis. Since paths $\{P^k\}_{k\in J\cup \{0,n+1\}}$ form a topological star whose tips are in the $x$-axis we conclude that $\{\calR_k\}_{k\in J+0}$ is a division of the region bounded by $P^0$, $P^{n+1}$ and the $x$-axis, and thus $\{A^k\}_{k\in J+0}$ is a covering of all arcs in $G^1(J)$. In fact, if $1\leq j'\neq n+1$, then every arc of $P^{j'}$ belongs to pieces $A^j$ and $A^{j'}$, while each arc in $G^1(J)\setminus \{P^k\}_{k\in J}$ belongs to a single piece of the covering. Furthermore, the only terminals contained in region $\mathcal{R}_j$ are those in $[j,j']$.

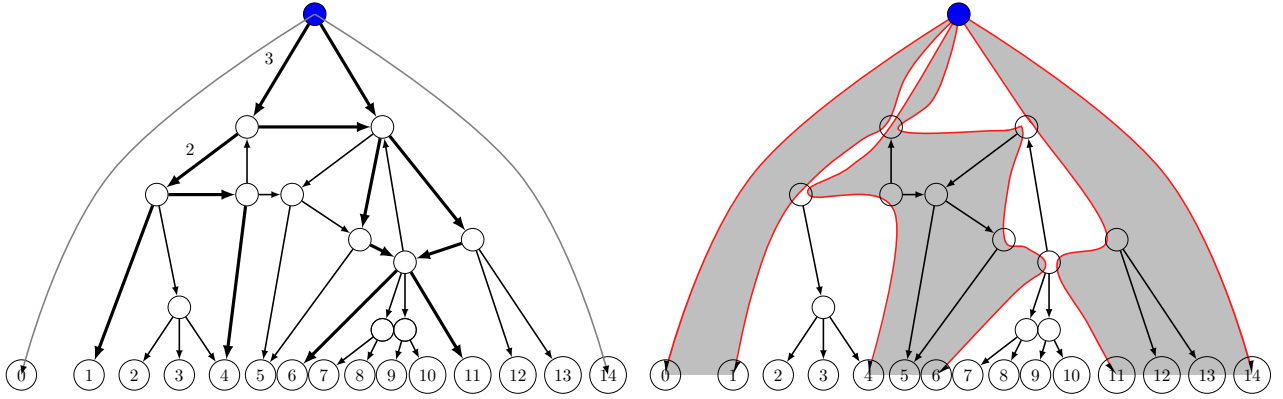
\begin{figure}[t]
	\centering
	{}\hspace{-10pt}{}
	\scalebox{0.6}{
		\begin{tikzpicture}
		\tikzstyle{myedgestyle} = [draw=gray!25, line width=1]
		\tikzstyle{fat} = [draw=black!90, line width=2,-latex]
		\tikzstyle{rojo} = [draw=red!90, line width=2,-latex]
		\tikzstyle{thin} = [draw=black!50, line width=1,-latex]
		\tikzset{every node/.style={draw,circle, minimum size=0.5cm}}
		
		\node  (v5) at (-1.5,5.5) {};
		\node[fill=blue] (v3) at (0,8) {};
		\node (v1) at (-5,0) {1};
		\node (v2) at (4.5,0) {12};
		\node (v17) at (3.5,0) {11};
		\node (v15) at (-0.5,0) {6};
		\node (v14) at (-1.2,0) {5};
		\node (v13) at (-2,0) {4};
		\node (v18) at (-6.5,0) {0};
		\node (v19) at (6.5,0) {14};
		\node (v23) at (1.7,0) {9};
		\node (v27) at (-4,0) {2};
		\node (v28) at (-3,0) {3};
		\node (v20) at (5.5,0) {13};
		\node (v26) at (0.2,0) {7};
		\node (v25) at (1,0) {8};
		\node (v22) at (2.5,0) {10};
		\node (v11) at (-3.5,4) {};
		\node (v7) at (3.5,3) {};
		\node (v6) at (1,3) {};
		\node (v8) at (2,2.5) {};
		\node (v9) at (-0.5,4) {};
		\node (v10) at (-1.5,4) {};
		\node (v12) at (-3,1.5) {};
		\node (v16) at (1.5,5.5) {};
		\node (v24) at (1.5,1) {};
		\node (v21) at (2,1) {};

		\draw [thin] plot[smooth, tension=.7] coordinates {(v3) (-4.5,4.5) (v18)};
		\draw [thin] plot[smooth, tension=.7] coordinates {(v3) (4.5,4.5) (v19)};
		
		\draw [fat] (v3) edge (v16);
		\draw [fat] (v16) edge (v7);
		\draw [thin] (v7) edge (v2);
		\draw [thin] (v7) edge (v20);
		
		\draw [fat] (v3) edge node[above left, draw=none] {3}  (v5);
		\draw [fat] (v5) edge (v16);
		\draw [fat] (v16) edge (v6);
		\draw [fat] (v7) edge (v8);
		\draw [fat] (v8) edge (v17);
		\draw [fat] (v6) edge (v8);
		\draw [fat] (v8) edge (v15);
		\draw [fat] (v5) edge node[above left, draw=none] {2}  (v11);
		\draw [fat] (v11) edge (v10);
		\draw [fat] (v10) edge (v13);
		\draw [fat] (v11) edge (v1);
		\node (v24) at (1.5,1) {};
		\node (v21) at (2,1) {};
		\draw [thin] (v8) edge (v21);
		\draw [thin] (v21) edge (v22);
		\draw [thin] (v21) edge (v23);
		\draw [thin] (v8) edge (v24);
		\draw [thin] (v24) edge (v25);
		
		\draw [thin] (v24) edge (v26);
		\draw [thin] (v16) edge (v9);
		\draw [thin] (v10) edge (v9);
		\draw [thin] (v9) edge (v14);
		\draw [thin] (v11) edge (v12);
		\draw [thin] (v12) edge (v27);
		\draw [thin] (v12) edge (v28);
		\draw [thin] (v10) edge (v5);
		\draw [thin] (v8) edge (v16);
		\draw [thin] (v9) edge (v6)  ;
		\draw [thin] (v12) edge (v13);
		\draw [thin] (v6) edge (v14);
		\end{tikzpicture} \quad
		\begin{tikzpicture}
		\tikzstyle{myedgestyle} = [draw=gray!25, line width=1]
		\tikzstyle{fat} = [draw=black!90, line width=2,-latex]
		\tikzstyle{rojo} = [draw=red!90, line width=2,-latex]
		\tikzstyle{thin} = [draw=black!50, line width=1,-latex]
		\tikzset{every node/.style={draw,circle, minimum size=0.5cm}}
		
		\node  (v5) at (-1.5,5.5) {};
		\node[fill=blue] (v3) at (0,8) {};
		\node (v1) at (-5,0) {1};
		\node (v2) at (4.5,0) {12};
		\node (v17) at (3.5,0) {11};
		\node (v15) at (-0.5,0) {6};
		\node (v14) at (-1.2,0) {5};
		\node (v13) at (-2,0) {4};
		\node (v18) at (-6.5,0) {0};
		\node (v19) at (6.5,0) {14};
		\node (v23) at (1.7,0) {9};
		\node (v27) at (-4,0) {2};
		\node (v28) at (-3,0) {3};
		\node (v20) at (5.5,0) {13};
		\node (v26) at (0.2,0) {7};
		\node (v25) at (1,0) {8};
		\node (v22) at (2.5,0) {10};
		\node (v11) at (-3.5,4) {};
		\node (v7) at (3.5,3) {};
		\node (v6) at (1,3) {};
		\node (v8) at (2,2.5) {};
		\node (v9) at (-0.5,4) {};
		\node (v10) at (-1.5,4) {};
		\node (v12) at (-3,1.5) {};
		\node (v16) at (1.5,5.5) {};
		\node (v24) at (1.5,1) {};
		\node (v21) at (2,1) {};
		
		\begin{scope}[on background layer]
		\draw [rojo, name path=one] plot[smooth, tension=.7] coordinates {(v3) (-4.5,4.5) (v18)};
		\draw [rojo, name path=six] plot[smooth, tension=.7] coordinates {(v3) (4.5,4.5) (v19)};
		
		\draw [rojo, name path=two] plot[smooth, tension=.7] coordinates {(v3) (-0.6,7.4)(-1.2,6.6) (-1.6,5.6) (-2.2,5.2) (-3.2,4.4) (-4,3.4) (-4.6,1.8) (v1)};

		\draw [rojo, name path=three] plot[smooth, tension=.7] coordinates {(v3) (v5) (-2.4,4.6) (-3.3,4)(-1.4,3.6)(v13)};
		
		\draw [rojo, name path=four] plot[smooth, tension=.7] coordinates {(v3) (-0.2,7.1) (-0.6,6.1) (-1.3,5.4) (0.9,5.4)  (1.4,5.4) (1.1,4.3) (v6)(1.5,2.8) (1.9,2.4)(1.2,1.8) (v15)};
		
		\draw [rojo, name path=five] plot[smooth, tension=.7] 
		coordinates {(v3) (1.4,5.9) (1.7,5.5) (3.3,3.1) (2.2,2.5)
			(2.9,1.1) (v17)};
			
		\tikzfillbetween[
		of=one and two,split
		] {gray!50};
		\tikzfillbetween[
		of=three and four,split
		] {gray!50};
		\tikzfillbetween[
		of=five and six,split
		] {gray!50};

		\end{scope}

		%\draw [fat] (v3) edge (v16);
		%\draw [fat] (v16) edge (v7);
		\draw [thin] (v7) edge (v2);
		
		\draw [thin] (v7) edge (v20);
		
		%\draw [fat] (v3) edge node[above left,draw=none] {3}  (v5);
		
		%\draw [fat] (v5) edge (v16);
		%\draw [fat] (v16) edge (v6);
		%\draw [fat] (v7) edge (v8);
		%\draw [fat] (v8) edge (v17);
		%\draw [fat] (v6) edge (v8);
		%\draw [fat] (v8) edge (v15);
		
		%\draw [fat] (v5) edge node[above left, draw=none] {2}  (v11);
		%\draw [fat] (v11) edge (v10);
		%\draw [fat] (v10) edge (v13);
		%\draw [fat] (v11) edge (v1);
		
		\draw [thin] (v8) edge (v21);
		\draw [thin] (v21) edge (v22);
		\draw [thin] (v21) edge (v23);
		\draw [thin] (v8) edge (v24);
		\draw [thin] (v24) edge (v25);
		\draw [thin] (v24) edge (v26);
		\draw [thin] (v16) edge (v9);
		\draw [thin] (v10) edge (v9);
		\draw [thin] (v9) edge (v14);
		\draw [thin] (v11) edge (v12);
		\draw [thin] (v12) edge (v27);
		\draw [thin] (v12) edge (v28);
		\draw [thin] (v10) edge (v5);
		\draw [thin] (v8) edge (v16);
		\draw [thin] (v9) edge (v6)  ;
		\draw [thin] (v12) edge (v13);
		\draw [thin] (v6) edge (v14);
		\end{tikzpicture}
	}\vspace{-15pt}
	\caption{\label{fig:1} \small On the left, a flow $f_J$ associated to the independent set $J=\{1,4,6,11\}$ on a semiplanar gammoid $\calM[N]$ with ground set $R=[13]$ ($0$ and $14$ are auxiliary nodes). On the right, $\calM[N^1(J)]$ and the partition of the drawing into regions $\calR_0$, $\calR_1$, $\calR_4$, $\calR_6$ and $\calR_{11}$ induced by the paths $\{P^j\}_{j\in J\cup \{0,14\}}$. The regions $\calR_0$, $\calR_4$ and $\calR_{11}$ are shaded.}
\end{figure}

\begin{lem}\label{lem:distinctpaths}
	Let $y \in [n]\setminus J$, and $j=\pre_{J+y}(y)$, $j'=\nex_{J+y}y$ be its neighbors in $J$. Then there is a directed $s$-$y$ path $D_J(y)$ in $G^1(J)$ whose arc set is completely contained in $A^j$. Furthermore, $D_J(y)$ can be chosen so that it arcs-intersects at most one path $\bar{P}$ in $(P^k)_{k\in J}$. In the semiplanar case, $\bar{P}$ must be one of $P^j$ and $P^{j'}$, and in the laminar case, $\bar{P}=P^{\pi_J(\ell)}$.
\end{lem}
\begin{proof}
	Note that $y\in [j,j']$ is in the planar region $\calR_j$. Since $y$ is reachable from $s$ in $G$, there is a path $Q$ from $s$ to $y$ in $G^1(J)$. Let $v$ be the last vertex in $Q$ contained in the vertices of $P^j$ and $P^{j'}$ ($v$ maybe in one or both paths), say $v$ is in $\bar{P}\in \{P^j,P^{j'}\}$. By concatenating the initial piece of $\bar{P}$ from $s$ to $v$ and the final piece of $Q$ from $v$ to $y$ we get a path $D_J(y)$ that is completely contained in $A^j$ and that intersects the arcs of at most one path in $\{P_j, P_{j'}\}$.
	
	Consider the same situation in the laminar case. All arcs in $P^j$, $P^{j'}$ and $Q$ are splitted versions of arcs in $sGj$, $sGj'$ and $sGy$ respectively. The vertex $v$ defined above satisfies $v=y\vee j \vee j'$. If $y\vee j \sqsubset y\vee j'$ then $v=y\vee j$, $\pi_J(y)=j$ and we can construct $D_J(y)$ avoiding all arcs in $P^{j'}$ by selecting $\bar{P}=P^j$ in the previous argument. Analogously, if $y\vee j' \sqsubseteq y\vee j$ then $v=y\vee j'$, $\pi_J(y)=j'$, and we can choose $D_J(\ell)$ to avoid $P^{j}$ by setting $\bar{P}=P^{j'}$.
\end{proof}

\begin{lem}\label{lem:submatroids}
	Let $J\subseteq [n]$ be a nonempty independent set in $\calM[N]$ and $I\subseteq [n]\setminus J$. Consider the following properties.
	\begin{compactenum}
		\item[(P1)] For every $x, y \in I$ with $x\neq y$, 
		$\{\pre_{J+x}(x), \nex_{J+x}(x)\}\cap \{\pre_{J+y}(y), \nex_{J+y}(y)\}=\emptyset$.
		\item[(P2)] $\calM[N]$ is laminar and for every $x, y \in I$, $\pi_J(x)\neq \pi_J(y)$.
	\end{compactenum}
	If either (P1) or (P2) holds then $I$ is independent in $\calM[N^1(J)]$
\end{lem}

\begin{proof}
	Suppose that $(P1)$ is satisfied. Then for every pair of distinct $x, y \in I$, the paths $D_J(x)$ and $D_J(y)$ constructed in Lemma \ref{lem:distinctpaths} must lie on non-consecutive (and hence disjoint) regions in $\{\calR_k\}_{k\in J+0}$. Therefore the paths $\{D_J(z)\}_{z\in I}$ are mutually arc-disjoint from which we conclude that $I$ is an independent set in $\calM[N^1(J)]$.
	
	Suppose now that $(P2)$ is satisfied. Let $x<y$ be distinct terminals in $I$. We claim that the paths $D_J(x)$ and $D_J(y)$ are arc disjoint. If $x$ and $y$ belong to different regions then the only possibility for them to share an arc is that $x$ is in $\calR_j$, $y\in \calR_{j'}$ with $j'=\nex_{J}(j)$, and both share arcs in $P^{j'}$. But then, by Lemma \ref{lem:distinctpaths}, $\pi_J(x)=\pi_J(y)={j'}$ which is a contradiction.
	
	If $x$ and $y$ are in the same region $\calR_j$, then $x, y \in [j, j']$ with $j=\pre_{J+x}(x)=\pre_{J+y}(y)$ and $j'=\nex_{J+x}(x)=\nex_{J+y}(y)\in J+(n+1)$. Since the function $\pi_J$ partitions $[n]$ into intervals and $x < y$, we must have $\pi_J(x)=j$, $\pi_J(y)=j'$. Suppose now that $D_J(x)$ and $D_J(y)$ had an arc $a$ in common and let $w$ be its head. Since $D_J(x)$ and $D_J(y)$ are split versions of $sGx$ and $sGy$ we get that $w\sqsupseteq x\vee y$. Since  $w$ and $x\vee j$ are both in $sGx$, one is an ancestor of the other. Note that $w$ cannot be an ancestor of $x\vee j$ since above the latter $D_J(x)$ coincides with $P^{j}$ which is arc-disjoint from $D_J(y)$ by Lemma \ref{lem:distinctpaths} (and $a$ is a common arc). It follows that $x\vee j \sqsupset w\sqsupseteq x\vee y \sqsupseteq y$, and thus, $x\vee j \sqsupseteq y\vee j$. But since $\pi_J(y)=j'$ we have $y\vee j \sqsupseteq y\vee j'$ and we conclude that $x\vee j\sqsupseteq y\vee j' \sqsupseteq j'$. From the last expression we get $x\vee j \sqsupseteq x\vee j'$ which contradicts the fact that $\pi_J(x)=j$. 
	
	We have thus shown that all paths $(D_J(z))_{z\in I}$ are arc-disjoint, thence $I$ is independent in $\calM[N^1(J)]$
\end{proof}

Now we are ready to define the forbidden sets of sizes at most 4 and 3 respectively for Algorithms \ref{alg:semiplanar} and \ref{alg:laminar}. 
For $(X,r^*)$ with $r^*\in [n]$, $X\subseteq [n]\setminus \{r^*\}$ define the sets
\begin{align*}
	\calF_4(X,r^*)&=\{\pre_{X}(\pre_{X+r^*}(r^*)), \pre_{X+r^*}(r^*), \nex_{X+r^*}(r^*),\nex_{X}(\nex_{X+r^*}(r^*)) \}\\
	I_4(X,r^*)&=[\pre_{X}(\pre_{X+r^*}(r^*)),\nex_{X}(\nex_{X+r^*}(r^*))].
\end{align*}
The set $\calF_4(X,r^*)$ contains the 2 closest terminals to the left of $r^*$ in $X$ and the 2 closest terminals to its right. $I_4(X,r^*)$ is the enclosing interval containing $\calF_4(X,r^*)$. Similarly, for the case of laminar matroids, define
\begin{align*}
	\calF_3(X,r^*)&=\{\pre_{X}(\pi_X(r^*)), \pi_X(r^*), \nex_{X}(\pi_X(r^*)) \}\\
	I_3(X,r^*)&=[\{\pre_{X}(\pi_X(r^*)), \nex_{X}(\pi_X(r^*)) \}].
\end{align*}
where $\calF_3(X,r^*)$ consists of the representative of $r^*$ in $X$ and its two neighbors, and $I_3(X,r^*)$ is its enclosing interval. We need one last technical lemma to analize the algorithms performance.

\begin{restatable}{lem}{lematecnico}
	\label{lem:tecnico} Let $X$ and $r^*$ as above, and let $x\in X$.
	\begin{compactenum}
		\item[(a)] If $x\in \calF_4(X,r^*)$ then $I_4(X,r^*)\subseteq I_4(X-x,r^*)$.
		\item[(b)] If $x\not\in \calF_4(X,r^*)$ then $\calF_4(X,r^*)=\calF_4(X-x,r^*)$ and $I_4(X,r^*)=I_4(X-x,r^*)$.
	\end{compactenum}
	If the matroid is laminar, the following properties also hold
	\begin{compactenum}
		\item[(c)] If $x\in \calF_3(X,r^*)$ then $I_3(X,r^*)\subseteq I_3(X-x,r^*)$.
		\item[(d)] If $x\not\in \calF_3(X,r^*)$ then $\calF_3(X,r^*)=\calF_3(X-x,r^*)$ and $I_3(X,r^*)=I_3(X-x,r^*)$.
	\end{compactenum}
\end{restatable}

%\lematecnico*
\begin{proof}[Proof of Lemma \ref{lem:tecnico}]\mbox{}\\
	
	\noindent \textbf{Part (a):}
	Let $\calF_4(X,r^*)=\{a,b,c,d\}$ from left to right (note that near the borders we could have $a=b$ or $c=d$), in particular $a\leq b \leq r^* \leq c \leq d$. Denote by $a^-=\pre_X(a)$ and $d^+=\nex_X(d)$.
	
	\begin{compactenum}
		\item[(i)] If $x\in \{a,b\}$ then $\calF_4(X-x,r^*)=(\calF_4(X,r^*)-x)+ a^-$ and so $I_4(X,r^*)=[a,d]\subseteq [a^-,d]=I_4(X-x,r^*)$.
		\item[(ii)] If $x\in \{c,d\}$ then $\calF_4(X-x,r^*)=(\calF_4(X,r^*)-x)+ d^+$ and so $I_4(X,r^*)=[a,d]\subseteq [a,d^+]=I_4(X-x,r^*)$.\\
	\end{compactenum}
	
	\noindent\textbf{Part (b):} Direct.
	
	For parts (c) and (d) let $\calF_3(X,r^*)=\{a,b,c\}$ from left to right (in particular $b=\pi_{X}(r^*)$, note that we could be in the cases $a=b=0$ or $b=c=n+1$). Denote by $a^-=\pre_X(a)$ and $c^+=\nex_X(c)$.\\
	
\noindent	\textbf{Part (c):} We have many possibilities to analyze. 
	
	\begin{compactenum}
		\item[(i)] If $x=b$, then the closest neighbors of $r^*$ in $X-x$ are $a$ and $c$. In particular $\pi_{X-x}(r^*)$ is either $a$ or $c$. In both cases, $I_3(X,r^*)=[a,c]\subseteq I_3(X-x,r^*)$.
		\item[(ii)] If $x=a\neq b$ and $r^*\in[b,c]$, then $b$ and $c$ are still the closest neighbors of $r^*$ in $X-x$, and in particular, the representative in $X-x$ is the same as that in $X$, i.e., $\pi_{X-x}(r^*)=b$. Note that since we removed $a$, $\pre_{X-x}(b)=a^-$ and so, $I_3(X,r^*)=[a,c]\subseteq [a^-,c]=I_3(X-x,r^*)$
		\item[(iii)] The case $x=c\neq b$ and $r^*\in [a,b]$ is analogous to the previous one
		\item[(iv)] If $x=a\neq b$,  and $r^*\in [a,b]$, then the closest neighbors of $r^*$ in $X-x$ are $a^-$ and $b$. We have that $r^*\vee b \sqsubseteq r^*\vee a \sqsubseteq r^*\vee a^-$ where the first inequality holds since $\pi_X(r^*)=b$ and the second one since $[a,r]\subseteq [a^-,r^*]$. We conclude that $\pi_{X-x}(r^*)=b$. In particular $I_3(X,r^*)=[a,c]\subseteq [a^-,c]=I_3(X-x,r^*)$.
		\item[(v)] If $x=c\neq b$ and $r^*\in [b,c]$ then the closest neighbors of $r^*$ in $X-x$ are $b$ and $c^+$. We have that
		$r^*\vee b \sqsubset r^*\vee c \sqsubseteq r^*\vee c^+$ where the first inequality holds since $\pi_X(r^*)=b$ and the second holds since $[r^*,c]\subseteq [r^*,c^+]$. We conclude that $\pi_{X-x}(r^*)=b$. In particular $I_3(X,r^*)=[a,c]\subseteq [a,c^+]=I_3(X-x,r^*)$.
	\end{compactenum}
	
\noindent	\textbf{Part (d):} 
	Since $x\not\in \calF_3(X,r^*)$ the neighbors of $r^*$ in $X-x$ are the same as those in $X$, it follows that $\pi_{X-x}(r^*)=\pi_{X}(r^*)=b$,  $\pre_{X-x}(b)=a$, $\nex_{X-x}(b)=c$. Therefore $\calF_3(X-x,r^*)=\calF_3(X,r^*)$ and $I_3(X-x,r^*)=I_3(X,r^*)$. \qedhere
	
\end{proof}

Now we are ready to prove the guarantees for Algorithms \ref{alg:semiplanar} and \ref{alg:laminar}.

\begin{proof}[Proofs of Theorems \ref{thm:semiplanar} and \ref{thm:laminar}]
	
	\mbox{}
	
	\noindent	For both algorithms, if $s=0$ then $|\ALG|=1$ and so it is independent. In the following assume that $s\geq 1$. Line \ref{alg:line-6} in the algorithms guarantees that the set $\ALG$ satisfies the conditions of Lemma \ref{lem:submatroids}, hence $\ALG$ is independent in $\calM[N^1(J)]$, and by Lemma \ref{lem:subsubmatroid}, $\ALG$ is independent in the original matroid. This proves correctness. Since the sampling condition holds by construction, we only need to check the forbidden property.  
	For the rest of the proof, define $\calF(X,r^*)=\calF_i(X,r^*)$, $I(X,r^*)=I_i(X,r^*)$ where $i=4$ on the semiplanar case, and $i=3$ for the laminar case.  We will show that the sets $\calF(X,Y,r^*):=\calF(X,r^*)\cap [n]$ are forbidden sets of size at most $4$ for Algorithm \ref{alg:semiplanar} and of size $3$ for Algorithm  \ref{alg:laminar}.
	
	Let $r^*\in \OPT(Y)$ where $Y$ is an arbitrary set of $t\geq s+1$ elements, and suppose that $R_t=Y$ and $r_t=r^*$. Assume now that the condition
	\begin{align*}
		\tag{$\star$}  \text{for every $i\in \{s+1,\dots, t-1\}$}, r_i\not\in \calF(\OPT(R_i), r_t),
	\end{align*}
	holds. We have to show that $r_t=r^*$ is chosen by the algorithm.\\
	
	\noindent	\emph{Claim. The intervals $I(\OPT(R_i),r^*)$ are non-decreasing in $i$, namely, for all $i\leq t-2$, \[I(\OPT(R_i),r^*)\subseteq I(\OPT(R_{i+1}),r^*).\] }
	Indeed, let $i\leq t-2$. If $\OPT(R_i)=\OPT(R_{i+1})$ then the claim is trivial, so assume otherwise. In particular, we have $r_{i+1}\in \OPT(R_{i+1})\setminus \OPT(R_i)$. To simplify notation, let $A=\OPT(R_{i+1})$ and $A'=A-r_{i+1}$. By condition $(\star)$, $r_{i+1} \not\in \calF(A,r^*)$, and then by Lemma \ref{lem:tecnico} using $X=A$,
	\begin{equation}\label{eq:proofclaim} \calF(A,r^*)=\calF(A',r^*) \text{ and } I(A,r^*)=I(A',r^*).
	\end{equation}
	
	By the matroid exchange axiom we have two cases: either $\OPT(R_i)=A'$ or $\OPT(R_i)=A'+\tilde{r}$ for some $\tilde{r}\neq r_{i+1}$. In the first case we have by \eqref{eq:proofclaim} that  $I(\OPT(R_i),r^*)=I(A',r^*)=I(A,r^*)=I(\OPT(R_{i+1}),r^*)$, which ends the proof, and so we focus on the the second case. If $\tilde{r} \in \calF(A'+\tilde{r},r)$, then by \eqref{eq:proofclaim} and Lemma \ref{lem:tecnico} applied to $X=A'+\tilde{r}$, we have $I(\OPT(R_i),r^*)=I(A'+\tilde{r}, r^*)\subseteq I(A',r^*)=I(A, r^*)=I(\OPT(R_{i+1}),r^*)$. On the other hand, if $\tilde{r}\not\in \calF(A'+\tilde{r},r)$, then, again by Lemma \ref{lem:tecnico} and and \eqref{eq:proofclaim} we have $I(\OPT(R_i),r^*)=I(A'+\tilde{r},r^*)= I(A',r^*)=I(A,r^*)=I(\OPT(R_{i+1}),r^*)$. This concludes the proof of the claim.
	
	We show how to finish the proof of the lemma using the claim. Suppose that $r_i$ is selected by the algorithm for some $i\in \{s+1,\dots, t-1\}$. By $(\star)$ and the claim we deduce that $r_i \not\in I(\OPT(R_i),r^*)\supseteq I(\OPT(R_s),r^*)$. In particular $r_i$ is \emph{far away} from $r^*$:
	\begin{compactenum}
		\item[(a)] In the semiplanar case, there are at least 2 terminals of $\OPT(R_s)$ between $r^*$ and $r_i$. In particular,  $\{\pre_{\OPT(R_s)+r_i}(r_i),$ $\nex_{\OPT(R_s)+r_i}(r_i)\}$ and $\{\pre_{\OPT(R_s)+r^*}(r^*),$ $\nex_{\OPT(R_s)+r^*}(r^*)\}$ do not intersect and so, at iteration $i$, neither $\pre_{\OPT(R_s)+r^*}(r^*)$ nor $\nex_{\OPT(R_s)+r^*}(r^*)$ are added into $B$.	
		\item[(b)] In the laminar case, there is at least one terminal of $\OPT(R_s)$ between $\pi_{\OPT(R_s)}(r^*)$ and $r_i$. In particular, $\pi_{\OPT(R_s)}(r_i)\neq \pi_{\OPT(R_s)}(r^*)$, and so, at iteration $i$, $\pi_{\OPT(R_s)}(r^*)$ is not added into $B$.
	\end{compactenum}
	
	Since the statements above are satisfied for every $i\leq t-1$, we conclude that at time $t$, $r_t=r^*$ satisfies the conditions in line \ref{alg:line-6} of the algorithms, and so it is selected. This concludes the proof that Algorithms \ref{alg:semiplanar} and \ref{alg:laminar} have forbidden sets of sizes 4 and 3 respectively.
\end{proof}

\subsection{Uniform matroids}

We devise a variant of Kleinberg's algorithm \cite{Kleinberg2005} for uniform matroids whose probability competitiveness tends to 1 as the rank $\rho$ goes to infinity. Kleinberg's algorithm is better described when both the rank $\rho=2^k$ and the number of elements $n=2^N$ are powers of 2. 
It was originally presented in a recursive manner, but it is illustrative to describe it without using recursion:\\

\noindent {\it Kleinberg's algorithm. }For every $i\in \mathbb{N}$, let $I_i$ be the interval of the first $n/2^i$ elements. 
Then the intervals $J_k:=I_{k}, J_{k-1}:=I_{k-1}\setminus I_k, \dots, J_1:=I_1\setminus I_2$ and $J_0=I_0\setminus I_1$ partition the ground set $R$. 
The algorithm treats each interval in $\{J_0,J_1,\ldots,J_k\}$ separately. 
For $0\leq i\leq k-1$, it selects at most $b_i=\rho/2^{i+1}$ elements from $J_{i}$ and at most one element from $J_k$, so that in total at most $\rho$ elements are selected. 
The way the algorithm selects an element or not is determined by a threshold for each interval. 
Namely, it selects the first arriving element from $J_k$ and for $i\in \{0,1,\ldots,k-1\}$, the algorithm uses as threshold in the interval $J_i$ the $(\rho/2^i)$-th highest element seen strictly before the interval, i.e., the $(\rho/2^i)$-th element of $I_{i+1}$. 
It selects every element better than this threshold until the budget $b_i$ of the interval is depleted, ignoring all the elements arriving later in this interval.\\

\noindent {\it The probability-ratio of Kleinberg's algorithm is at least 4/3.} Kleinberg \cite{Kleinberg2005} shows that the previous algorithm is guaranteed to obtain an expected fraction $1-O(\sqrt{1/\rho})$ of the optimum weight, which in our notation means to be $1/(1-O(\sqrt{1/\rho}))=1+O(\sqrt{1/\rho})$ utility-competitive. 
This guarantee also holds for the ordinal notion since the algorithm does not need weights. 
However, as we show next, its probability competitiveness is bounded away from 1. 
%For that, consider the set $R^{\rho-1}=\{r^1,r^2,\ldots,r^{\rho-1}\}$ of the $\rho-1$ highest elements of $R$ and let $r^\rho$ be the $\rho$-th highest element in $R$. 
Since $J_0$ and $I_1$ have the same cardinality, with probability tending to 1/4 as $\rho$ goes to infinity, we simultaneously have that $|I_1\cap R^{\rho-1}|<|J_0\cap R^{\rho-1}|$ and $r^\rho\in J_0$. 
Given that, the threshold for $J_0$, which is the $\rho/2$-th element of $I_1$, will be attained by an element of $R^{\rho-1}$ which is strictly better than $r^\rho$. 
Thus, $r^\rho$ will not be selected. Therefore, the algorithm selects $r^\rho$ (which is in $\OPT$) with probability at most $1-1/4=3/4$. \\%This bound is not tight, as we can refine this argument on each interval to get smaller upper bounds for this probability.\\ 

\noindent {\it An asymptotically 1 probability-competitive algorithm.} The next algorithm is a simple modification of Kleinberg's, which 
we write in a continuous setting. 
Every element is associated with a uniform random variable with support $[0,1)$.
It is useful for this part to imagine $[0,1)$ as a time interval, and identify the realization of the uniform random variable as the {\it arrival time} of the element.
For each $j\in \mathbb{N}$, let $I_j$ be the interval $[0,2^{-j})$ and $J_j=[2^{-j-1},2^{-j})$ to be its second half. The sequence $\{J_j\}_{j\geq 0}$ partitions the interval $[0,1)$. 
For convenience, let $K_j=[2^{-j-1},2^{-j-1}(2-4\varepsilon_j))$ be the left $(1-4\varepsilon_j)$ fraction of the interval $J_j$, for some  parameter $\varepsilon_j$, depending on $\rho$, to be chosen later.\\ 

\begin{algorithm}[H]
\caption{\small for uniform matroids. \label{alg:uniform}}
\begin{algorithmic}[1]
	\Require{A uniform matroid $U(n,\rho)$ with ground set $R$ and rank $\rho$.} 
	\State Get a sample of size $n$, independently and uniformly at random from $[0,1)$. Sort it from smallest to largest as $t_1<t_2<\dots <t_n$. Thus $t_i$ is interpreted as  the arrival time of $r_i$.
	\State $\ALG\gets \emptyset$
	\For {$i=1$ to $n$}
	\State Compute the index $j$ such that $t_i\in J_j=[2^{-j-1},2^{-j})$.
	\State Compute the threshold $f_j$ equal to the $\lceil (\frac12)^{j+1}(1+\varepsilon_j)\rho\rceil$-th highest element in $R_i$ with arrival time in $I_{j+1}=[0,2^{-j-1})$.
	\If{less than $(\frac12)^{j+1}\rho$ elements with arrival time in $J_j$ have been selected and $r_i\succ f_j$}
	\State $\ALG\gets \ALG+r_i$
	\EndIf
	\EndFor
	\State Return $\ALG$.
\end{algorithmic}
\end{algorithm}
\vspace{-10pt}
\begin{figure}[H]
	\centering
	\input{kleinberg-dibujo.tex}
	\caption{\small Definition of the sets $I_j$ and $J_j$}
\end{figure}
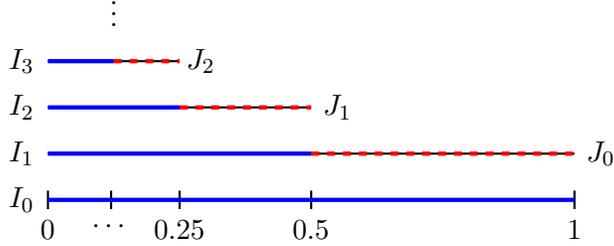

\begin{thm}\label{thm:uniform-prob}
Algorithm \ref{alg:uniform} is $1+O(\sqrt{\log \rho/ \rho})$ probability-competitive.
\end{thm}

By going from utility to probability we incur in a factor of $\sqrt{\log \rho}$ on the error probability. It remains open whether we can achieve Kleinberg's optimal utility-competitiveness of $1+O(1/\sqrt{\rho})$ for the stronger probability notion. 

\begin{proof}[Proof of Theorem \ref{thm:uniform-prob}]
	Since the algorithm selects at most $(1/2)^{j+1}\rho$ elements from interval $J_j$, in total at most $\sum_{j=0}^{\infty}(1/2)^{j+1}\rho=\rho$ elements are included into $\ALG$. 
	Therefore, the algorithm is correct.
	Since the matroid is uniform of rank $\rho$, the optimal base is $\OPT=\{r^1,r^2,\ldots,r^{\rho}\}$.  
	Let $r^*=r^k\in \OPT$, for some $k\in \{1,\ldots, \rho\}$. 
	In what follows, let $j$ be an integer such that $0\leq j\leq \lfloor\frac12\log (\rho/96)\rfloor:=j^*$ and let $\varepsilon_j=\sqrt{12\cdot 2^j\ln \rho/\rho}$. 
	We first find a lower bound on the probability that $r^*$ is selected, given that its arrival time is in $J_j$.

	For each $r^i\in R-r^*$, let $X^{(j)}_i$ and $Y^{(j)}_i$ be the indicator variables that the arrival time of $r^i$ is in $I_{j+1}$ and $K_j$ respectively. 
	Let $\sigma(f_j)$ be the global ranking of the computed threshold $f_j$, that is, $r^{\sigma(f_j)}=f_j$. 
%	Furthermore, let $f_j$ be the threshold element computed in 		$I_{j+1}$ by the algorithm, and call $i(f_j)$ its global ranking in the value order (i.e., $f_j=r^{i(f_j)})$). 
	Consider the following three events,
	\begin{align*}
	\sum_{i\in [\rho+1]\setminus\{k\}}\hspace{-10pt} X^{(j)}_i <  \frac{(1+\varepsilon_j)}{2^{j+1}}\rho,\quad 	\sum_{i\in \bigl[\bigl\lceil \frac{1+\varepsilon_j}{1-\varepsilon_j}\rho \bigr\rceil +1\bigr]\setminus\{k\}}\hspace{-10pt} X^{(j)}_i \geq  \frac{1+\varepsilon_j}{2^{j+1}}\rho,\quad 
	\sum_{i\in \bigl[\bigl\lceil \frac{1+\varepsilon_j}{1-\varepsilon_j}\rho \bigr\rceil +1\bigr]\setminus\{k\}}\hspace{-10pt} Y^{(j)}_i < \frac1{2^{j+1}} \rho,
	\end{align*}	
%	\begin{alignat*}{4}
%	(U_1)&& \sum_{i\in [\rho+1]\setminus\{k\}} X^{(j)}_i &<  \left(\frac12\right)^{j+1}(1+\varepsilon_j)\rho.   \qquad &(U_3)&& \sum_{i\in [\lceil (1+\varepsilon_j)\rho/(1-\varepsilon_j)\rceil +1]\setminus\{k\}}  Y^{(j)}_i &< 		\left(\frac12\right)^{j+1}\rho.\\
%	(U_2)&& \sum_{i\in [\lceil (1+\varepsilon_j)\rho/(1-\varepsilon_j)\rceil +1]\setminus\{k\}} X^{(j)}_i &\geq  \left(\frac12\right)^{j+1}(1+\varepsilon_j)\rho.   \quad & (U_4)&& t_k&\in K_j.
%	\end{alignat*}
 	that we call $U_1,U_2$ and $U_3$ respectively. 
	Consider a fourth event, $U_4$, that is $t^k\in K_j$ where $t^k$ is the arrival time of $r^k$.
	We claim that provided $t^k\in J_j$, the intersection of these four events guarantees that $r^*$ is selected by the algorithm.
	Conditional on the event that $t^k\in J_j$,  event $U_1\cap U_2$ implies that $\rho+1 < \sigma(f_j) \leq \lceil (1+\varepsilon_j)\rho/(1-\varepsilon_j)\rceil +1 $ and $r^*\succ f_j$.
	Event $U_2\cap U_3$ implies that the number of elements strictly higher than the threshold $f_j$ and arriving on interval $K_j$ is strictly less than $(\frac12)^{j+1}\rho$.
	In particular, this implies that the algorithm has not selected enough elements in the interval $J_j$. 
	Therefore, conditional on the event that $t^k\in J_j$, the event $U_1\cap U_2\cap U_3\cap U_4$ implies that $r^*\in \ALG$. 
%	and then $r^*$ is selected by the algorithm. if all four events hold then $r^*$ will be selected by the algorithm. 
	Calling $\overline{U_i}$ to the negation of event $U_i$, by union bound and noting that events from 1 to 3 are independent of the arrival time of $r^*=r^k$, we have
%	\begin{align*}
	$\Pr(r^k\not\in \ALG\,|\, t^k\in J_j) \leq \sum_{i=1}^4 \Pr(\overline{U_i}\,|\, t^k\in J_j) = \Pr(\overline{U_1})+\Pr(\overline{U_2})+\Pr(\overline{U_3}) + 4\varepsilon_j. $
%	\end{align*}

	The random variables $\{X_i^{(j)}\}_{i\in [n]\setminus \{k\}}$ are independently and identically Bernoulli distributed of parameter equal to the length of $I_{j+1}$, that is $2^{-j-1}$.  
	Similarly, the random variables $\{Y_i^{(j)}\}_{i\in [n]\setminus \{k\}}$ are independently and identically Bernoulli distributed of parameter equal to the length of $K_{j}$ that is $(1-4\varepsilon_j)2^{-j-1}$.
	In the following we upper bound the probabilities in the sum above by using the Chernoff bound \cite[p. 64-66, Theorems 4.4 and 4.5]{MU2005}.	
	\begin{align*}
	\Pr(\overline{U_1})%&= \Pr\left(\sum_{i\in [\rho+1]\setminus\{k\}} X^{(j)}_i \geq \left(\frac12\right)^{j+1}(1+\varepsilon_j)\rho \right)\\
	& \leq \exp\Bigl(-\frac{\varepsilon^2_j}{3}2^{-j-1}\rho\Bigr)= \exp\Bigl(-\frac{12\cdot 2^j\ln \rho}{3\rho}2^{-j-1}\rho\Bigr)=\exp\Bigl(-2\ln  \rho\Bigr) \leq \frac{1}{\rho}.
	\end{align*}

	Let $\mu_X$ and $\mu_Y$ be the expected sums of the random variables $X^{(j)}_i$, and respectively $Y_{i}^{(j)}$,  for $i \in \bigl[\bigl\lceil \frac{1+\varepsilon_j}{1-\varepsilon_j}\rho \bigr\rceil +1\bigr]\setminus\{k\}$. For $\mu_X$ we have  
	%\begin{align*}
	$\mu_X=\left\lceil\frac{1+\varepsilon_j}{1-\varepsilon_j}\rho\right\rceil\left(\frac12\right)^{j+1}\geq \frac{1+\varepsilon_j}{1-\varepsilon_j}\cdot \frac{\rho}{2^{j+1}}.$
%	\end{align*}
	The choice of $j$ guarantees that $1/\rho< \epsilon_j<1/8$, and therefore
	\begin{align*}
	\mu_Y&=\left\lceil \frac{1+\varepsilon_j}{1-\varepsilon_j}\rho\right\rceil \frac{1-4\varepsilon_j}{2^{j+1}} 
	\leq \left( \frac{1+\varepsilon_j}{1-\varepsilon_j}  + \varepsilon_j \right) \rho \cdot \frac{1-4\varepsilon_j}{2^{j+1}} 
	\leq \frac{(1-4\varepsilon_j)(1+2\varepsilon_j)}{1-\varepsilon_j}\frac{\rho}{2^{j+1}} \leq \frac1{1+\varepsilon_j} \cdot \frac{\rho}{2^{j+1}}.
	\end{align*}
	In the last inequality we used that $\varepsilon_j<1/8$. By Chernoff bound on events $\overline{U_2}$ and $\overline{U_3}$, we obtain
	\begin{align*}
	\Pr(\overline{U_2})%&=\Pr\left(\sum_{i\in [\lceil (1+\varepsilon_j)\rho/(1-\varepsilon_j)\rceil+1]\setminus\{k\}} \hspace{-20pt}X^{(j)}_i < \left(\frac12\right)^{j+1}(1+\varepsilon_j)\right)\\ 
	&\leq \Pr(\sum_i
	%_{i\in [\lceil (1+\varepsilon_j)\rho/(1-\varepsilon_j)\rceil+1]\setminus\{k\}} \hspace{-20pt}
	X_i^{(j)} < \mu_X(1-\varepsilon_j))
	\leq \exp\Bigl(-\frac{\varepsilon_j^2}{2}\mu_X\Bigr) \leq \exp\Bigl(-\frac{12\cdot 2^j\ln \rho}{2\rho}\cdot \frac{1+\varepsilon_j}{1-\varepsilon_j}\cdot \frac{\rho}{2^{j+1}}\Bigr)\leq 
	%\exp\Bigl(-3\ln \rho\Bigr)
	\frac{1}{\rho},\\
%\intertext{and}
\Pr(\overline{U_3})%&=\Pr\left(\sum_{i\in [\lceil (1+\varepsilon_j)\rho/(1-\varepsilon_j)\rceil+1]\setminus\{k\}}\hspace{-20pt}Y_i \geq \left(\frac12\right)^{j+1}\rho \right)\\
	&\leq \Pr(\sum_i
	%_{i\in [\lceil (1+\varepsilon_j)\rho/(1-\varepsilon_j)\rceil+1]\setminus\{k\}}\hspace{-20pt}
	Y_i^{(j)} \geq \mu_Y(1+\varepsilon_j)) \leq \exp\Bigl(-\frac{\varepsilon_j^2}{3}\mu_Y\Bigr)  \leq \exp\Bigl(-\frac{12\cdot 2^j\ln \rho}{3\rho}\cdot \frac{1+\varepsilon_j}{1-\varepsilon_j} \cdot \frac{(1-4\varepsilon_j)\rho}{2^{j+1}}\Bigl) \leq
	% \exp\Bigl(-\log \rho\Bigr)=
	\frac{1}{\rho}.
	\end{align*}

%The previous bounds only work for $j$ such that $1/\rho<\varepsilon_j<1/8$, i.e., for $j$ between 0 and some value $j^*=\Theta(\log \rho)$ with $2^{j^*} = K\rho/\ln \rho$, for some constant $K>0$. Therefore,
Putting all together, it follows that 
\begin{align*}
\Pr(r^k\not\in \ALG) &= \sum_{j>j^*}\Pr(r^k\not\in \ALG | t^k\in J_j)\frac{1}{2^{j+1}} +\sum_{j=0}^{j^*}\Pr(r^k\not\in \ALG | t^k\in J_j)\frac{1}{2^{j+1}}\\
&\leq \sum_{j>j^*}\frac1{2^{j+1}} +\sum_{j\geq 0}\left(\frac3\rho  + 4\varepsilon_j\right) \frac1{2^{j+1}} = \frac1{2^{j^*+1}} + \biggl( \frac3\rho + 4\sqrt{\frac{12\ln \rho}{\rho}}\sum_{j\geq 0}2^{j/2} \frac1{2^{j+1}} \biggr),
\end{align*}
which is $O\left(\sqrt{\frac{1}{\rho}}+\frac{1}{\rho}
+\sqrt{\frac{\log{\rho}}{\rho}}\right)=O\left(\sqrt{\frac{\log{\rho}}{\rho}}\right)$. Therefore, the algorithm is  $1/(1 - O(\sqrt{\log \rho/\rho}))= (1+O(\sqrt{\log \rho/\rho}))$ probability-competitive.
\end{proof}

\section{Algorithms for the ordinal MSP on general matroids} \label{sec:general}

\subsection{$O(1)$ intersection-competitiveness: Proof of Theorem \ref{thm:intersection}}

Our algorithm for the intersection notion works as follows. We first sample $s$ elements. After that, we select an element as long as it is part of the optimum of the subset of elements seen so far, and it preserves the independence of the current solution. Recall that we denote by $R_i=\{r_1,r_2,\ldots,r_i\}$ the first $i$ elements seen by the algorithm. 

\begin{algorithm}[H]
	\caption{Improving Greedy
		\label{alg:cardinal}}
	\begin{algorithmic}[1]
		\Require{Matroid $\calM(R,I)$ in random order $r_1,r_2,\dots, r_n$, and a fixed integer $s\in \{1,\ldots n\}$.}
		\State $\ALG\gets \emptyset$
		\For{$i=s+1$ to $n$}
		\If{$r_i \in \OPT(R_i)$ \text{ and } $\ALG+r_i\in \calI$}
		\State $\ALG\gets \ALG+r_i$.	
		\EndIf \EndFor
		\State Return $\ALG$.
	\end{algorithmic}
\end{algorithm}

The idea of only considering elements that belong to the current optimum is not new. 
Ma et al.~\cite{MTW2016} consider an algorithm that after sampling a fraction, selects an element as long as it belongs to the current offline optimum, and they prove this to be 9.6 utility-competitive (and the analysis holds for probability as well) for laminar matroids.
The same algorithm was suggested by Babaioff et al.~\cite{BIK2007}, and they showed how this algorithm fails to be even constant utility-competitive. 
In contrast, we show this algorithm to be $O(1)$ intersection-competitive.
\begin{lem}\label{lem:harmonic} 
	Let $B=\{r_i:r_i\in \OPT(R_i),\; i\in \{s+1,\ldots,n\}\}$.
	%Let $B$ be the random set of elements in $\{e_{s+1}, \dots, e_n\} 
	%$ that satisfy the condition $\{e_i\in \OPT(E_i)\}$. 
	Then $\E[|B|]=(H_n - H_s)\rho$, where $H_j$ denotes the $j$-th harmonic number.
\end{lem}
\begin{proof}
	%	Suppose for a second that $s$ is deterministic and fix $i$. 
	%	Consider $i\in \{s+1,\ldots,n\}$.
	For every $i\in [n]$ if we condition the set $R_i$ to be some fixed set $F\in \binom{R}{i}$, we have 
	$$\Pr(r_i\in \OPT(R_i) | R_i = F)= \Pr(r_i \in \OPT(F)) \leq \frac{\rho}{i}.$$
	Therefore, by linearity of expectation we conclude that
	\begin{align*}
	\E[|B|] &= \sum_{i=s+1}^n \Pr(r_i \in \OPT(R_i)) \leq \sum_{i=s+1}^n \frac{\rho}{i} = (H_n - H_s)\rho.\qedhere
	\end{align*}
\end{proof}

%\begin{thm}
%	By choosing $s=n/2$, Algorithm \ref{alg:cardinal} is $1/\ln(e/2)$ intersection-competitive.
%\end{thm}
\begin{proof}[Proof of Theorem \ref{thm:intersection}]
	We study the competitiveness of the algorithm when the sample size is $s$, and then we optimize over this value to conclude the theorem. 
	For any random ordering, $|\ALG|=\rho(R\setminus R_s)\geq \rho(\OPT\setminus R_s)=|\OPT\setminus R_s|$. Then, $	\E[|\ALG\cap \OPT|]$ equals
	\begin{align*}
\E[|\ALG|]+\E[|\OPT\setminus R_s|]-\E[|\ALG\cup (\OPT\setminus R_s)|] \ge 2\E[|\OPT\setminus R_s|]-\E[|\ALG\cup (\OPT\setminus R_s)|].
	\end{align*}
	Furthermore, $\ALG\cup (\OPT\setminus R_s)\subseteq B$, where $B$ is the set defined in Lemma \ref{lem:harmonic}. 
	Since the elements arrive in uniform random order, for $r\in \OPT$ we have that $\Pr(r\notin R_s)=(n-s)/n=1-s/n$. 
	Therefore, the right hand side of the previous inequality is at least
	\begin{align*}
	%	\E[|\ALG\cap \OPT|]&=\E[|\ALG\cap \OPT\setminus E_s|]= \E[|\ALG|]+\E[|\OPT\setminus E_s|]-\E[|\ALG\cup (\OPT\setminus E_s)|]\\
	%	&\geq 2\E[|\OPT\setminus E_s|]-2\E[|B|]
	%2\E[|\OPT\setminus R_s|]-\E[|\ALG\cup (\OPT\setminus R_s)|]	&\geq
	2\E[|\OPT\setminus R_s|]-\E[|B|] 	%&= \frac{2(n-s)}{n}\rho-(H_n-H_s)\rho\\
	%	&=r\left(2-\frac{2s}{n}-\sum_{i=s+1}^n \frac1i\right)\\
	&\geq \left(2-\frac{2s}{n}-\int_{s}^n\frac{1}{x}dx\right)\rho =\left(2-\frac{2s}{n} + \ln(s/n)\right)\rho.
	\end{align*}
	This quantity is maximized in $s=n/2$. 
	So, by assuming $n$ even (which can be done by adding an extra dummy element if $n$ is odd), and setting the algorithm for $s=n/2$, we obtain
	\begin{align*}
	\E[|\ALG\cap \OPT|]&\geq \rho(1 - \ln(2))= |\OPT| / \ln(e/2).\qedhere
	\end{align*} 
	%	meaning that the algorithm is $\frac{1}{1-\ln(2)}\sim 3.258891$ intersection-competitive.
	%	
\end{proof}

%\subsection{$O(\log\log r)$-ordinal and $O(\log r)$-probability competitiveness:\\ Proof of Theorem \ref{thm:ord-prob}}
\subsection{Ordinal/Probability-competitiveness: Proof of Theorem \ref{thm:ord-prob}}
\label{subsec:reduction-layered-ordinal}
We introduce a variant of the MSP that helps us to leverage existing algorithms for the utility version of the MSP, in order to get competitive algorithms for the ordinal and probability variants. 
We need a concept similar to the {\it aided sample-based MSP} introduced by Feldman et al .\cite{FSZ2015}. 

In the {\it Layered-MSP} the input is a tuple  $(\mathcal{M}, F, C, \succ)$ where $\calM=(R,\calI,\succ)$ is a totally ordered matroid, $C=\{c_1,c_2,\ldots,c_k\}$ is a finite set with $C\cap R=\emptyset$, $\succ$ is a total order over $R\cup C$ with $c_1\succ c_2\succ \cdots \succ c_k$, and  $F$ is a random subset of $R$ in which every element is present with probability $1/2$. The set $C$ defines a partition of $R$ in the following way. 
We call $C_0=\{r\in R:r\succ c_1\}$ the {\it highest layer} and $C_k=\{r\in R:c_k\succ r\}$ the {\it lowest layer}. 
For $j\in \{1,\ldots,k-1\}$, the {\it $j$-th layer} is the set $C_j=\{r\in R:c_j\succ r\succ c_{j+1}\}$. 
By construction, the layers $\{C_0,C_1,\ldots,C_k\}$ induced by $C$ form a partition of $R$.
In the following we call a tuple $(\mathcal{M},F,C, \succ)$ as a {\it layered matroid}.

An algorithm for the Layered-MSP first sees $F$ but is unable to select any of its elements. The rest of the elements of $\mathcal{M}$ arrive in uniform random order. At time step $t$, the algorithm can get the full value order in $R_t$ by using only ordinal information, it can use an independence oracle to test any subset of $R_t$ and it can also check membership of any element to each layer induced by $C$. 

\begin{defi}\label{def:level-alg-competitive}
	We say that an algorithm for the Layered-MSP is $\alpha$-competitive if it returns an independent set $\ALG\in \calI$, and for each $ j \in \{0,1,\ldots,|C|\},\, \E|\ALG\cap C_j| \geq \frac{1}{\alpha} |\OPT\cap C_j|,$ where the expectation is taken both over the distribution of $F$ and the internal algorithm randomness.
\end{defi}

\begin{thm}[Feldman et al., {\cite[Corollary 4.1]{FSZ2015}}]\label{thm:feldman}
	There exists an $8\lceil \log (|C|+1)+1\rceil$-competitive algorithm for the Layered-MSP.
\end{thm} 

For the sake of completeness we include the mentioned algorithm for Layered-MSP of Feldman et al. We remark that the algorithm was introduced in the context of the utility version. Nevertheless, the result above follows in the absence of weights. 

%An algorithm for the Layered-MSP does not know (a priori) the number of elements of the matroid, nor its exact rank, but is provided an upper bound $\tilde \rho\ge \rho$ on the matroid rank.  
%The main result of Feldman et al. \cite{FSZ2015} can now be stated in terms of this framework.
\begin{algorithm}[H]
	\caption{(Feldman et al. \cite{FSZ2015}) for Layered-MSP}
	\label{alg:FSZ}
	\begin{algorithmic}[1]
		\Require{A layered matroid $(\mathcal{M}, F,C,\succ)$}
		% h es el numero de weight classes en FSZ
		\State Let $\tau$ be a uniformly at random number from $\{0,1,\ldots,\lceil \log(|C|+1)\rceil\}$.
		\State Let $\Delta$ be a uniformly at random number from $\{0,1,\ldots,2^{\tau}-1\}$.
		\State Let $\mathbf{B}=\{B_1,B_2,\ldots,B_{\lceil (\Delta+|C|)/2^{\tau}\rceil}\}$ where $ B_i=\bigcup_{j=\max\{0,2^{\tau}(i-1)-\Delta+1\}}^{\min\{|C|,2^{\tau} i-\Delta\}}C_j.$
		\State With probability $1/2$, set $H=\text{odd}(|C|)$ or $\text{even}(|C|)$ otherwise.
		\State For each $i\in H$ let $T_i\leftarrow \emptyset$.
		\For{each element in $r\in R\setminus F$}
		\State Let $i$ be such that $r\in B_i$.
		\If{ $i\in H$ and $r\in N_i$ and $T_i+r\in \mathcal{I}_i$}
		\State $T_i\leftarrow T_i+r$.
		\EndIf	
		\EndFor
		\State Return $\ALG=\cup_{i\in H}T_i$.
	\end{algorithmic}
\end{algorithm}

\noindent {\it About the algorithm of Feldman et al.} We denote by $\text{odd}(k)$ and $\text{even}(k)$ the odd and even numbers, respectively, in the set $\{0,1,\ldots,k\}$. The set $\calI_i$ is the independent sets family of a matroid $M_i=(N_i,\calI_i)$ defined as follows. 	
Let $B_{\ge i}=\cup_{j\in \{i,\ldots,\lceil (\Delta+|C|)/2^{\tau}\rceil\}}B_j$. 
The matroid $M_1$ is obtained from $M$ by contracting
\footnote{The {\it contraction} of  $\calM=(R,\calI)$ by $Q$, $\calM/Q$, has ground set $R-Q$ and a set $I$ is independent if $\rho(I\cup Q)-\rho(Q)=|I|$.} $F\cap B_{\ge 2}$ and then restricting\footnote{The {\it restriction} of $\calM=(R,\calI)$ to $Q$, $\calM|_Q$, has ground set $Q$ and a set $I$ is independent if $I\in \calI$ and $I\subseteq Q$.} to $B_1$.
For $i>1$, $M_i$ is obtained from $M$ by contracting $F\cap B_{\ge i+1}$ and then restricting it to $B_i\cap \text{span}(F\cap B_{\ge i-1})$.

\subsubsection{Reduction from Layered-MSP to ordinal and probability variants of the MSP}

The main result of this section corresponds to the lemma below. Theorem \ref{thm:ord-prob} follows directly using this lemma, and the rest of this section is devoted to prove the lemma.

\begin{lem}\label{lem:level-reduction}
	Suppose there exists a $g(|C|)$-competitive algorithm for the Layered-MSP, where $g$ is a non-decreasing function. Then, 
	\begin{enumerate}	
		\item[(i)]\label{thm:level-to-ordinal} there exists an $O(g(1+\log \rho))$ ordinal-competitive algorithm for the MSP, and 
		\item[(ii)]\label{thm:level-to-probability} there exists an $O(g(1+\rho))$ probability-competitive algorithm for the MSP. 
	\end{enumerate}	
	%	If we are given a $g(k)$ level-competitive algorithm, where $f$ is a nondecreasing function in $k$, we can use it to construct 
	%	\begin{itemize}
	%		\item An $O(g(1+\log r))$ ordinal-competitive algorithm for $\mathcal{M}$.
	%		\item An $O(g(1+r))$ probability-competitive algorithm for $\mathcal{M}$. 
	%	\end{itemize}
\end{lem}

\begin{proof}[Proof of Theorem  \ref{thm:ord-prob}]
	We observe that $g(x)=8\lceil \log (x+1)+1\rceil$ is non-decreasing, so we can apply Lemma \ref{lem:level-reduction} using Algorithm \ref{alg:FSZ} and Theorem \ref{thm:feldman}.
\end{proof}

%\subsection{From level-MSP to ordinal and probability variants: Proof of Lemma \ref{lem:level-reduction}}

In the following, let $\mathcal{A}^{\text{layer}}$ be a $g(|C|)$-competitive algorithm for the layered-MSP. Our 
%ordinal-competitive 
algorithm for Lemma \ref{lem:level-reduction}  (i), depicted as Algorithm \ref{alg:reductioncardinal}, first gets a sample from $R$ with expected size $n/2$, and constructs a partition $C$ using the optimum of the sample. 
By sampling a set $F$ over the remaining elements it feeds $\mathcal{A}^{\text{layer}}$ with a layered matroid.  

\begin{algorithm}[H]
	\caption{$O(g(1+\log \rho))$ ordinal-competitive algorithm
		\label{alg:reductioncardinal}}
	\begin{algorithmic}[1]
		\Require{Matroid $\calM(R,\calI,\succ)$ in random order $r_1,r_2,\dots, r_n$.}
		\State Let $s\sim \Bin(n,1/2)$ and compute $\OPT(R_s)=\{s(1),\dots,s(\ell)\}$, where $s(1)\succ s(2)\succ \cdots \succ s(\ell)$.
		%			\State Define the partitioning $\mathcal{L}=\{L_0,L_1,\ldots,L_k\}$ of $E\setminus S$, with $k= \lfloor\log_2 \ell \rfloor +1$ given by:
		%			\Statex $L_0=\{e\in E\setminus S\colon e\succ s(1)\}$,
		%			\Statex $L_i=\{e\in E\setminus S\colon s(2^{i-1})\succ e \succ s(2^i)\}$, for $1\leq i \leq k-1$, and
		%			\Statex $L_k=\{e\in E\setminus S\colon s(2^{k-1})\succ e\}$.
		\State Let $C=\{s(1),s(2),s(4),\ldots, s(2^{k-1})\}$, where $k=\lfloor\log \ell \rfloor +1$.
		\State Let $t\sim \Bin(n-s,1/2)$ and let $F=\{r_{s+1},\ldots,r_{s+t}\}$ be the next $t$ elements from $R\setminus R_s$. 
		\State Return $\ALG=\mathcal{A}^{\text{layer}}(\calM|_{R\setminus R_s},F,C,\succ)$.
	\end{algorithmic}
\end{algorithm}

\begin{proof}[Proof of Lemma \ref{lem:level-reduction} (i)]	
	Let $\OPT=\{f(1),\dots,f(\rho)\}$ be such that $f(1)\succ f(2)\succ \cdots \succ f(\rho)$. 
	Consider the function $T:\mathbb{N}\to \mathbb{N}$ given by $T(0)=0$, $T(i)=|\{r\in R\colon r\succ f(i)\}|$ if $i\in [1,\rho]$ and $T(i)=n$ if $i>	 \rho$. 
	%	\[T(i)=\begin{cases}0 & \text{ if }i=0,\\ |\{e\in E\colon e\succ f(i)\}| & \text{ if }1\le i\le r,\\ n & \text{ if }i\ge r. \end{cases}\]
	In particular, for $i\in [1,\rho]$, $T(i)$ is the number of elements in $\mathcal{M}$ that are at least as high as the $i$-th element of $\OPT$, so $f(i)=r^{T(i)}$.
	Observe that for all $k$ such that $T(i)\leq k < T(i+1)$ we have $|\OPT\cap R^k|=i$.  
	In the following we study the expected number of elements from $R^k$ that the algorithm selects, so we can conclude using Lemma \ref{lem:ordinal-equivalence}. 
	If $T(0)\leq k<T(1)$ it holds $\E[|\ALG\cap R^k|]=0=|\OPT\cap R^k|$, so this case is done.
	
	Let $k\in \mathbb{N}$ be such that $T(1)\leq k< T(8)$. 
	Let $f$ be the highest non-loop element in the value order with $f(1)\succ f$.
	%	\jscom{El caso donde $f$ no existe molesta, pero no debería ser problema omitirlo} 	
	%Either $f$ does not exist, or $f=f(2)$, or $f$ is parallel to $f(1)$.
	With probability at least 1/4, it holds that $f(1)\in R\setminus R_s$ and $f\in R_s$. 
	In this case, $f\in \OPT(R_s)$, since $f(1)=s(1)$, and the only non-loop element of $C_0$ is $f(1)$. 
	Thus, $C_0=\{f(1)\}=\OPT(R\setminus R_s)\cap C_0$. 
	Therefore,
	\[\Pr(C_0=\{f(1)\})\ge \Pr(f(1)\in R\setminus R_s,f\in R_s)\ge 1/4.\]
	Since $g$ is non-decreasing and $|C|\le 1+ \log \ell\le 1+ \log \rho$, Algorithm $\mathcal{A}^{\text{layer}}$ is $g(1+\log \rho)$-competitive. Furthermore, the event $C_0=\{f(1)\}$ depends only on steps 1 and 2 of the algorithm before executing $\mathcal{A}^{\text{layer}}$. It follows that 
	\begin{align*}
	\Pr(f(1)\in \ALG) &\geq \frac{1}{4}\Pr(f(1)\in \ALG| C_0=\{f(1)\})=\frac{1}{4}\E[|\ALG\cap C_0|\;|\;C_0 =\{f(1)\}]\\
	&\geq \frac{1}{4g(1+\log \rho)}\E[|\OPT(R\setminus R_s)\cap C_0|\; |\; C_0=\{f(1)\}]= \frac1{4g(1+\log \rho)}.
	\end{align*}
	Since $|\OPT\cap R^{k}|\leq 8$, we conclude in this case that
	\begin{align*}
	\E[|\ALG\cap R^k|]&\geq \Pr(f(1)\in \ALG)\geq \frac{1}{4g(1+\log \rho)} \geq \frac{1}{32g(1+\log \rho)}|\OPT\cap R^{k}|.
	\end{align*}
	Let $j\geq 3$ and $k$ be such that $T(2^j)\leq k< T(2^{j+1})$. We denote $q=2^{j-3}$. 
	Let $A_j$ be the event where $|\{f(1),\ldots,f(2q)\}\cap R\setminus R_s|\geq q$ and $B_j$ is the event where $|\{f(2q+1),\ldots,f(6q)\}\cap R_s|\geq 2q$.
	We have that $\Pr(A_j\cap B_j)\geq 1/4$, since in our algorithm the probability for an element to be sampled equals  the probability of not being sampled.
	Observe that any subset of elements strictly better than $f(t)$ has rank at most $t-1$, and therefore $f(2q)\succeq s(2q)$. 	
	If $B_j$ holds, then $s(2^{j-2})=s(2q)\succeq f(6q)$. Since $f(6q)\succeq f(8q)=f(2^j)$, it follows that
	\[ \bigcup_{i=0}^{j-3}C_i\subseteq \{r\in R\setminus R_s:r\succeq f(2^j)\}= R^{T(2^j)} \cap R\setminus R_s.\]
	This implies that 
	%	\jscom{La unión está bien que parta de 1 o la hacemos partir de $i=0$?}
	\begin{align*}
	\E[|\ALG\cap R^k|]%&\geq \E[|\ALG\cap R^{T(2^j)}|] \\
	&\geq \frac14 \E\Bigl[|\ALG\cap R^{T(2^j)}|\, \Big\vert\, A_j\cap B_j \Bigr]\geq \frac14 \E\Bigl[\Bigl|\ALG \cap \bigcup_{i=0}^{j-3}C_i\Bigr| \,\, \Big\vert\, A_j\cap B_j \Bigr].
	\end{align*}
	Furthermore, if $A_j$ holds, then
	\begin{align*}
	\Bigl|\OPT(R\setminus R_s)\cap \bigcup_{i=0}^{j-3} C_i\Bigr|&=|\{f\in \OPT \cap R\setminus R_s \colon f\succeq s(2q)\}|\\
	&\geq |\{f\in \OPT \cap R \setminus R_s\colon f\succeq f(2q)\}|\geq q.
	\end{align*}		
	Since $g$ is non-decreasing and $|C|\le 1+\log \ell\le 1+\log \rho$, Algorithm $\mathcal{A}^{\text{layer}}$ is $g(1+\log \rho)$-competitive. Since events $A_j$ and $B_j$ depend only on the sampling at line 1 of the algorithm (before executing $\mathcal{A}^{\text{layer}}$) and by using linearity of the expectation and the observation above we have that
	\begin{align*}	
	\frac14 \E\Bigl[\Bigl|\ALG \cap \bigcup_{i=0}^{j-3}C_i\Bigr| \,\, \Big\vert\, A_j\cap B_j \Bigr]&\geq \frac1{4g(1+\log \rho)}\E\Bigl[\Bigl|\OPT(R\setminus R_s)  \cap \bigcup_{i=0}^{j-3}C_i\Bigr| \,\,\Big\vert\, A_j\cap B_j \Bigr]\\
	&\geq \frac1{4g(1+\log \rho)}q \geq \frac1{64g(1+\log \rho)}|\OPT\cap R^k|,
	\end{align*}
	where the last inequality holds since $|\OPT\cap R^k|\le 2^{j+1}=16q$. By using Lemma \ref{lem:ordinal-equivalence} we conclude that the algorithm is $O(g(1+\log \rho))$ ordinal-competitive.
\end{proof}

\begin{algorithm}[H]
	\caption{$O(g(1+\rho))$ probability-competitive algorithm
		\label{alg:reductionprobability}}
	\begin{algorithmic}[1]
		\Require{Matroid $\calM(R,\calI,\succ)$ in random order $r_1,r_2,\dots, r_n$.}
		\State Let $s\sim \Bin(n,1/2)$ and compute $\OPT(R_s)=\{s(1),\dots,s(\ell)\}$, where $s(1)\succ s(2)\succ \cdots \succ s(\ell)$. 
		%			Denote by $S^*=\OPT(R_s)$ the sample optimal base.
		%			\State Compute the optimum set $S^*=\{s(1),\dots,s(\ell)\}$ of the sample $E_s$.
		%			\State Define (implicitly) $S^+$ as the set of all $e\in E\setminus S$ such that $e\in \OPT(S^*+e)$. 
		%			Define the partitioning $\mathcal{L}=\{L_0,L_1,\ldots,L_{\ell +1}\}$ of $S^+$ given by:
		%			\State Define the level-partition $\mathcal{L}=(L_i)_{i\in \{0,1,\dots, k\}}$ of $S^+$, with $k= \ell +1$ given by:
		%			\Statex $L_0=\{e\in S^+\colon e\succ s(1)\},$
		%			\Statex $L_i=\{e\in S^+\colon s(i-1)\succ e \succ s(i)\}$, for $1\leq i \leq \ell$, and
		%			\Statex $L_{\ell+1}=\{e\in S^+\colon s(\ell)\succ e\}.$ 
		\State Let $t\sim \Bin(n-s,1/2)$ and let $F= \{r_{s+1},\ldots,r_{s+t}\}$ be the next $t$ elements from $R\setminus R_s$. 
		%			\State Compute $t\sim \Bin(n-s,1/2)$ and let $F'$ as the next $t$ elements from $E$. 
		\State Return $\ALG=\mathcal{A}^{\text{layer}}(\calM {|_{R_s^+}},F\cap R_s^+,\OPT(R_s),\succ)$, where $R_s^+=\{r\in R\setminus R_s:r\in \OPT(R_s+r)\}$.
	\end{algorithmic}
\end{algorithm}

To prove Lemma \ref{lem:level-reduction} (ii), consider Algorithm \ref{alg:reductionprobability} depicted above.
Before the analysis, it will be useful to consider the next process.
Let $(X_t)_{t\in \mathbb{N}}$ be a sequence of Bernoulli independent random variables such that $\Pr(X_t=0)=\Pr(X_t=1)=1/2$ for every $t\in \mathbb{N}$. We create two sets $V,W\subseteq R$ iteratively using the following procedure. 
%\jscom{Yo cambiaría $C$ por otra letra ya que ya la usamos para los threshold, además $C_h$ es confuso abajo.}
\begin{algorithm}[H]
	\caption{Coupling procedure}
	\label{alg:coupling}
	\begin{algorithmic}[1]
		\Require{Matroid $\calM(R,\calI,\succ)$.}
		\State Initialize $V\leftarrow \emptyset$, $W\leftarrow \emptyset$ and $\theta\leftarrow 0$.
		\For{$i=1$ to $n$}
		\If{$r^i\in \OPT(V+r^i)$}
		\State $\theta\leftarrow \theta+1$ and $Y_{\theta}=X_i$.
		\If{$Y_{\theta}=0$}
		\State $V\leftarrow V+r^i$
		\Else{ $W\leftarrow W+r^i$.}
		\EndIf
		\EndIf
		\EndFor
	\end{algorithmic}
\end{algorithm}

The value $\theta$ represents a counter on the elements that improve over the current set $V$. When the counter is updated, we say that the element considered on that iteration {\it is assigned coin $Y_\theta$.}	
\begin{lem}\label{lem:coupling}
	$(V,W)$ has the same distribution as $(\OPT(R_s),R_s^+)$ in Algorithm \ref{alg:reductionprobability}.
\end{lem}

\begin{proof}
	Let $(V_i,W_i)_{i=1}^n$ be the states of the coupling process at the end of each iteration.  
	Let $Z=\{r^i:X_i=0\}$. 
	Observe that $Z$ and the sample $R_s$ of Algorithm \ref{alg:reductionprobability} have the same distribution.
	Since the coupling procedure checks from highest to lowest element, it follows that $V_i= \OPT(Z\cap R^i)$  for every $i\in \{1,\ldots,n\}$, and therefore $V_n=\OPT(Z)$. 
	%	
	%	To show this is enough to check that $\OPT(Z)\subseteq V_n$, since the equality in this case holds by optimality and using the exchange property in matroids.
	%	Let $r^i\in \OPT(Z)$. By Lemma \ref{lem:opt-restricted}, we have $\OPT(Z)\cap (V_{i-1}+r^i)\subseteq \OPT(V_{i-1}+r^i)$ and therefore $r^i\in \OPT(V_{i-1}+r^i)$.
	%	If follows that $r^i\in V_i$, and therefore every element of $\OPT(Z)$ is in $V_n$.
	To conclude the lemma it suffices to check that $W_n=\{r\in R\setminus Z:r\in \OPT(Z+r)\}$. 
	In fact, it can be proven by induction that 
	\[W_i= \{r\in R^i\setminus Z:r\in \OPT(V_i+r)\}\]
	for every $i\in \{1,\ldots,n\}$, and so the lemma follows, since $R^n=R$ and $V_n=\OPT(Z)$.
	%	Let $H_i$ be the set on the right hand side above. 
	%	If $X_1=0$, then $V_1=\{r^1\}$ and $H_1=W_1=\emptyset$. 
	%	If $X_0=1$ then $V_1=\emptyset$ and $W_1=\{r^1\}=H_1$, since $r^1\in R^1\setminus Z$ and $r^1\in \OPT(\emptyset+r^1)$. 
	%	The only case we have to consider next is when $r^{i+1}\in \OPT(V_i+r^{i+1})$.
	%	In particular, if $X_{i+1}=1$ then $H_{i+1}=H_i\cup \{r^{i+1}\}=W_{i}\cup \{r^{i+1}\}=W_{i+1}$.
	%	Otherwise, $V_{i+1}=V_i\cup \{r^{i+1}\}$ and $W_{i+1}=W_i=H_i$. 
	%	In particular, $H_{i+1}\subseteq R^i\setminus Z$.
	%	By Lemma \ref{lem:opt-properties}, $\OPT(V_{i+1}+r)\cap (V_i+r)\subseteq \OPT(V_i+r)$, so $H_{i+1}\subseteq H_i$, and $\OPT(V_i+r)\subseteq \OPT(V_{i+1}+r)$ and then $H_i\subseteq H_{i+1}$. 
	%	That finishes the induction.
\end{proof}	

\begin{proof}[Proof of Lemma \ref{lem:level-reduction} (ii)]		
	
	Thanks to Lemma \ref{lem:coupling} we assume that $(\OPT(R_s),R_s^+)$ is generated by the process described in Algorithm \ref{alg:coupling}.
	In what follows, fix $f(j)=r^{T(j)}\in \OPT$. We know that at step
	$T(j)$ of the coupling procedure, $f(j)\in \OPT(R_s+f(j))$ no matter the trajectory of $(X_t)_{t\in \mathbb{N}}$. 
	Let $\theta$ be such that $r^{T(j)}$ is assigned coin $Y_{\theta}$, that is, $Y_{\theta}=X_{T(j)}$. 
	Then, with probability at least 1/8, the event $\mathcal{E}$ defined as $Y_{\theta-1}=Y_{\theta+1}=0$ and $Y_{\theta}=1$, holds. 
	If $\mathcal{E}$ happens, let $s(h)$ be the element of $\OPT(R_s)$ who was assigned coin $Y_{\theta-1}$. 
	In particular, the element $s(h+1)$ is assigned coin $Y_{\theta+1}$.
	Thus,
	\[C_{h}=\{r\in R_s^+\colon s(h)\succ r \succ s(h+1)\}=\{f(j)\}=\OPT(R_s^+)\cap C_{h}.\]
	Therefore by using that the occurrence of $\calE$ is decided before executing $\mathcal{A}^{\text{layer}}$ which is $g(1+\ell)\leq g(1+\rho)$ competitive, we get that
	\begin{align*}
	\Pr(f(j)\in \ALG) &\geq \frac{1}{8}\Pr(f(j)\in \ALG | \mathcal{E})\\
	&= \frac{1}{8}\E[|\ALG\cap C_{h}| | \mathcal{E}]\geq \frac{1}{8g(\rho+1)}\E[|\OPT(R_s^+)\cap C_{h}| | \mathcal{E}]= \frac1{8g(\rho+1)},
	\end{align*}
	and we conclude that Algorithm \ref{alg:reductionprobability} is $O(g(1+\rho))$ probability-competitive.
\end{proof}

\subsection{Comparison between ordinal measures}\label{sec:incomparable}
%\jscom{El nombre de la sección no me gusta mucho}
In this section we discuss about some incomparability results for the competitiveness measures previously introduced. We show that an algorithm that is utility-competitive is not necessarily competitive for the rest of the measures. In particular, we provide an instance where the $O(\log \rho)$ utility-competitive algorithm by Babaioff, Immorlica and Kleinberg \cite{BIK2007} has a poor competitiveness for the other three. Regarding the notions of the ordinal MSP, we show that the intersection and the ordinal measures are incomparable. More specifically, we show the existence of an algorithm and an instance where it is arbitrarily close to 1 intersection-competitive, but have an unbounded ordinal/probability-competitiveness. And on the other hand, we also show the existence of an algorithm and an instance where it is arbitrarily close to 1 ordinal-competitive, but have an unbounded intersection/probability-competitiveness. Recall that probability is the strongest in the sense that it implies competitiveness for all the other measures considered (see Lemma \ref{lem:variant-relations}).

In the utility variant, the weight $w(r)$ of an element is revealed to the algorithm when arrived. Suppose for simplicity that the rank $\rho$ of the matroid is known\footnote{This assumption can be removed by using standard arguments; we can set $\rho$ equal to twice the rank of the sampled part.} and that the weights the algorithm sees are all different. In the above mentioned algorithm, called by the authors the {\it Threshold Price Algorithm} (TPA), it is taken a sample $R_s$  of $s\sim \text{Bin}(n,1/2)$ elements\footnote{The original analysis uses half of $n$, but the analysis gets simpler if one uses $\text{Bin}(n,1/2)$ since one can assume that each element is in the sample with probability 1/2 independent of the rest.} and it records the top weight $w^*$ of a non-loop\footnote{A {\it loop} in a matroid is an element that belongs to no basis.} element seen in $R_s$. It chooses uniformly at random a number $\tau$ in the set $\{0, 1, 2, \dots, \lceil \log_2 \rho\rceil\}$, and then it selects greedily any non-sampled element whose weight is at least $T=w^*/2^{\tau}$. 

\begin{thm}[Babaioff et al. \cite{BIK2007}]
\label{thm:bik-2007}
The algorithm TPA is $O(\log \rho)$ utility-competitive.
\end{thm}
\subsubsection{TPA is $\Omega(\rho)$ intersection, ordinal and probability-competitive}
We show in this section that TPA is $\Omega(\rho)$-competitive in the intersection, ordinal and probability measures. We first revisit the proof by \cite{BIK2007} for the $O(\log \rho)$ utility-competitiveness of TPA. 
\begin{proof}[Proof of Theorem \ref{thm:bik-2007}]
Let $\OPT=\{f(1), \ldots, f(\rho)\}$ be the optimal base such that $w_1>w_2>\cdots>w_{\rho}$, where $w_i=w(f(i))$  for each $i\in \{1,\ldots,\rho\}$.
Suppose that $f(1)\notin R_s$. Then, if the second non-loop element of the matroid is sampled and if $\tau=0$, the element $f(1)$ will be selected in the second phase of TPA. Hence $\Pr(f(1)\in \ALG)\ge 1/4\cdot 1/(\lceil \log \rho\rceil+1)=\Omega(1/\log \rho)$.

Let $B=\{i\in \{2,\ldots,\rho\}:w_i\geq w_1/\rho\}$, and let $\mathcal{E}_i$ be the event where $f(1)\in R_s$ and $w_i/2< T\leq  w_i$. For each $i\in B$, we have $\log(w_1/w_i)<\log \rho$ and therefore $\Pr(\mathcal{E}_i)=1/2\cdot \Pr(\tau=\lceil \log(w_1/w_i)\rceil)=\Omega(1/\log \rho)$. The random order assumption implies that in expectation $(i-1)/2\geq i/4$ elements in $\{f(2),\dots,f(i)\}$ are non-sampled, hence the expected rank of the matroid restricted  to non-sampled elements of weight at least $w_i=w(f(i))$ is $\Omega(i)$. Therefore, given $i\in B$ and conditioned on $\mathcal{E}_i$, the algorithm selects $\Omega(i)$ elements of value at least $w_i/2$. 
%(this last statement, also holds for $i=1$) 
It follows that for each $i\in B$, 
\[\E[|\ALG\cap \{r\colon w(r)\geq w_i/2\}|]\geq\Pr(\mathcal{E}_i)\E[|\ALG\cap \{r\colon w(r)\geq w_i/2\}||\mathcal{E}_i]=\Omega\left(\frac{i}{\log \rho}\right).\]

In addition, observe that the elements in $\OPT\setminus \{f(i):i\in B\cup \{1\}\}$ have total weight less than $w_1(\rho-|B|-1)/\rho<w_1$, and therefore $2\sum_{i\in B\cup\{1\}}w_i\geq w(\OPT)$. Putting all together, we have that the expected weight of the output is
	\begin{align*}
	\E[w(\ALG)]&\geq \frac12 \sum_{i=1}^{\rho}(w_i-w_{i+1})\E[|\ALG\cap \{r\colon w(r)\geq w_i/2\}|]\\
	&= \Omega\left(\frac{1}{\log \rho}\right) \sum_{i\in B\cup\{1\}}(w_i-w_{i+1})\cdot i \\
	&= \Omega\left(\frac{1}{\log \rho}\right)\sum_{i\in B\cup\{1\}}w_i= \Omega\left(\frac{1}{\log \rho}\right) w(\OPT).\qedhere
	\end{align*}

\end{proof}
%\noindent {\it TPA is $\Omega(r)$ intersection, ordinal, and probability-competitive.}
We study in the following the competitiveness of TPA for the ordinal measures. 
Since probability is the strongest (Lemma \ref{lem:variant-relations}), it is enough to check that TPA is $\Omega(\rho)$ ordinal and intersection competitive. 
%In fact, we show that it holds for both. 
Let $R=\{r^1,\dots, r^n\}$ with $n=2\rho^3$.
Consider the laminar family $\{R^{n/2},R\}$, where $R^{n/2}=\{r^1,\dots, r^{n/2}\}$, and take the laminar matroid $(R,\calI)$ where a set $I\in \calI$ is independent if $|I\cap R^{n/2}|\leq 1$ and $|I|\le \rho$.
In particular, the matroid rank is $\rho$.  
Given $\varepsilon>0$, the weights are given by $w(r^i)=8-\varepsilon i$ for $i\in \{1,\ldots n/2\}$, and $w(r^i)=7-\varepsilon i$ for $i\in \{n/2+1,\ldots,n\}$. 

Observe that the optimal basis is given by $\OPT=\{r^1,r^{n/2+1},\dots,r^{n/2+\rho-1}\}$.
If we run TPA on this instance, with probability $p=1-2^{-n/2}$ the top weight $w^*$ in the sample is from $R^{n/2}$, and thus, $7<w^*<8$. 
Let $\mathcal{E}$ be this event.
No matter the value of $T$, the algorithm always select at most $\rho$ elements in the non-sampled part, and therefore 
\begin{align*}
\E[|\ALG\cap \OPT|]&=(1-2^{-n/2})\E[|\ALG\cap \OPT||\mathcal{E}]+2^{-n/2}\E[|\ALG\cap \OPT||\overline{\mathcal{E}}]\\
			      &\leq \E[|\ALG\cap \OPT||\mathcal{E}]+\rho\cdot 2^{-\rho^3}\leq \E[|\ALG\cap \OPT||\mathcal{E}]+1.
\end{align*}
Conditional on $\mathcal{E}$, we have that either $T>7$ or $T<4$. 
In the first case, TPA will select at most one element, which would come from $R^{n/2}$, and so $|\ALG\cap \OPT|\leq 1$. 
Otherwise, if $T<4$, the algorithm will select at most one element from $R^{n/2}$ and at most the first $\rho$ non-sampled elements from $R\setminus R^{n/2}$. 
The expected number of elements in $\{r^{n/2+1},\dots,r^{n/2+\rho-1}\}$ appearing in $\{r_{s+1},\ldots,r_{s+\rho}\}$, that is the first $\rho$ elements appearing after the sample, is $(\rho-1)\rho/n<1$. 
It follows that $\E[|\ALG\cap \OPT||\mathcal{E}]$ is upper bounded by $2$, and therefore $\E[|\ALG\cap \OPT|]\leq 3$. 
The discussion above implies as well that $\E[|\ALG\cap R^{n/2+\rho-1}|]\leq 3$, and then TPA is $\Omega(\rho)$ intersection and ordinal competitive. 
The conclusion for the ordinal case follows by using the characterization for ordinal competitiveness in Lemma \ref{lem:ordinal-equivalence}. 

\subsubsection{The ordinal and intersection measures are not comparable}

We show in the following that there is no competitiveness dominance between the intersection and the ordinal notions.  
%Let $m$ and $M$ be two positive integers such that $M\ge 1$. 
Let $m$ and $M$ be two positive integers.
Consider $R=\{r^1,r^2,\ldots, r^{(M+1)m}\}$ and let $\mathcal{L}$ be the partition of $R$ given by
\[\mathcal{L}=\bigcup_{i=1}^m\{\{r^i\}\}\cup \bigcup_{j=1}^M \{\{r^{jm+1},\ldots, r^{(j+1)m}\}\}.\]
In particular, 
%$\mathcal{L}$ is a laminar family and we consider the laminar matroid %
we consider the partition matroid 
$\mathcal{M}_{m,M}=(R,\calI)$ where $I\in \calI$ if $|I\cap L|\le 1$ for every $L\in \mathcal{L}$. 
The matroid rank is $m+M$. 
The algorithm we choose is greedy: we initialize $\ALG\leftarrow \emptyset$, and when an element $r$ arrives it is selected if $\ALG+r\in \calI$.

\begin{lem}\label{lem:ordinal-vs-inter-1}
Suppose $M\ge m^2$. Then, the greedy algorithm over instance $\mathcal{M}_{m,M}$ is $(1+1/m)$ ordinal-competitive and $\Omega(m)$ intersection-competitive.
\end{lem}
In fact, the ordinal competitiveness holds no matter what the value of $M$ is. We adjust the value of $M$ in order to get a poor competitiveness for the intersection notion.  
\begin{proof}[Proof of Lemma \ref{lem:ordinal-vs-inter-1}]
	The optimal base of $\mathcal{M}_{m,M}$ is the highest element of each part in $\mathcal{L}$, that is,
	\[\OPT=\{r^1,r^2,\ldots,r^m\}\cup \{r^{jm+1}:j\in \{1,\ldots,M\}\}.\] 
	In particular, we have 
	\[|\OPT\cap R^k|=\begin{cases}k & \text{ if }k\in \{1,\ldots,m\},\\ m+j & \text{ if }k\in \{jm+1,\ldots,(j+1)m\} \text{ and }j\in \{1,\ldots,M\}.\end{cases}\] 
	Observe that an element $r\in Q$, with $Q\in \mathcal{L}$, is selected by the algorithm if and only if $r$ is the first element of $Q$ that arrives. 
	Therefore, $\Pr(r^i\in \ALG)=1$ if $i\in \{1,\ldots,m\}$ and $\Pr(r^i\in \ALG)=1/m$ if $i\in \{m+1,\ldots,(m+1)m\}$. 
	It follows that $\E[|\ALG\cap R^k|]=k$ if $k\in \{1,\ldots, m\}$. 
	If $k>m$, then 
	\begin{align*}
	\E[|\ALG\cap R^k|]&=\sum_{i=1}^m\Pr(r^i\in \ALG)+\sum_{i=m+1}^k\Pr(r^i\in \ALG)=m+\frac{1}{m}(k-m).
	\end{align*}
	Thus, when $k\in \{1,\ldots,m\}$, we have $|\OPT\cap R^k|/\E[|\ALG\cap R^k|]=1$.
	Suppose $k=jm+r$ with $j\ge 1$ and $0\leq r\leq m$. Then, 
	\begin{align*}
	\frac{|\OPT\cap R^k|}{\E[|\ALG\cap R^k|]}&=\frac{m+j}{m+\frac{1}{m}(k-m)}\\
								      &=1+\frac{m-r}{m^2+mj+r-m}\leq 1+\frac{1}{m+j-1}\leq \frac{1}{m},
	\end{align*}
	and thus the greedy algorithm is $(1+1/m)$ ordinal-competitive for $\mathcal{M}_{m,M}$. 
	In the first inequality we used that $\phi(x)=(m-x)/(m^2+mj+x-m)$ is a decreasing function in the interval $[0,m]$.
	Observe that the competitiveness result holds no matter the value of $M$.  
	It remains to study the intersection competitiveness. 
	By the observations above, we have that 
	\[\frac{|\OPT|}{\E[|\ALG\cap \OPT|]}=\frac{m+M}{m+\frac{1}{m}\cdot M}\geq \frac{m+m^2}{m+m}\geq \frac{m}{2},\]
	and so the algorithm is at least $\Omega(m)$ intersection-competitive.
\end{proof}
Although not mentioned explicitly in the proof, since $\Pr(r^i\in \ALG)=1/m$ for $i>m$ it follows that the algorithm is $m$ probability competitive for $\mathcal{M}_{m,M}$, and for every $M$. 
In the following we construct an instance for which the intersection competitiveness is close to 1, but the ordinal competitiveness is poor. 
Let $m$ be a positive integer and $R=\{r^1,\ldots,r^{2m-1}\}$.
Consider the partition $\mathcal{L}$ given by
\[\mathcal{L}=\left\{\{r^1,\dots, r^m\}, \{r^{m+1}\}, \{r^{m+2}\},\dots, \{r^{2m-1}\}\right\}.\]
Let $\mathcal{N}_m=(R,\calI)$ be the 
%laminar matroid where %
partition matroid where $I\in \calI$ if $|I\cap L|\le 1$ for every $L\in \mathcal{L}$. 
The matroid rank is $m$.

\begin{lem}
The greedy algorithm over instance $\mathcal{N}_{m}$ is $(1+1/m)$ intersection-competitive and $\Omega(m)$ ordinal-competitive.
\end{lem}
\begin{proof}
	An element $r\in Q$, with $Q\in \mathcal{L}$, is selected by the algorithm if and only if $r$ is the first element of $Q$ that arrives. 
	Therefore, $\Pr(r^i\in \ALG)=1/m$ if $i\in \{1,\ldots,m\}$ and $\Pr(r^i\in \ALG)=1$ if $i\in \{m+1,\ldots,2m-1\}$. 
	Since $\OPT=\{r^1,r^{m+1},\ldots,r^{2m-1}\}$,
	\[\frac{|\OPT\cap R^1|}{\E[|\ALG\cap R^1|]}=\frac{1}{\Pr(r^1\in \ALG)}=1/(1/m)=m,\]
	hence the algorithm is $\Omega(m)$ ordinal-competitive. Finally, we have
	\begin{align*}
	\frac{|\OPT|}{\E|\ALG\cap\OPT|}=\frac{m}{1/m+m-1}\leq 1+\frac{1}{m},
	\end{align*}
	and we conclude that the greedy algorithm is $(1+1/m)$ intersection-competitive over $\mathcal{N}_m$.
\end{proof}

\bibliography{bibsoda}{}
\bibliographystyle{abbrv}
\end{document}

%% file: kleinberg-dibujo.tex
\begin{tikzpicture}[scale=7]
%\begin{axis}[axis y line=none, y=0.5cm/3, restrict y to domain=0:1, axis lines=left]
%%\begin{axis}[
%%%    title = {Number line},
%%    axis y line=none,
%%    y=0.4cm/3,
%%    restrict y to domain=0:1,
%%    axis lines=left,
%%    enlarge x limits=upper,
%%%    scatter/classes={
%%%        o={mark=*,fill=white}
%%%    },
%%%    scatter,
%%%    scatter src=explicit symbolic,
%%    every axis plot post/.style={mark=*,thick},
%%%    legend style={
%%%        draw=none,
%%%        at={(1,1)},
%%%        anchor=south east
%%%    },
%%%    legend image post style={mark=none}
%%]
%\addplot+ coordinates  {( 0, 0 )( 0.5, 0 )}; 
%\addplot+ coordinates  {( 0.5,0 )( 1,0 )}; 
%
%
%%\addplot table [y expr=0,meta index=1, header=false] {
%%0 c
%%1 c
%%%10 o
%%%
%%%13 c
%%%15 c
%%%
%%%17 c
%%};
%%%\addlegendentry{Set 1}
%%\addplot table [y expr=0,meta index=1, header=false] {
%%0.5 c
%%1 c
%%%21 c
%%};
%%%\addlegendentry{Set 2}
%\end{axis}

\draw[thick] (0,0) -- (1,0);
\foreach \x/\xtext in {0/0,0.12/$\cdots$,0.25/0.25,0.5/0.5,1/1}
%\foreach \x/\xtext in {0/0,0.2/,0.4/$1-b$,0.8/$1-a$,1}
    \draw[thick] (\x,0.5pt) -- (\x,-0.5pt) node[below] {\xtext};
%\draw (0.2,0.5pt) node[above] {$c$};
\draw[[-, ultra thick, blue] (0,0) -- (1,0);
%\draw[-), ultra thick, blue] (0.19,0) -- (0.2,0);
%\fill[opacity = 0.2, blue,rounded corners=1ex] (0,-.16ex) -- (0.2, -.16ex) -- (0.2, .16ex) -- (0,.16ex) -- cycle;
\draw (-0.05,0) node {$I_0$};

\draw[thick] (0,2.5pt) -- (1,2.5pt);
%\foreach \x/\xtext in {0/0,0.12/$\cdots$,0.25/0.25,0.5/0.5,1/1}
%%\foreach \x/\xtext in {0/0,0.2/,0.4/$1-b$,0.8/$1-a$,1}
%    \draw[thick] (\x,2.5pt) -- (\x,0.5pt) node[below] {\xtext};
%%\draw (0.2,0.5pt) node[above] {$c$};
\draw[[-, ultra thick, blue] (0,2.5pt) -- (0.5,2.5pt);
\draw[[-, ultra thick, red, dashed] (0.5,2.5pt) -- (1,2.5pt);
%\draw[-), ultra thick, blue] (0.19,0) -- (0.2,0);
%\fill[opacity = 0.2, blue,rounded corners=1ex] (0,-.16ex) -- (0.2, -.16ex) -- (0.2, .16ex) -- (0,.16ex) -- cycle;
\draw (-0.05,2.5pt) node {$I_1$};
\draw (1.05,2.5pt) node {$J_0$};

\draw[thick] (0,5pt) -- (0.5,5pt);
%\foreach \x/\xtext in {0/0,0.12/$\cdots$,0.25/0.25,0.5/0.5,1/1}
%%\foreach \x/\xtext in {0/0,0.2/,0.4/$1-b$,0.8/$1-a$,1}
%    \draw[thick] (\x,2.5pt) -- (\x,0.5pt) node[below] {\xtext};
%%\draw (0.2,0.5pt) node[above] {$c$};
\draw[[-, ultra thick, blue] (0,5pt) -- (0.25,5pt);
\draw[[-, ultra thick, red, dashed] (0.25,5pt) -- (0.5,5pt);
%\draw[-), ultra thick, blue] (0.19,0) -- (0.2,0);
%\fill[opacity = 0.2, blue,rounded corners=1ex] (0,-.16ex) -- (0.2, -.16ex) -- (0.2, .16ex) -- (0,.16ex) -- cycle;
\draw (-0.05,5pt) node {$I_2$};
\draw (0.55,5pt) node {$J_1$};

\draw[thick] (0,7.5pt) -- (0.25,7.5pt);
%\foreach \x/\xtext in {0/0,0.12/$\cdots$,0.25/0.25,0.5/0.5,1/1}
%%\foreach \x/\xtext in {0/0,0.2/,0.4/$1-b$,0.8/$1-a$,1}
%    \draw[thick] (\x,2.5pt) -- (\x,0.5pt) node[below] {\xtext};
%%\draw (0.2,0.5pt) node[above] {$c$};
\draw[[-, ultra thick, blue] (0,7.5pt) -- (0.125,7.5pt);
\draw[[-, ultra thick, red, dashed] (0.125,7.5pt) -- (0.25,7.5pt);
%\draw[-), ultra thick, blue] (0.19,0) -- (0.2,0);
%\fill[opacity = 0.2, blue,rounded corners=1ex] (0,-.16ex) -- (0.2, -.16ex) -- (0.2, .16ex) -- (0,.16ex) -- cycle;
\draw (-0.05,7.5pt) node {$I_3$};
\draw (0.29,7.5pt) node {$J_2$};

\draw (0.125,10.5pt) node {$\vdots$};

\end{tikzpicture}